\documentclass[aps,pra,reprint,superscriptaddress,amsmath,amssymb,nofootinbib,longbibliography,floatfix,nobalancelastpage]{revtex4-2}
\usepackage{amssymb}
\usepackage{soul}
\usepackage[T1]{fontenc}
\usepackage{amsthm,mathtools,mathrsfs,booktabs,microtype,xcolor,xurl,needspace}
\usepackage{tikz}
\usepackage[most]{tcolorbox}
\usepackage{algpseudocode}
\algrenewcommand\algorithmicrequire{\textbf{Input:}}
\algrenewcommand\algorithmicensure{\textbf{Output:}}
\algrenewcommand\alglinenumber[1]{{\scriptsize\color{black!50}#1}}
\algrenewcommand{\algorithmiccomment}[1]{\hfill{\small\color{black!55}$\triangleright$ #1}}
\usetikzlibrary{arrows.meta,calc,positioning,shapes.geometric,patterns}
\tikzset{
  pfnode/.style={draw=black!55,rounded corners=1pt,text=black,
    align=center,font=\footnotesize,inner xsep=4pt,inner ysep=3pt},
  pfgoal/.style={pfnode,draw=black!80,line width=.6pt},
  pfedge/.style={-{Latex},draw=black!75,line width=.45pt},
  pfline/.style={draw=black!75,line width=.45pt},
  pflab/.style={font=\scriptsize,text=black,align=center,inner sep=1.5pt}}
\usepackage[colorlinks=true,linkcolor=black,citecolor=blue!45!black,urlcolor=blue!45!black,hypertexnames=false]{hyperref}
\usepackage{bookmark}
\hypersetup{pdfpagemode=UseOutlines,bookmarksopen=true}
\bookmarksetup{open,openlevel=2,depth=2}
\makeatletter
\def\toclevel@section{0}
\def\toclevel@subsection{1}
\AddToHook{cmd/appendix/before}{\def\p@subsection{\thesection.}}
\makeatother
\ifdefined\XeTeXversion
\AddToHook{shipout/firstpage}{}
\fi
\theoremstyle{plain}
\newtheorem{theorem}{Theorem}
\newtheorem{proposition}[theorem]{Proposition}

\newtheorem{lemma}[theorem]{Lemma}
\newtheorem{conjecture}{Conjecture}
\newtheorem*{conjecture*}{Conjecture}
\theoremstyle{definition}

\newtheorem{example}[theorem]{Example}
\newtheorem{remark}[theorem]{Remark}
\theoremstyle{plain}
\makeatletter
\protected\def\rep@title{}
\newtheorem*{rep@theorem}{\rep@title}

\makeatother
\AddToHook{cmd/subsection/before}{\Needspace{10\baselineskip}}
\makeatletter\AddToHook{env/itemize/before}{\if@inlabel\leavevmode\fi}\makeatother
\AtBeginEnvironment{theorem}{\Needspace{8\baselineskip}}
\AtBeginEnvironment{lemma}{\Needspace{5\baselineskip}}
\makeatletter
\newenvironment{inlinefigure}{\par\medskip\noindent\begin{minipage}{\columnwidth}\def\@captype{figure}\centering}{\end{minipage}\par\medskip}
\makeatother
\makeatletter
\AddToHook{env/finitefact/before}{\if@inlabel\leavevmode\fi}
\newcounter{certfact}
\renewcommand{\thecertfact}{C\arabic{certfact}}
\NewTColorBox{finitefact}{m m m}{enhanced,breakable,width=\dimexpr\hsize-\@totalleftmargin\relax,
  code={\refstepcounter{certfact}},before={\par\medskip\label{#1}},after={\par\medskip},
  colback=white,colframe=black!45,boxrule=.4pt,arc=1pt,
  left=2mm,right=2mm,top=1mm,bottom=1.2mm,
  colbacktitle=black!6,coltitle=black,fonttitle=\small,
  colbacklower=black!2,segmentation style={black!35,dashed},
  title={\textbf{(\thecertfact)}\enspace\textbf{#2}}}
\makeatother
\newcommand{\certrepo}{\url{https://github.com/numberingeometry/clifford-hierarchy-plane-certificate}}
\newcommand{\certlink}[1]{\nolinkurl{certificates/#1} in \certrepo}
\newsavebox{\fittablebox}
\newcommand{\FitTable}[1]{\sbox{\fittablebox}{#1}\ifdim\wd\fittablebox>\linewidth\resizebox{\linewidth}{!}{\usebox{\fittablebox}}\else\usebox{\fittablebox}\fi}
\newcommand{\SC}{\mathsf{SC}}
\newcommand{\GSC}{\mathsf{GSC}}
\newcommand{\C}{\mathcal{C}}
\newcommand{\F}{\mathbb F}
\newcommand{\Z}{\mathbb Z}
\newcommand{\Q}{\mathbb Q}
\newcommand{\CC}{\mathbb C}
\newcommand{\T}{\mathbb T}
\newcommand{\A}{\mathcal A}
\newcommand{\Pau}{\mathcal P}

\newcommand{\ket}[1]{|#1\rangle}
\newcommand{\bra}[1]{\langle#1|}

\newcommand{\CCZ}{\mathrm{CCZ}}
\newcommand{\CZ}{\mathrm{CZ}}
\newcommand{\CNOT}{\mathrm{CNOT}}
\newcommand{\step}[2]{\textbf{#1}\ \emph{#2}\par\nopagebreak}
\newcommand{\tool}[1]{\hyperref[tool:#1]{(#1)}}
\DeclareMathOperator{\diag}{diag}
\DeclareMathOperator{\rank}{rank}
\DeclareMathOperator{\Span}{span}
\DeclareMathOperator{\Sp}{Sp}

\DeclareMathOperator{\im}{im}

\DeclareMathOperator{\tr}{tr}

\allowdisplaybreaks[1]
\hypersetup{pdftitle={When is the Clifford hierarchy generalized semi-Clifford?},pdfauthor={Rongbiao Thomas Wang and Bobby Zixuan Zhang},bookmarksdepth=2}
\begin{document}
\bookmarksetup{depth=-1}
\raggedbottom
\title{When is the Clifford hierarchy generalized semi-Clifford?}
\author{Rongbiao Thomas Wang}
\email{rbwang@uchicago.edu}
\affiliation{Computational and Applied Mathematics, University of Chicago, Chicago, Illinois 60637, USA}
\author{Bobby Zixuan Zhang}
\email{bobby.zhang.math@duke.edu}
\affiliation{Department of Mathematics, Duke University, Durham, North Carolina 27708, USA}
\begin{abstract}
We study semi-Clifford and generalized semi-Clifford structures in the Clifford hierarchy as a function of hierarchy level $k$ and qubit number $n$. We show that $\mathcal{C}_k(n)$ is semi-Clifford if and only if $k\le2$, $n\le2$, or $k=3$ and $n\le6$. We also establish that $\mathcal{C}_k(n)$ is generalized semi-Clifford when $k\le4$ or $n\le3$ and construct counterexamples when $k,n \geq 5$. For $n=4$ and $k\ge5$, we show that every gate defined over $\mathbb{Q}(i)$ up to a scalar is generalized semi-Clifford, reduce the general case to a single level, and propose a strategy to prove our conjecture that $\mathcal{C}_k(n)$ is generalized semi-Clifford if and only if $k\le4$ or $n\le4$. Lastly, we go beyond this structure by introducing $n_c$-generalized semi-Clifford gates, which we show are exactly the gates of the form $\sum_x|\pi(x)\rangle\langle x|\otimes V_x$ up to Clifford gates on both sides, where $\pi$ permutes the computational basis states of $n_c$ control qubits. We conjecture that for $k\ge3$ every gate in $\mathcal{C}_k(n)$ has this form for some $n_c\ge1$ with every $V_x\in\mathcal{C}_{k-3}(n-n_c)$, refining the conjecture of Zeng, Chen, and Chuang.
\end{abstract}
\maketitle

\section{Introduction}\label{sec:intro}

Let $\mathcal{P}(n)$ denote the Pauli group with $n$ qubits. With $\mathcal{C}_1(n)=\mathcal{P}(n)$, the Clifford hierarchy is recursively defined as
\[
\mathcal{C}_{k+1}(n)
=
\left\{
U : U\mathcal{P}(n) U^\dagger\subseteq \mathcal{C}_k(n)
\right\}, \quad \mathcal{C}(n) = \bigcup_{k=1}^\infty \mathcal{C}_k(n).
\]
The hierarchy was introduced by Gottesman and Chuang~\cite{GC1999} in the context of fault-tolerant quantum computation, where gates in the hierarchy can be implemented using teleportation.

However, the structure of the Clifford hierarchy remains obscure. In contrast to the Pauli group $\mathcal{C}_1(n)$ and the Clifford group $\mathcal{C}_2(n)$, $\mathcal{C}_k(n)$ is not a group for $k\geq 3$. To better understand their structure, Zeng, Chen, and Chuang~\cite{sc-ZCC2008} introduced the notions of \emph{semi-Clifford (SC)} and \emph{generalized semi-Clifford (GSC)} operators: A gate $U$ is semi-Clifford if conjugation by $U$ maps a maximal abelian subgroup of $\mathcal{P}(n)$ to a maximal abelian subgroup of $\mathcal{P}(n)$; $U$ is generalized semi-Clifford if it maps the linear span of some maximal abelian subgroup of $\mathcal{P}(n)$, under conjugation, to the linear span of another maximal abelian subgroup of $\mathcal{P}(n)$. Equivalently~\cite{sc-ZCC2008,sc-BS2010}, $U$ is semi-Clifford if $U= C_1 DC_2$ and generalized semi-Clifford if $U = C_1 \Pi D C_2$, where $C_1, C_2 \in \mathcal{C}_2(n)$, $D$ is diagonal, and $\Pi$ is a permutation in the computational basis. We denote the set of semi-Clifford gates by $\mathsf{SC}(n)$ and the set of generalized semi-Clifford gates by $\mathsf{GSC}(n)$.

In this work, we investigate when $\mathcal{C}_k(n)$ is semi-Clifford and generalized semi-Clifford as $k$ and $n$ vary. More precisely, we determine the complete picture of when $\mathcal{C}_k(n)$ is semi-Clifford (Theorem~\ref{thm:sc-c3-small}) and an almost complete picture of when it is generalized semi-Clifford (Theorems~\ref{thm:main-k4}, \ref{thm:main-small-width}, and~\ref{thm:gsc-counterexamples}), with the exception of $n=4$. In the latter case, we show that every gate defined over $\Q(i)$ up to a scalar is generalized semi-Clifford (Theorem~\ref{thm:n4-gaussian}), reduce the general case to a single level (Proposition~\ref{prop:n4-finite-cutoffs}), and propose a strategy to prove that $\C_k(4)\subseteq\GSC(4)$ for every $k$ (Conjecture~\ref{conj:gsc-boundary}). Going beyond the generalized semi-Clifford structure, we introduce \emph{$n_c$-generalized semi-Clifford} gates, which up to Clifford gates on both sides permute the computational basis states of $n_c$ control qubits and apply a gate $V_x$ to the remaining qubits, and we conjecture that every gate in the Clifford hierarchy has this form for some $n_c\ge1$ with each $V_x$ three levels lower (Conjecture~\ref{conj:fibers}). The results are visualized in Fig.~\ref{fig:combined}, where this article's contributions are circled. We summarize the proven results in the following theorem.

\begin{figure}[t]
\centering
\begin{tikzpicture}[x=.74cm,y=.6cm,font=\footnotesize]
\foreach \k in {1,...,7}{\foreach \n in {1,...,7}{
\pgfmathtruncatemacro{\issc}{(\k<=2)||(\n<=2)||((\k==3)&&(\n<=6))}
\pgfmathtruncatemacro{\isgsc}{(\k<=4)||(\n<=3)}
\ifnum\issc=1
\fill[black!20](\k-.5,\n-.5)rectangle(\k+.5,\n+.5);
\node at(\k,\n){$+$};
\else\ifnum\isgsc=1
\fill[black!7](\k-.5,\n-.5)rectangle(\k+.5,\n+.5);
\node at(\k,\n){$+$};
\else\ifnum\n=4
\fill[black!3](\k-.5,\n-.5)rectangle(\k+.5,\n+.5);
\begin{scope}
\clip(\k-.5,\n-.5)rectangle(\k+.5,\n+.5);
\foreach \t in {-.75,-.5,-.25,0,.25,.5}{\draw[black!55,line width=.35pt](\k+\t,\n-.5)--(\k+\t+.45,\n+.5);}
\end{scope}
\node at(\k,\n){$\star$};
\else\node at(\k,\n){$-$};\fi\fi\fi}}
\foreach \j in {0,...,7}{
\draw[black!20,very thin](.5+\j,.5)--(.5+\j,7.5);
\draw[black!20,very thin](.5,.5+\j)--(7.5,.5+\j);}
\begin{scope}
\clip(.5,.5)rectangle(8.2,8.2);
\foreach \k in {1,...,8}{\foreach \n in {1,...,8}{
\pgfmathtruncatemacro{\isedge}{(\k==8)||(\n==8)}
\ifnum\isedge=1
\pgfmathtruncatemacro{\issc}{(\k<=2)||(\n<=2)||((\k==3)&&(\n<=6))}
\pgfmathtruncatemacro{\isgsc}{(\k<=4)||(\n<=3)}
\ifnum\issc=1
\fill[black!11](\k-.5,\n-.5)rectangle(\k+.5,\n+.5);
\node[black!40]at(\k,\n){$+$};
\else\ifnum\isgsc=1
\fill[black!4](\k-.5,\n-.5)rectangle(\k+.5,\n+.5);
\node[black!40]at(\k,\n){$+$};
\else\ifnum\n=4
\fill[black!2](\k-.5,\n-.5)rectangle(\k+.5,\n+.5);
\begin{scope}
\clip(\k-.5,\n-.5)rectangle(\k+.5,\n+.5);
\foreach \t in {-.75,-.5,-.25,0,.25,.5}{\draw[black!27,line width=.3pt](\k+\t,\n-.5)--(\k+\t+.45,\n+.5);}
\end{scope}
\node[black!40]at(\k,\n){$\star$};
\else\node[black!40]at(\k,\n){$-$};\fi\fi\fi\fi}}
\foreach \j in {0,...,7}{
\draw[black!10,very thin](.5+\j,7.5)--(.5+\j,8.5);
\draw[black!10,very thin](7.5,.5+\j)--(8.5,.5+\j);}
\draw[line width=.65pt,black!40](2.5,7.5)--(2.5,8.5);
\draw[line width=.65pt,black!40](7.5,2.5)--(8.5,2.5);
\draw[line width=.65pt,black!40](4.5,7.5)--(4.5,8.5);
\draw[line width=.65pt,black!40](7.5,3.5)--(8.5,3.5);
\draw[line width=.7pt,black!40](8,3)ellipse(.19 and .24);
\draw[line width=.7pt,black!40](8,4)ellipse(.19 and .24);
\draw[line width=.7pt,black!40](4,8)ellipse(.19 and .24);
\foreach \j in {5,6,7}{\draw[line width=.7pt,black!40](8,\j)ellipse(.19 and .24);\draw[line width=.7pt,black!40](\j,8)ellipse(.19 and .24);}
\draw[line width=.7pt,black!40](8,8)ellipse(.19 and .24);
\end{scope}
\draw[-{Stealth[length=1.5mm]}](.2,.2)--(8.5,.2)node[right]{$k$};
\draw[-{Stealth[length=1.5mm]}](.2,.2)--(.2,8.6)node[above]{$n$};
\foreach \j in {1,...,6}{\node[below]at(\j,.18){$\j$};\node[left]at(.18,\j){$\j$};}
\node[below]at(7,.18){$\ge7$};\node[left]at(.18,7){$\ge7$};
\draw[line width=.65pt](2.5,7.5)--(2.5,6.5)--(3.5,6.5)--(3.5,2.5)--(7.5,2.5);
\draw[line width=.65pt](4.5,7.5)--(4.5,3.5)--(7.5,3.5);
\foreach \n in {5,6}{\draw[line width=.7pt](3,\n)ellipse(.19 and .24);}
\foreach \n in {3,...,7}{\draw[line width=.7pt](4,\n)ellipse(.19 and .24);}
\foreach \k in {5,6,7}{\draw[line width=.7pt](\k,3)ellipse(.19 and .24);}
\foreach \k in {5,6,7}{\draw[line width=.8pt](\k,4)ellipse(.19 and .24);}
\foreach \k in {5,6,7}{\foreach \n in {5,6,7}{\draw[line width=.7pt](\k,\n)ellipse(.19 and .24);}}
\draw[{Stealth[length=1.5mm]}-](8.3,1.5)--(8.6,1.5)node[right,align=left,font=\scriptsize,inner xsep=2pt]{semi-Clifford};
\draw[{Stealth[length=1.5mm]}-](8.3,3)--(8.6,3)node[right,align=left,font=\scriptsize,inner xsep=2pt]{generalized\\semi-Clifford};
\draw[{Stealth[length=1.5mm]}-](8.3,6.5)--(8.6,6.5)node[right,align=left,font=\scriptsize,inner xsep=2pt]{not\\generalized\\semi-Clifford};
\end{tikzpicture}
\caption{At position $(k,n)$, a $+$ sign represents that $\mathcal{C}_k(n)$ belongs to the indicated class and a $-$ sign represents that $\mathcal{C}_k(n)$ is not generalized semi-Clifford. For the cases of $n=4$ and $k \geq 5$, the $\star$ sign represents a partial answer where gates defined over $\Q(i)$ up to scalar are generalized semi-Clifford.}
\label{fig:combined}
\end{figure}
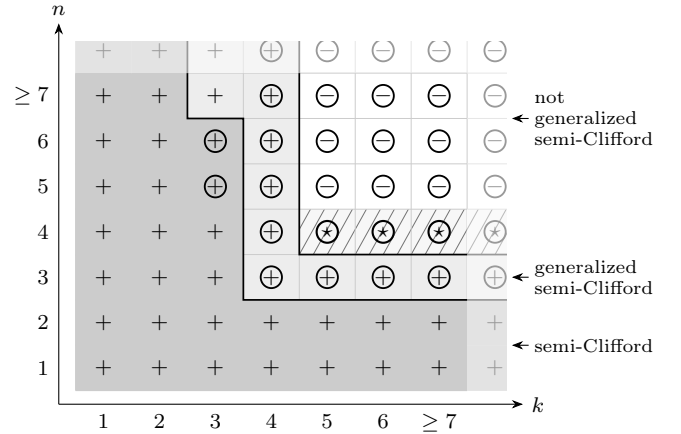

\begin{theorem}\label{thm:sc-boundary}
For positive integers $k,n$,
\begin{enumerate}
    \item $\displaystyle \mathcal C_k(n)\subseteq\SC(n) \quad\Longleftrightarrow\quad
\begin{cases}
k\leq2,\quad\text{or}\\
n\leq2,\quad\text{or}\\
k=3\ \text{and}\ n\leq6,
\end{cases}$
    \item $\displaystyle k\leq4\quad\text{or}\quad n\leq3 \quad\Longrightarrow\quad \C_k(n)\subseteq\GSC(n)$.
    \item If $U\in \C_k(4)$ and up to a scalar $U$ can be defined over $\Q(i)$, then $U\in \GSC (4)$.
    \item $\displaystyle k\geq 5\quad\text{and}\quad n\geq 5 \quad\Longrightarrow\quad \C_k(n)\not \subseteq\GSC(n)$.
\end{enumerate}
\end{theorem}

This article continues a line of work investigating the relationship between the Clifford hierarchy and semi-Clifford and generalized semi-Clifford gates, initiated in~\cite{sc-ZCC2008}. There, the authors showed that $\mathcal{C}_k(n) \not\subseteq \mathsf{SC}(n)$ in general and conjectured that $\mathcal{C}_k(n)\subseteq\mathsf{GSC}(n)$ for all $k$ and $n$. Subsequent works~\cite{sc-BS2010,pllahaEtAl2020,Anderson_2024,sc-AC2025,sc-HRT2025arxiv} further investigated the semi-Clifford and generalized semi-Clifford structure of gates in the Clifford hierarchy. The conjecture was recently disproved in~\cite{deSilvaLautsch2026False} by a counterexample in $\mathcal{C}_5(5)\setminus\mathsf{GSC}(5)$. The inclusion $\mathcal{C}_4(n) \subseteq \mathsf{GSC}(n)$, which we prove in Section~\ref{sec:gsc}, was recently obtained independently in~\cite{level4} by similar methods. We discuss the contributions of these related works in more detail in Sections~\ref{sec:SC} and~\ref{sec:gsc}.

\section{Symplectic structures in the Clifford hierarchy}\label{sec:symp}

A recurring theme throughout the article is exploiting symplectic structures in the Clifford hierarchy, inspired by~\cite{sc-ZCC2008,sc-BS2010}. Consider the vector space $\F_2^{2n}$ over the field $\F_2$. For $v = (x,z)$ and $v' = (x',z')$ where $x, z, x', z' \in \F_2^n$, consider the symplectic form
\[
\omega(v,v') = x \cdot z' + x' \cdot z \mod 2.
\]
This vector space is of particular interest since every Hermitian Pauli gate is $\pm P_v$ for some $v$, where
\[
P_v \coloneqq i^{x \cdot z} X^x Z^z = i^{x \cdot z} \bigotimes_{j=1}^n X^{x_j} Z^{z_j}.
\]
The quotient map
\[
q:\mathcal{P}(n) \longrightarrow \mathcal{P}(n)/\langle iI\rangle\cong\mathbb F_2^{2n},\qquad \lambda P_v\longmapsto v,
\]
is a surjective group homomorphism. It is straightforward to check that for all $v, w \in \F_2^{2n}$, $P_v P_w = P_w P_v$ if and only if $\omega(v,w) = 0$.

Recall that a subspace $L \subseteq \F_2^{2n}$ is isotropic if $\omega(v,w)=0$ for all $v,w\in L$. It is Lagrangian if it is an isotropic subspace with maximal dimension, i.e., $\dim_{\mathbb F_2}L=n$. Let
\[
\mathcal A_L \coloneqq \operatorname{span}_{\mathbb{C}}\{P_v:v\in L\}.
\]
We write $[U]$ for the class of a unitary in $\mathrm{PU}(2^n)=\mathrm U(2^n)/\{\lambda I:|\lambda|=1\}$, and $A\doteq B$ for $A=\lambda B$ with $\lambda\in\CC^\times$. Since global phases cancel in conjugation, the conditions $U \in \mathsf{SC}(n)$, $U \in \mathsf{GSC}(n)$, and $U \in \C_k(n)$ are determined by $[U]$.

The following lemmas characterize the generalized semi-Clifford property in terms of Lagrangian subspaces.

\begin{lemma}\label{lem:lagrangian}
A gate $U$ is generalized semi-Clifford if and only if there exist Lagrangian subspaces $L$, $M \subseteq \F_2^{2n}$ such that $U \mathcal{A}_L U^\dagger = \mathcal{A}_M$.
\end{lemma}
We call a pair $(L,M)$ of Lagrangians a \emph{generalized semi-Clifford Lagrangian pair} of $U$ if $U\A_LU^\dagger=\A_M$. By Lemma~\ref{lem:lagrangian}, $U$ is generalized semi-Clifford exactly when it has such a pair.

On the other hand, we can also characterize the generalized semi-Clifford property of $U$ using \emph{corrections} $h_v \coloneqq UP_vU^\dagger$ for all $v\in \F_2^{2n}$. Notice that if $U\in \C_k(n)$, then $h_v\in \C_{k-1}(n)$. The following simple fact is the motivation for our induction proof in Theorem~\ref{thm:main-small-width}.
\begin{lemma}\label{lem:main-local-global}
A gate $U$ is generalized semi-Clifford if and only if there is a Lagrangian $M\le\F_2^{2n}$ such that
\[
h_v\A_Mh_v^\dagger=\A_M\qquad \text{for every $v\in \F_2^{2n}$}.
\]
For such $M$, $U^\dagger\A_MU=\A_L$ for a Lagrangian $L$.
\end{lemma}

Furthermore, it is useful to consider how the semi-Clifford and generalized semi-Clifford properties change after basic operations like taking tensor products.
\begin{lemma}\label{lem:basic-operations}
Let $U\in\C_k(n)$ and $V\in\C_\ell(m)$. Write $\mathsf X(n)$ for either $\SC(n)$ or $\GSC(n)$ and $\T = \{z \in \CC: \lvert z \rvert = 1\}$.
\begin{enumerate}
\item \label{itm:op1} For $C_1,C_2\in\C_2(n)$ and $\lambda\in\T$,
\[
U\in\mathsf X(n)\quad\Longleftrightarrow\quad
\lambda C_1UC_2\in\mathsf X(n).
\]
If $k\geq2$, then $\lambda C_1UC_2\in\C_k(n)$.
\item \label{itm:op2} $U\otimes V\in\C_{\max\{k,\ell\}}(n+m)$. If $U\in\mathsf X(n)$ and $V\in\mathsf X(m)$, then $U\otimes V\in\mathsf X(n+m)$. If $V$ is diagonal, then
\[
U\otimes V\in\mathsf X(n+m)
\quad\Longleftrightarrow\quad U\in\mathsf X(n).
\]
\end{enumerate}
In particular, if $\C_k(n)\not\subseteq\mathsf X(n)$, then $\C_r(N)\not\subseteq\mathsf X(N)$ for every $r\geq k$ and $N\geq n$.
\end{lemma}

The proofs are given in Appendix~\ref{app:prerequisites}.

\section{Semi-Clifford gates}\label{sec:SC}

Most cases of whether $(k,n)$ satisfies $\C_k(n) \subseteq \SC(n)$ have been established in the literature. Zeng, Chen, and Chuang~\cite{sc-ZCC2008} proved that every gate is semi-Clifford at every level on at most two qubits and at the third level on three qubits, and constructed a gate in $\C_4(3)\setminus\SC(3)$. Beigi and Shor~\cite{sc-BS2010} reported a seven-qubit third-level counterexample by Gottesman and Mochon and proved that every third-level gate is generalized semi-Clifford. Anderson and Connelly established $\C_3(4)\subseteq\SC(4)$~\cite{sc-AC2025}.

It thus remains to check $\C_3(5)$ and $\C_3(6)$, which were left open in~\cite{sc-AC2025}. For permutation gates, these two cases follow from~\cite[Appendix~A]{sc-HRT2025arxiv}. We settle these two cases below to complete the picture as shown in Fig.~\ref{fig:sc-plane}.

\begin{inlinefigure}
\centering
\begin{tikzpicture}[x=.85cm,y=.6cm,font=\footnotesize]
\foreach \k in {1,...,7}{\foreach \n in {1,...,7}{
\pgfmathtruncatemacro{\positive}{(\k<=2)||(\n<=2)||((\k==3)&&(\n<=6))}
\ifnum\positive=1
\fill[black!7](\k-.5,\n-.5)rectangle(\k+.5,\n+.5);
\node at(\k,\n){$+$};
\else\node at(\k,\n){$-$};\fi}}
\foreach \j in {0,...,7}{
\draw[black!20,very thin](.5+\j,.5)--(.5+\j,7.5);
\draw[black!20,very thin](.5,.5+\j)--(7.5,.5+\j);}
\begin{scope}
\clip(.5,.5)rectangle(8.2,8.2);
\foreach \k in {1,...,8}{\foreach \n in {1,...,8}{
\pgfmathtruncatemacro{\edge}{(\k==8)||(\n==8)}
\ifnum\edge=1
\pgfmathtruncatemacro{\positive}{(\k<=2)||(\n<=2)||((\k==3)&&(\n<=6))}
\ifnum\positive=1
\fill[black!4](\k-.5,\n-.5)rectangle(\k+.5,\n+.5);
\node[black!40]at(\k,\n){$+$};
\else\node[black!40]at(\k,\n){$-$};\fi\fi}}
\foreach \j in {0,...,7}{
\draw[black!10,very thin](.5+\j,7.5)--(.5+\j,8.5);
\draw[black!10,very thin](7.5,.5+\j)--(8.5,.5+\j);}
\draw[line width=.65pt,black!40](2.5,7.5)--(2.5,8.5);
\draw[line width=.65pt,black!40](7.5,2.5)--(8.5,2.5);
\end{scope}
\draw[-{Stealth[length=1.5mm]}](.2,.2)--(8.5,.2)node[right]{$k$};
\draw[-{Stealth[length=1.5mm]}](.2,.2)--(.2,8.6)node[above]{$n$};
\foreach \j in {1,...,6}{\node[below]at(\j,.18){$\j$};\node[left]at(.18,\j){$\j$};}
\node[below]at(7,.18){$\ge7$};\node[left]at(.18,7){$\ge7$};
\foreach \n in {5,6}{\draw[line width=.7pt](3,\n)ellipse(.19 and .24);}
\draw[line width=.65pt](2.5,7.5)--(2.5,6.5)--(3.5,6.5)--(3.5,2.5)--(7.5,2.5);
\end{tikzpicture}
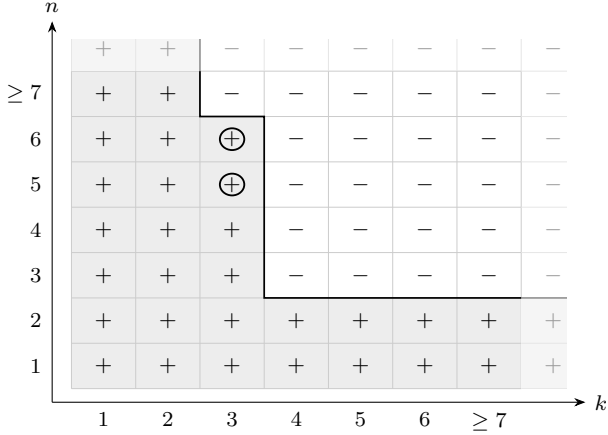
\caption{$\mathcal{C}_k(n) \subseteq \mathsf{SC}(n)$ if and only if $(k,n)$ lies in a region marked $+$.}
\label{fig:sc-plane}
\end{inlinefigure}

Our proof uses two structural results from the previous literature. By~\cite[Theorem~1.1]{sc-BS2010} and~\cite[Proposition~2]{sc-ZCC2008}, for $U\in\C_3(n)$, there are $C_{1},C_{2}\in\C_2(n)$, a diagonal gate $D$, and a computational-basis permutation $\Pi$ such that
\[
C_{1}UC_{2}=D\Pi.
\]
Furthermore, for $n\leq6$, by~\cite[Lemma~3.8, Appendix~A, and Corollary~B.6]{sc-HRT2025arxiv}, there exist permutations $C_{3},C_{4} \in \mathcal{C}_2(n)$ with integers $r,m$ satisfying $n=r+m$ such that
\[
\begin{aligned}
C_3\Pi C_4\ket{x,y}
  &=\ket{x,y+q(x)},\\
\text{where } q(x)&=\sum_{i<j}x_ix_jq_{ij},
  \qquad q_{ij}\in\F_2^m.
\end{aligned}
\]
We now sketch the proof of the result whose full proof is presented in Appendix~\ref{app:sc-boundary}.

\begin{theorem}\label{thm:sc-c3-small}
For $n =5,6$, $\mathcal{C}_3(n)\subseteq \mathsf{SC}(n)$.
\end{theorem}
\begin{proof}[Sketch of Proof]
By the two reductions above, we reduce the problem to investigating $U\in\mathcal C_3(n)$ of the form
\[
U\ket{x,y}=e^{i\theta(x,y)}\ket{x,y+q(x)},
\]
where $q$ is a quadratic function over $\mathbb F_2$. The key observation is that $\mathcal C_3(n)$ imposes strong restrictions on the phase $\theta(x,y)$. We show that, after a further change of basis by a target Clifford, we may reduce it to a gate of the form
\[
M=\sum_x\zeta(x)\ket{x}\bra{x}\otimes P_{v(x)},
\]
where $P_{v(x)}$ is a Pauli operator acting on the $y$-qubits. Thus it remains to check that the $P_{v(x)}$'s mutually commute, which holds for $n\leq6$ by a dimension argument. Since $U=C_1MC_2$ for $C_1,C_2\in\mathcal C_2(n)$, this proves that $U\in\mathsf{SC}(n)$.
\end{proof}

\section{Generalized semi-Clifford gates}
\label{sec:gsc}

In~\cite{sc-ZCC2008}, the authors established that $\mathcal{C}_k(n) \subseteq \mathsf{GSC}(n)$ for all $k$ and $n \leq 2$. Subsequently,~\cite{sc-BS2010} established that $\mathcal{C}_3(n) \subseteq \mathsf{GSC}(n)$ for all $n$. In this section and the subsequent one, we determine the rest of the cases with the exception of $n=4$ and $k\geq 5$, for which we provide some arithmetic characterizations and a heuristic strategy. The results are illustrated in Fig.~\ref{fig:gsc-plane}.

\begin{inlinefigure}
\centering
\begin{tikzpicture}[x=.85cm,y=.6cm,font=\footnotesize]
\foreach \k in {1,...,7}{\foreach \n in {1,...,7}{
\pgfmathtruncatemacro{\positive}{(\k<=4)||(\n<=3)}
\ifnum\positive=1
\fill[black!7](\k-.5,\n-.5)rectangle(\k+.5,\n+.5);
\node at(\k,\n){$+$};
\else\ifnum\n=4
\fill[black!3](\k-.5,\n-.5)rectangle(\k+.5,\n+.5);
\begin{scope}
\clip(\k-.5,\n-.5)rectangle(\k+.5,\n+.5);
\foreach \t in {-.75,-.5,-.25,0,.25,.5}{\draw[black!55,line width=.35pt](\k+\t,\n-.5)--(\k+\t+.45,\n+.5);}
\end{scope}
\node at(\k,\n){$\star$};
\else\node at(\k,\n){$-$};\fi\fi}}
\foreach \j in {0,...,7}{
\draw[black!20,very thin](.5+\j,.5)--(.5+\j,7.5);
\draw[black!20,very thin](.5,.5+\j)--(7.5,.5+\j);}
\begin{scope}
\clip(.5,.5)rectangle(8.2,8.2);
\foreach \k in {1,...,8}{\foreach \n in {1,...,8}{
\pgfmathtruncatemacro{\edge}{(\k==8)||(\n==8)}
\ifnum\edge=1
\pgfmathtruncatemacro{\positive}{(\k<=4)||(\n<=3)}
\ifnum\positive=1
\fill[black!4](\k-.5,\n-.5)rectangle(\k+.5,\n+.5);
\node[black!40]at(\k,\n){$+$};
\else\ifnum\n=4
\fill[black!2](\k-.5,\n-.5)rectangle(\k+.5,\n+.5);
\begin{scope}
\clip(\k-.5,\n-.5)rectangle(\k+.5,\n+.5);
\foreach \t in {-.75,-.5,-.25,0,.25,.5}{\draw[black!27,line width=.3pt](\k+\t,\n-.5)--(\k+\t+.45,\n+.5);}
\end{scope}
\node[black!40]at(\k,\n){$\star$};
\else\node[black!40]at(\k,\n){$-$};\fi\fi\fi}}
\foreach \j in {0,...,7}{
\draw[black!10,very thin](.5+\j,7.5)--(.5+\j,8.5);
\draw[black!10,very thin](7.5,.5+\j)--(8.5,.5+\j);}
\draw[line width=.65pt,black!40](4.5,7.5)--(4.5,8.5);
\draw[line width=.65pt,black!40](7.5,3.5)--(8.5,3.5);
\draw[line width=.7pt,black!40](8,3)ellipse(.19 and .24);
\draw[line width=.7pt,black!40](8,4)ellipse(.19 and .24);
\draw[line width=.7pt,black!40](4,8)ellipse(.19 and .24);
\foreach \j in {5,6,7}{\draw[line width=.7pt,black!40](8,\j)ellipse(.19 and .24);\draw[line width=.7pt,black!40](\j,8)ellipse(.19 and .24);}
\draw[line width=.7pt,black!40](8,8)ellipse(.19 and .24);
\end{scope}
\draw[-{Stealth[length=1.5mm]}](.2,.2)--(8.5,.2)node[right]{$k$};
\draw[-{Stealth[length=1.5mm]}](.2,.2)--(.2,8.6)node[above]{$n$};
\foreach \j in {1,...,6}{\node[below]at(\j,.18){$\j$};\node[left]at(.18,\j){$\j$};}
\node[below]at(7,.18){$\ge7$};\node[left]at(.18,7){$\ge7$};
\foreach \n in {3,...,7}{\draw[line width=.7pt](4,\n)ellipse(.19 and .24);}
\foreach \k in {5,6,7}{\draw[line width=.7pt](\k,3)ellipse(.19 and .24);}
\foreach \k in {5,6,7}{\draw[line width=.8pt](\k,4)ellipse(.19 and .24);}
\foreach \k in {5,6,7}{\foreach \n in {5,6,7}{\draw[line width=.7pt](\k,\n)ellipse(.19 and .24);}}
\draw[line width=.65pt](4.5,7.5)--(4.5,3.5)--(7.5,3.5);
\end{tikzpicture}
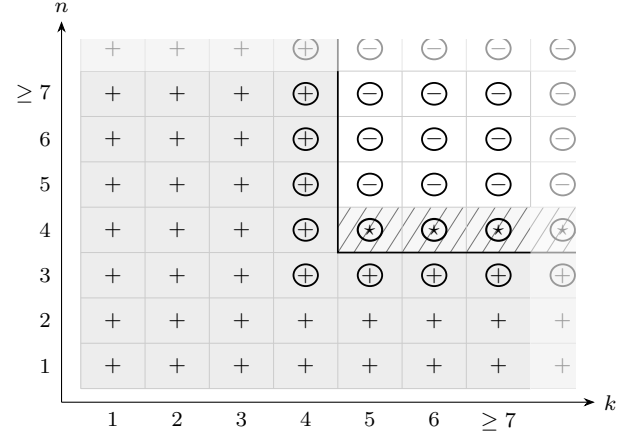
\caption{$\mathcal{C}_k(n) \subseteq \mathsf{GSC}(n)$ if $(k,n)$ lies in a region marked $+$ and $\mathcal{C}_k(n) \not\subseteq \mathsf{GSC}(n)$ if $(k,n)$ lies in a region marked $-$. For $n=4$ and $k\ge5$, gates defined over $\Q(i)$ up to scalar are generalized semi-Clifford.}
\label{fig:gsc-plane}
\end{inlinefigure}

We begin with fourth-level gates, which can be handled cleanly by a group-theoretic argument with full details in Appendix~\ref{app:finite-group}.

\begin{theorem}\label{thm:main-k4}
    For any $n$, $\C_4(n) \subseteq \GSC (n)$.
\end{theorem}
\begin{proof}[Sketch of Proof]
For $U\in\C_4(n)$ with arbitrary $n$, our argument relies on investigating the group
\[
\Gamma(U)=\langle\mathcal P(n),U\mathcal P(n)U^\dagger\rangle\subseteq\mathrm{PU}(2^n).
\]
The proof rests on the observation that if $\Gamma(U)$ is a finite $2$-group, then $U$ is generalized semi-Clifford. It thus suffices to show that $\Gamma(U)$ is a finite $2$-group for $U\in\C_4(n)$. Since the corrections lie in $\C_3(n)$, the conjugates $h_v\mathcal P(n)h_v^\dagger$ consist of Clifford gates and generate a finite group $G$, normalized by the $2$-group $U\mathcal P(n)U^\dagger$. One can show that $U\mathcal P(n)U^\dagger$ fixes a Sylow $2$-subgroup of $G$, which contains $\mathcal P(n)$ and hence every $h_v\mathcal P(n)h_v^\dagger$, so it equals $G$ and $\Gamma(U)$ is a $2$-group.
\end{proof}

We then consider the cases $n=3$ and $n=4$. We begin with the observation that although there are ways to reduce the generalized semi-Clifford problem directly to a finite problem, they are impractical to implement.

\begin{example}
To prove $\C_k(n)\subseteq\GSC(n)$, we can first reduce to a finite list of gates. Fix a basis $v_1,\dots,v_{2n}$ of $\F_2^{2n}$. If two gates $U,U'$ have the same corrections there, that is, $UP_{v_i}U^\dagger\doteq U'P_{v_i}U'^\dagger$ for every $i$, then $U\doteq U'P$ for a Pauli $P$. Indeed, $U'^\dagger U$ commutes with every $P_{v_i}$ up to a scalar, hence with every Pauli, so it is a Pauli up to a scalar. Now since all corrections lie in $\C_{k-1}(n)$, counting gates up to scalars gives
\[
 |\C_k(n)|\le4^n\,|\C_{k-1}(n)|^{2n}.
\]
We can then test each gate $U$ on the list. By Lemma~\ref{lem:main-local-global} and $h_{u+v}\doteq h_uh_v$, $U\in\GSC(n)$ exactly when $h_{v_1},\dots,h_{v_{2n}}$ normalize one common algebra $\A_M$, among $\prod_{j=1}^n(2^j+1)$ choices of the Lagrangian $M$. Starting from the $4^n|\Sp(2n,2)|$ Clifford gates, where $\Sp(2n,2)$ is the symplectic group of $(\F_2^{2n},\omega)$, the bound gives, up to scalars:
\begin{center}
\small
\begin{tabular}{@{}c|c|c|c|c@{}}
$n$&$|\C_2(n)|$&$|\C_3(n)|\le$&$|\C_4(n)|\le$&tests per gate\\\hline
$3$&$\approx9.3\times10^{7}$&$\approx10^{50}$&$\approx10^{299}$&$6\cdot135=810$\\
$4$&$\approx1.2\times10^{13}$&$\approx10^{107}$&$\approx10^{859}$&$8\cdot2{,}295=18{,}360$
\end{tabular}
\end{center}
Here a test asks whether one correction normalizes one algebra $\A_M$, and it has to be repeated for every level $k$.
\end{example}

Our strategy is to leverage the arithmetic information of a gate $U$ to reduce the problem to a small enough number of cases that we can check exhaustively using computer programs, which we produced in \texttt{Python} and \texttt{C} and made available at \begin{quote} \certrepo.
\end{quote}
The general theoretical tools in this section are developed in Appendix~\ref{app:common-strategy} with results on three qubits in Appendix~\ref{app:three-qubit} and four qubits in Appendix~\ref{app:four-qubit-results}.

Let $U\in\C_k(n)$. In general, the coefficients of $U$ lie in the tower of dyadic cyclotomic fields
\[
\Q(i)\subsetneq\Q(e^{i\pi/4})\subsetneq\Q(e^{i\pi/8})\subsetneq\cdots.
\]
In fact, there exists $\lambda\in \CC^\times$ such that all the coefficients of $W=\lambda U$ can be defined over a single field $K_r=\Q(\zeta_{2^r})$, where $\zeta_{2^r}=e^{2\pi i/2^r}$ and $r\ge2$ with $r\le\max\{2,k\}$ (Proposition~\ref{lem:conductor-plane}). We call the least such $r$ the \emph{conductor} of $U$ (the conductor of $K_r$ in the number-theoretic sense is $2^r$, and here we only record its exponent). More precisely,
\[
\operatorname{cond}(U)\coloneqq
\min\left\{r\ge2:
\begin{array}{l}
\lambda U\in\operatorname{Mat}_{2^n}(K_r)\\
\text{for some }\lambda\in\CC^\times
\end{array}\right\}.
\]
We call $U$ \textit{Gaussian} if $\operatorname{cond}(U)=2$, that is, some nonzero scalar multiple of $U$ has entries in $\Q(i)$.

An interesting observation is that for a fixed $n$, the conductor sometimes lets us check a single level to infer the information for all $k$'s at once.
\begin{example}
For four qubits, we have
\[
\text{$\C_k(4)\subseteq\GSC(4)$ for every $k$ $\ \Longleftrightarrow\ $ $\C_{4290}(4)\subseteq\GSC(4)$},
\]
proved as Proposition~\ref{prop:n4-finite-cutoffs} in the appendix: a counterexample, if one exists, can be taken of conductor at most $285$ and level at most $4290$. We remark that although the problem is now only on a single level, it is still far beyond enumeration. The diagonal gates of conductor at most $285$, that is, those with entries in $\mu_{2^{285}}=\{z\in\CC:z^{2^{285}}=1\}$, alone give $2^{285\cdot15}\approx10^{1287}$ gates of $\C_{4290}(4)$ up to scalars~\cite[Lemma~3 and Theorem~1]{CGK2017}. Even getting rid of the Clifford factors does not reduce it much. Since $\lambda C_1UC_2$ with $C_1,C_2$ Clifford is generalized semi-Clifford exactly when $U$ is (Lemma~\ref{lem:basic-operations}(1)), testing $U$ covers all of them. But there are only about $(1.2\times10^{13})^2\approx1.5\times10^{26}$ pairs $(C_1,C_2)$ up to scalars, so the diagonal gates alone still leave more than $10^{1287}/(1.5\times10^{26})>10^{1260}$ gates to test.
\end{example}

The conductor describes the arithmetic complexity of a gate and plays two roles in our proof. First, when every correction of a gate $U$ is Gaussian, i.e., has the minimal conductor, we can reduce the problem to a case check on a single exceptional gate for $n=3$ and $1{,}020$ gates for $n=4$. Moreover, it also acts as a second induction parameter. Once we fix the qubit number $n$, the conductor $r$, along with the hierarchy level $k$, can be used to do a nested double induction as in Theorem~\ref{thm:main-small-width} with $r=2$ or $k=4$ as the base case.

\begin{remark}
    We remark here that the notion of arithmetic invariants of gates is well-studied in the literature. Exact synthesis identifies the Clifford$+T$ circuits with the unitaries over $\Z[1/\sqrt2,i]$, graded by the \emph{least denominator exponent}, the least $k$ with $\sqrt2^{\,k}U$ over $\Z[\zeta_8]$~\cite{KMM2013,GilesSelinger2013}, with analogues over $\Z[e^{i\pi/n},1/2]$ for some even $n$ and over subrings of $\Z[1/\sqrt2,i]$~\cite{FGKM2015,AGR2020}. In the Clifford hierarchy, a diagonal gate $\operatorname{diag}_xe^{2\pi if(x)}$ with $f(x)=\sum_A\theta_A\prod_{i\in A}x_i$ lies in $\C_k$ exactly when $2^{k+1-|A|}\theta_A\in\Z$ for every $A\ne\varnothing$~\cite{CGK2017}, and a semi-Clifford gate lies in the hierarchy exactly when it is $C_1\prod_je^{i\pi a_jZ_j/2^{m_j}}C_2$ with Clifford $C_1,C_2$ and $Z$-type Paulis $Z_j$~\cite{Anderson_2024}. On the qudit side,~\cite{ChenDeSilva2024} shows that third-level gates on at most two qudits of prime dimension are semi-Clifford, by comparing algebraic sets reduced modulo the dimension. We notice here that the notion of the conductor is coarser than these invariants: it is only the least dyadic cyclotomic field of definition, and we track it for a family of gates rather than as a cost of a single gate.
\end{remark}

Our starting point is to use all the Pauli corrections $h_v=UP_vU^\dagger$ to characterize the structure of a gate $U$ itself. Recall from Lemma~\ref{lem:main-local-global} that a gate $U$ is generalized semi-Clifford if and only if all $h_v$'s have a common Lagrangian $M$ such that
\[
h_v\A_Mh_v^\dagger=\A_M\qquad \text{for every $v\in \F_2^{2n}$},
\]
which is strictly stronger than requiring every $h_v$ to be generalized semi-Clifford, and hence very difficult to check in practice. However, for $n\leq 4$, it suffices to show that all $h_v$ are Gaussian.

\begin{theorem}\label{gauss-thm:rigidity}\label{thm:n3-gaussian}\label{thm:n4-gaussian}
Let $n\in\{3,4\}$ and $U\in\C(n)$, and suppose that every correction $h_v=UP_vU^\dagger$ is Gaussian. Then $U$ and every $h_v$ are generalized semi-Clifford. Moreover,
\begin{enumerate}
\item if $n=3$, then $h_v\in\C_3(3)$ for every $v$ and $U\in\C_4(3)$;
\item if $n=4$, then $h_v\in\C_5(4)$ for every $v$ and $U\in\C_6(4)$.
\end{enumerate}
\end{theorem}
\begin{proof}[Sketch of Proof]

By induction on the level (Lemma~\ref{app:lem-ge-equiv}), we may assume that every correction $h_v$ is also generalized semi-Clifford. For $v\ne0$, $h_v$ is then a traceless Hermitian involution, and up to Clifford conjugation it has the normal form
\[
 \ket x\mapsto\zeta_8^\lambda\,i^{f(x)}\ket{\pi x},\qquad\lambda\in\{0,1\},\quad f:\F_2^n\to\Z/4,
\]
where $\pi$ is an involution preserving a complete affine flag (Proposition~\ref{app:lem-R_n}). Since $\lambda$, $f$ and $\pi$ range over finite sets, this leaves finitely many gates to examine.

For $n=3$, a finite check shows that each such gate lies in $\C_3(3)$, except the Clifford conjugates of one fourth-level gate $H_*$ (Proposition~\ref{gauss-lem:exceptional}), and we then show that the latter case is impossible. Hence every $h_v$ lies in $\C_3(3)$, and $U\in\C_4(3)\subseteq\GSC(3)$ by Theorem~\ref{thm:main-k4} (Proposition~\ref{gauss-prop:n3-gsc}).

For $n=4$, the normal forms give $1{,}020$ gates up to Clifford conjugation. A finite search over them shows that either the $h_v$ already have a common generalized semi-Clifford Lagrangian pair, or $U$ maps some Pauli to a Pauli (Lemma~\ref{app:finite-endpoint}). In the latter case $U$ is Clifford-equivalent to $\operatorname{diag}(A,B)$ with three-qubit blocks, and a case analysis on $A$ and $B$ produces a common generalized semi-Clifford Lagrangian pair. The details are in Appendices~\ref{gauss-app:ge3} and~\ref{app:four-qubit}.
\end{proof}

This result settles the case when all $h_v$ are Gaussian. To handle the case for a gate $U$ with some non-Gaussian correction $h_v$, we need to construct a gate $U'$ with a smaller conductor while still capturing the structure of $U$.

We construct $U'$ from the Galois action on the entries of $U$. Let $M=\operatorname{cond}(U)$, choose $W\doteq U$ with entries in $K_M$, and let $\delta$ generate $\operatorname{Gal}(K_M/K_{M-1})$ with $\delta(\zeta_{2^M})=-\zeta_{2^M}$. Then $\operatorname{cond}(U)<M$ exactly when $\delta(W)\doteq W$, that is, when the \emph{Galois cocycle} $\sigma\mapsto W^{-1}\sigma(W)$ of $U$ is trivial up to scalar at $\delta$. So $W^{-1}\delta(W)$ is the obstruction to a smaller conductor, and one can use it to detect suitable candidates for $U'$.

For our gates the cocycle has a particular shape. Let $n\le4$, let every $h_v$ be generalized semi-Clifford, let $m=\max_v\operatorname{cond}(h_v)>2$, and let $\sigma$ generate $\operatorname{Gal}(K_M/K_{m-1})$, so that $\delta=\sigma$ if $m=M$ and $\delta=\sigma^2$ if $m=M-1$. Then, after multiplying $U$ on the right by a Clifford gate,
\[
 W^{-1}\sigma(W)\doteq\operatorname{diag}_x\,i^{p\cdot x}(-1)^{q(x)}
\]
for some $p\in\F_2^n$ and $q:\F_2^n\to\F_2$, and hence $W^{-1}\delta(W)\doteq Z^p$ if $m=M-1$ (Steps~1--3 of the proof of Lemma~\ref{lem:galois-normal-form}, which apply for $n\le4$).

\begin{theorem}\label{thm:conductor-descent}
    For any $U\in \C(3)$ such that every $h_v$ is generalized semi-Clifford, if $h_v$ is non-Gaussian for some $v\in \F_2^6$, then there exists another gate $U'\in \C(3)$ such that
    \[
    U'\in \GSC(3) \Longleftrightarrow U\in \GSC(3), \quad \mathrm{cond}(U')< \mathrm{cond}(U).
    \]
\end{theorem}
\begin{proof}[Sketch of Proof]
    Let $U\in\C(3)\setminus\GSC(3)$ with generalized semi-Clifford corrections, but not all Gaussian. On three qubits two finite arguments sharpen the shape of the cocycle. First, $\deg q\le2$ (Lemma~\ref{lem:galois-normal-form}), so the cocycle is determined by $p$ and the quadratic part of $q$. Second, writing $T$ for the symmetric binary matrix formed by $p$ and the quadratic part of $q$, the gates $UZ^cU^\dagger$ with $c\in\operatorname{im}T$ are Clifford involutions, and a case analysis of the groups they generate, at most three commuting involutions among $3{,}843$ explicit Clifford gates, shows that $m=M-1$, $p\ne0$ and $UZ^pU^\dagger$ is a Pauli (Lemma~\ref{lem:m-M-1}). Thus, after multiplying $U$ by Clifford gates on both sides,
\[
 \delta(W)\doteq WZ^p,\qquad UZ^pU^\dagger=P\in\mathcal P(3)
\]
for some $0\ne p\in\F_2^3$. We define
\[
U'=UR,\qquad R=\frac{I+Z^p}{2}+\xi\,\frac{I-Z^p}{2},\qquad \xi=\zeta_{2^M}^{-1}.
\]
The rotation $R$ is chosen to cancel the Galois cocycle. Since $\delta(\xi)=-\xi$, we have $\delta(R)=\frac{I+Z^p}{2}-\xi\frac{I-Z^p}{2}=Z^pR$, so $W'=WR$ has entries in $K_M$ and satisfies
\[
W'^{-1}\delta(W')=R^{-1}\,\bigl(W^{-1}\delta(W)\bigr)\,\delta(R)\doteq R^{-1}Z^p\,Z^pR=I.
\]
Hence $\delta(W')\doteq W'$, and $\mathrm{cond}(U')<M$. Since $P$ is a Pauli, $U'$ stays in the hierarchy and is generalized semi-Clifford exactly when $U$ is (Lemma~\ref{lem:conductor-step}).
\end{proof}

This conductor descent argument paves the way to prove that all gates in the Clifford hierarchy for three qubits are generalized semi-Clifford.
\begin{theorem}\label{thm:main-small-width}
For every $k$, $\C_k(3)\subseteq\GSC(3)$.
\end{theorem}
\begin{proof}

Let $U\in \C_k(3)$. We induct on $(r,k)$ in lexicographic order, where $r=\operatorname{cond}(U)$ (see Fig.~\ref{fig:conductor-plane}).

For $k\le4$, $U\in\GSC(3)$ by Theorem~\ref{thm:main-k4}. If $r=2$, then $U$ is Gaussian, so every $h_v$ is Gaussian, and $U\in\GSC(3)$ by Theorem~\ref{gauss-thm:rigidity}.

Now assume $r>2$, $k>4$, and the statement is true for all $(r',k')$ with $r'<r$, or with $r'=r$ and $k'<k$. Then $h_v \in \C_{k-1}(3)$ with $\mathrm{cond}(h_v)\leq r$ for all $v \in \F^{6}_2$. For any $v$, if $\mathrm{cond}(h_v)<r$, then $h_v$ is generalized semi-Clifford by the induction on $r$. If $\mathrm{cond}(h_v)=r$, then $h_v$ is generalized semi-Clifford by the inner induction on the hierarchy level. Thus every $h_v$ is generalized semi-Clifford. If all $h_v$ are Gaussian, then again $U$ is generalized semi-Clifford by Theorem~\ref{gauss-thm:rigidity}.

Assume otherwise, so $\max_v\operatorname{cond}(h_v)>2$. Suppose $U\not \in \GSC(3)$. Then by Theorem~\ref{thm:conductor-descent} there exists
\[ U'\in\C(3)\setminus\GSC(3)
\]
with $\operatorname{cond}(U')<r$. But by the outer induction, $U'\in \GSC(3)$, hence a contradiction.
\end{proof}

\begin{inlinefigure}
\centering
\begin{tikzpicture}[x=.85cm,y=.6cm,font=\footnotesize]
\color{black}
\foreach \rr in {2,...,7}{\foreach \k in {2,...,7}{
\ifnum\rr>\k\relax
\else\ifnum\rr<4
\fill[black!20](\rr-.5,\k-.5)rectangle(\rr+.5,\k+.5);
\node at(\rr,\k){$+$};
\else\ifnum\rr=4
\ifnum\k<6
\fill[black!20](\rr-.5,\k-.5)rectangle(\rr+.5,\k+.5);
\node at(\rr,\k){$+$};
\else
\fill[black!7](\rr-.5,\k-.5)rectangle(\rr+.5,\k+.5);
\fi
\else
\fill[black!7](\rr-.5,\k-.5)rectangle(\rr+.5,\k+.5);
\fi\fi\fi}}
\foreach \j in {2,...,8}{
\draw[black!20,very thin](\j-.5,1.5)--(\j-.5,7.5);
\draw[black!20,very thin](1.5,\j-.5)--(7.5,\j-.5);}
\draw[line width=.65pt](2.5,1.5)--(2.5,2.5)--(3.5,2.5)--(3.5,3.5)--(4.5,3.5)--(4.5,4.5)--(5.5,4.5)--(5.5,5.5)--(6.5,5.5)--(6.5,6.5)--(7.5,6.5)--(7.5,7.5);
\draw[-{Stealth[length=1.5mm]},line width=.7pt](4.08,5.8)--(4.08,5.2);
\draw[-{Stealth[length=1.5mm]},line width=.7pt](3.85,5.85)--(3.2,5.2);
\fill(4,6)circle(1.3pt);
\draw[-{Stealth[length=1.5mm]}](1.2,1.2)--(7.85,1.2)node[right]{$r$};
\draw[-{Stealth[length=1.5mm]}](1.2,1.2)--(1.2,7.85)node[above]{$k$};
\foreach \j in {2,...,7}{\node[below]at(\j,1.18){$\j$};}
\foreach \j in {2,...,6}{\node[left]at(1.18,\j){$\j$};}
\node[left]at(1.18,7){$\ge7$};
\end{tikzpicture}
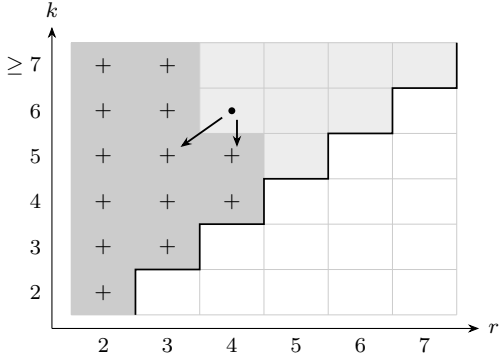
\caption{The induction process of Theorem~\ref{thm:main-small-width}, where $r=\operatorname{cond}(U)$ is the conductor and $k$ the level of $U\in\C_k(3)$.}
\label{fig:conductor-plane}
\end{inlinefigure}

For the case $n=4$, we expect a similar conductor descent argument to follow through. The difficulty lies in the finite case analysis to find a Pauli $Q$ such that $P=UQU^\dagger$ is a Pauli. The Galois cocycle still has the shape above, but the argument bounding $\deg q$ does not extend. A priori the cocycle has $2^4\cdot2^{11}=32{,}768$ possible shapes, given by $p$ and by $q$ modulo affine functions. Moreover, the gates $UZ^cU^\dagger$ need not be Clifford, which makes the case analysis even more complicated.

If such a $Q$ exists, then we can take
\[
U'=U\Bigl(\tfrac{I+Q}{2}+\xi\,\tfrac{I-Q}{2}\Bigr)=\Bigl(\tfrac{I+P}{2}+\xi\,\tfrac{I-P}{2}\Bigr)U,\qquad \xi=\zeta_{2^M}^{-1},
\]
which would satisfy our desired properties.
\begin{conjecture}\label{conj:n4-descent}
Let $U\in\C(4)$. Assume every $h_v$ is generalized semi-Clifford and some $h_v$ is not Gaussian. Then there exists another gate $U'\in \C(4)$ such that
    \[
    U'\in \GSC(4) \Longleftrightarrow U\in \GSC(4), \quad \mathrm{cond}(U')< \mathrm{cond}(U).
    \]
\end{conjecture}

Now, by the same induction proof of Theorem~\ref{thm:main-small-width}, Conjecture~\ref{conj:n4-descent} would imply $\C_k(4)\subseteq\GSC(4)$ for every $k$, leading to the following characterization.
\begin{conjecture}\label{conj:gsc-boundary}
For positive integers $k,n$,
\[
\mathcal C_k(n)\subseteq\GSC(n)
\quad\Longleftrightarrow\quad
k\leq4\quad\text{or}\quad n\leq4.
\]
\end{conjecture}

\section{Beyond generalized semi-Clifford}\label{sec:conjectures}

We now present some structural results that are natural extensions of the generalized semi-Clifford structure, with details discussed in Appendix~\ref{sec:beyond-gsc}. 

A gate $U\in\C_k(n)$ is generalized semi-Clifford exactly when it can be written as \[
U=C_1
\left(\sum_{x\in\F_2^{n}}\lambda_x\ket{\pi(x)}\bra{x}\right)C_2,
\]
for some $C_1, C_2\in \C_2(n)$, a permutation $\pi:\F_2^n\to \F_2^n$, and $\lambda_x\in\mathbb T$.

Let $n=n_c+n_t$. We denote by $\mathcal{F}(n_c, n_t)$ the set of gates $U\in\C(n)$ of the form
\[
C_1
\left(\sum_{x\in\F_2^{n_c}}\ket{\pi(x)}\bra{x}
\otimes V_x\right)C_2
\]
where $C_1, C_2\in \C_2(n)$, $\pi:\mathbb F_2^{n_c}\to \F_2^{n_c}$ is a permutation, and $V_x\in \mathrm{U}(2^{n_t})$. Here $n_c$ is the number of control qubits and $n_t$ is the number of target qubits.

Another characterization of a generalized semi-Clifford gate is that it maps the algebra of an isotropic subspace of dimension $n$ onto the algebra of another (Lemma~\ref{lem:lagrangian}). More generally, for $0\le n_c\le n$ we say that a gate $U$ is \emph{$n_c$-generalized semi-Clifford} if $U\mathcal A_SU^\dagger=\mathcal A_T$ for some isotropic subspaces $S,T\subseteq\F_2^{2n}$ of dimension $n_c$. Thus the generalized semi-Clifford gates are the $n$-generalized semi-Clifford gates, and every gate is $0$-generalized semi-Clifford trivially. Notice that a gate can be $n_c$-generalized semi-Clifford for several values of $n_c$. We prove the following lemma in Appendix~\ref{sec:beyond-gsc}.
\begin{lemma}\label{lem:control-form}
Let $0\le n_c\le n$ and $n_t=n-n_c$. A gate $U\in\C(n)$ lies in $\mathcal F(n_c,n_t)$ if and only if it is $n_c$-generalized semi-Clifford, that is, $U\mathcal A_SU^\dagger=\mathcal A_T$ for some isotropic subspaces $S,T\subseteq\F_2^{2n}$ of dimension $n_c$.

\end{lemma}

We denote
\[
\mathcal{F}_j(n_c, n_t)\coloneqq\{U\in \mathcal{F}(n_c, n_t): V_x\in \C_j(n_t), \ \forall x\}.
\]

When $j=2$, that is, each block $V_x$ is Clifford, there is a clean relationship between them and generalized semi-Clifford gates, which we prove in Appendix~\ref{subsec:clifford-fibers}.
\begin{proposition}\label{clifford-fiber}
    \[
    \mathcal{F}_2(n_c, n_t) \subseteq \GSC(n_c+n_t) \Longleftrightarrow n_c \leq 1 \text{ or } n_t\leq 2.
    \]
\end{proposition}

Hence the first case of a non-generalized semi-Clifford gate $U\in \mathcal{F}_2(n_c, n_t)$ should have $n_c\geq 2$ and $n_t\geq 3$, which requires $n=n_c+n_t\geq 5$. In fact, we can show that for all $k,n\geq 5$, there is a non-generalized semi-Clifford gate with $n$ qubits and exact level $k$ which offers the negative results in Fig.~\ref{fig:combined}. We prove it in Appendix~\ref{subsec:orthogonal-fibers}.
\begin{theorem}\label{thm:gsc-counterexamples}
    For every $k\geq 5$ and $n\geq 5$, there is a gate $U\in \mathcal{F}_2(2, n-2)$ such that
    \[
    U\in \C_k(n) \setminus \C_{k-1}(n), \quad U\not \in \GSC(n).
    \]
\end{theorem}

We note that a different counterexample for gates in $\mathcal{C}_k(n) \setminus \GSC(n)$ was also obtained in~\cite{deSilvaLautsch2026False} and arises from a different construction. 

We propose the following conjecture as a refinement to the original conjecture of~\cite{sc-ZCC2008}.
\begin{conjecture}\label{conj:fibers}
Every gate in $\C(n)$ is $n_c$-generalized semi-Clifford for some $1\leq n_c\leq n$. More precisely, if $U\in\C_k(n)$ and $k\geq3$, there are Clifford gates $C_1,C_2\in\C_2(n)$ and $1\leq n_c\leq n$, with $n_t=n-n_c$, such that
\[
U=C_1
\left(\sum_{x\in\F_2^{n_c}}\ket{\pi(x)}\bra{x}
\otimes V_x\right)C_2,
\]
where $\pi:\F_2^{n_c}\to\F_2^{n_c}$ is a permutation and $V_x\in\C_{k-3}(n_t)$ for every $x$, with the convention $\C_0(m)=\C_j(0)=\mathbb T$ for scalars. In other words,
\[
\C_k(n)\subseteq \bigcup_{n_c=1}^{n}\mathcal{F}_{k-3}(n_c, n-n_c).
\]
\end{conjecture}

\begin{inlinefigure}
\centering
\begin{tikzpicture}[x=1mm,y=1mm,font=\footnotesize,
 gate/.style={draw,fill=white,minimum height=6mm,inner sep=2pt}]
\draw(0,9)--(78,9);
\draw(0,0)--(78,0);
\draw(3.2,7.6)--(4.8,10.4);
\draw(3.2,-1.4)--(4.8,1.4);
\node[above,font=\scriptsize]at(4,10.4){$n_c$};
\node[above,font=\scriptsize]at(4,1.4){$n_t$};
\node[draw,fill=white,minimum width=9mm,minimum height=17mm]at(13,4.5){$C_2$};
\draw[dashed,black!50,rounded corners=2pt](21,-5.5)rectangle(60,14.5);
\draw(31,9)--(31,0);
\node[gate,minimum width=9mm]at(31,0){$V_x$};
\node[draw,fill=white,minimum width=4.5mm,minimum height=4.5mm,inner sep=0pt]at(31,9){$x$};
\node[gate,minimum width=7mm]at(46,9){$\pi$};
\node[above]at(55,9){$\pi(x)$};
\node[draw,fill=white,minimum width=10mm,minimum height=17mm]at(69,4.5){$C_1$};
\node[below]at(40.5,-5.5){$\sum_x\ket{\pi(x)}\bra{x}\otimes V_x$};
\end{tikzpicture}
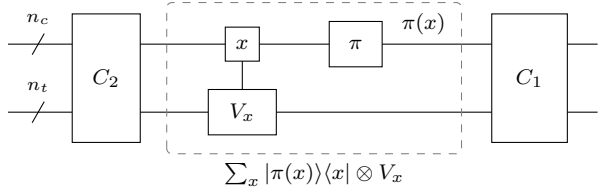
\caption{The decomposition of a gate $U\in\C_k(n)$ predicted by Conjecture~\ref{conj:fibers}.}
\label{fig:controlled-block}
\end{inlinefigure}

Figure~\ref{fig:controlled-block} depicts this decomposition as a circuit. Here the integers $n_c$ and $n_t$ are chosen together with the Clifford gates $C_i$ and need not be unique. We remark here that the conjecture for $k=3$ and $k=4$ says exactly that $\C_3(n)$ and $\C_4(n)$ are generalized semi-Clifford, which holds by~\cite{sc-BS2010} and Theorem~\ref{thm:main-k4}. For $k=5$ it predicts that every $U\in \C_5(n)$ has the above form with Clifford blocks $V_x\in \C_2(n_t)$. If true, Proposition~\ref{clifford-fiber} then gives $\C_5(4)\subseteq\GSC(4)$, since $n_c\le1$ or $n_t\le2$ whenever $n_c\ge1$ and $n_c+n_t=4$.

\section{Discussion}

In this article, we have investigated when $\mathcal{C}_k(n)$ is a subset of $\SC(n)$ and $\GSC(n)$ as $k$ and $n$ vary, which leads to several directions for future work. Generalized semi-Clifford gates were proposed in~\cite{sc-ZCC2008} as a structural description of the Clifford hierarchy. Apart from $n=4$ and $k\ge5$, our results determine for which pairs $(k,n)$ it describes all of $\mathcal C_k(n)$, and Conjecture~\ref{conj:fibers} proposes how the description should change beyond them. 

The first natural direction is to explore the conductor descent argument in Section~\ref{sec:gsc} for $n=4$. If successful, it will complete the picture of when $\mathcal{C}_k(n) \subseteq \GSC(n)$.

Second, as exhibited in Section~\ref{sec:gsc}, the conductor of a gate plays a central role in the study of three-qubit and four-qubit gates. It is thus natural to explore whether there is a finer structure determined by the triple $(r,k,n)$ in relation to the characterization of the Clifford hierarchy and fault-tolerant quantum computing.

We also believe that some of our techniques naturally extend to the case of qudits. For example, when the local dimension $d\geq3$ is a prime, there are more algebraic tools that can reduce the finite-case analysis substantially.

Lastly, we would like to point out that we learned of the results of~\cite{deSilvaLautsch2026False,level4} after we had independently obtained our results.

\begin{acknowledgments}
We thank the organizers of the 2026 Gene Golub SIAM Summer School on Fault-Tolerant Algorithms in Quantum Computing, which provided us with the opportunity to learn about this fascinating subject. We are especially grateful to Robert Calderbank and Di Fang for their efforts in organizing the event, as well as for their helpful comments and discussions regarding this work. We also thank the other members of our quantum error correction reading group for many stimulating and insightful discussions.

\end{acknowledgments}

\section*{AI disclosure}

This article originated from the authors' attempt to test the conjecture of~\cite{sc-ZCC2008} using ChatGPT 5.6 Sol after learning about the conjecture at the Gene Golub SIAM Summer School. On August 3, 2026, a counterexample was produced in $\mathcal{C}_5(5)$. The authors then interacted extensively with ChatGPT 5.6 Sol (and, more recently, ChatGPT 6 Astra) and Claude Fable 5 (and, more recently, Fable 5.1 and Opus 5.5) to explore structures underlying the Clifford hierarchy. As the project evolved, the authors gradually formed the conjectural picture summarized in Figs.~\ref{fig:sc-plane} and~\ref{fig:gsc-plane}. These patterns were initially explored computationally, and we subsequently developed a clearer understanding of the arithmetic structure underlying these phenomena and derived the reductions presented in Appendices~\ref{app:common-strategy},~\ref{app:three-qubit}, and~\ref{app:four-qubit-results}. We have also used the LLMs to help with the writing and polishing of the paper. The authors welcome questions regarding the role of AI in this work, as well as comments and suggestions concerning best practices for the disclosure of AI-assisted mathematical research.

\renewcommand{\bibsection}{\section*{\refname}}
\addtocontents{toc}{\protect\AppendixContentsOff}
\bibliographystyle{apsrev4-2}
\bibliography{WZ}

\clearpage
\onecolumngrid
\setlength{\abovedisplayskip}{7pt plus 2pt minus 2pt}
\setlength{\belowdisplayskip}{7pt plus 2pt minus 2pt}
\setlength{\abovedisplayshortskip}{4pt plus 2pt}
\setlength{\belowdisplayshortskip}{5pt plus 2pt minus 1pt}
\bookmarksetup{depth=2}

\begingroup
\makeatletter
\def\GscTocSection{\@dottedtocline{1}{0em}{2em}}
\def\GscTocSubsection{\@dottedtocline{2}{1.5em}{2em}}
\def\AppendixContentsOn{\let\l@section\GscTocSection\let\l@subsection\GscTocSubsection}
\def\AppendixContentsOff{\let\l@section\@gobbletwo\let\l@subsection\@gobbletwo}
\def\l@f@section{\addpenalty{\@secpenalty}\addvspace{5pt}\bfseries}
\AppendixContentsOff
\makeatother
\endgroup

\addtocontents{toc}{\protect\AppendixContentsOn}
\appendix
\makeatletter
\counterwithin*{theorem}{section}%
\renewcommand{\thetheorem}{\thesection.\arabic{theorem}}%
\renewcommand{\theequation}{\thesection.\arabic{equation}}%
\makeatother

\section{Preliminaries}\label{app:prerequisites}

We summarize the notations used in the appendix in Table~\ref{tab:notation}. Throughout, a Pauli is written as $\lambda P_v$ with $\lambda\in\mu_4$ and $v\in V_n=\F_2^{2n}$, and we call $v$ its \emph{label}. We write $x^a$ and $z^b$ for the labels of $X^a$ and $Z^b$, and $x_j$ and $z_j$ for those of $X_j$ and $Z_j$, also when the qubit $j$ is named by a letter, as in $z_y$.

\begin{table}
\begin{tabular}{@{}p{.23\linewidth}p{\dimexpr.77\linewidth-2\tabcolsep\relax}@{}}
\toprule
Notation & Object \\
\midrule
$\T$ &
$\{\lambda\in\CC:|\lambda|=1\}$ \\
$\mathrm{PU}(2^n)$ &
$\mathrm U(2^n)/\T I$ \\
$A\doteq B$ &
$A=\lambda B$ for some $\lambda\in\CC^\times$ \\
$\mu_j$ &
$\{z\in\CC:z^j=1\}$ \\
$\zeta_{2^M}$ &
$e^{2\pi i/2^M}$ for $M\geq2$ \\
$K_M$ &
$\Q(\zeta_{2^M})$ \\
$\operatorname{cond}(U)$ &
the conductor of $U$\\
$V_n$ &
$\F_2^{2n}$, with elements $v=(x,z)$ for $x,z\in\F_2^n$ \\
$\omega$ &
$\omega((x,z),(x',z'))=x\cdot z'+x'\cdot z$ on $V_n$ \\
$\Sp(V_n)$ &
the symplectic group of $(V_n,\omega)$, also written as $\Sp(2n,2)$ \\
$P_v$ &
$i^{x\cdot z}X^xZ^z$ for $v=(x,z)\in V_n$, the canonical Hermitian Pauli; $v$ is the \emph{label} of $\lambda P_v$ for $\lambda\in\mu_4$, and $z_i$ is the label of $Z_i$ \\
$\mathcal H_x$ &
the joint eigenspace of $Z_1,\dots,Z_j$ with eigenvalues $(-1)^{x_1},\dots,(-1)^{x_j}$, for $x\in\F_2^j$ \\
$\mathscr R_n$, $n=3,4$ & traceless finite-level Gaussian generalized semi-Clifford Hermitian involutions (Proposition~\ref{app:lem-R_n}) \\
$\Span S$, $\Span_{\CC}S$ &
the $\F_2$-span of a set $S$ of vectors over $\F_2$, and the complex span of a set $S$ of matrices \\
$\langle S\rangle$ &
the group generated by a set $S$ \\
$\A_S$ &
$\Span_{\CC}\{P_v:v\in S\}$\\
$L_Z$ &
$\Span\{z_1,\dots,z_n\}$ \\
$\A_Z$ &
$\A_{L_Z}$, the algebra of diagonal matrices \\
$\operatorname{supp}_PG$ &
$\{u:c_u\ne0\}$ for $G=\sum_uc_uP_u$ \\
$\mathscr L(U)$ &
$\{(L,M):L,M\le V_n\text{ Lagrangian},\ U\A_LU^\dagger=\A_M\}$, the GSC pairs of $U$ \\
$h_v$ &
$UP_vU^\dagger$ \\
$\mathsf P(U)$, $\mathsf C(U)$ &
$\{v:h_v\in\Pau(n)\}$ and $\{v:h_v\in\C_2(n)\}$ \\
$\gamma_C$ &
the symplectic action of $C\in\C_2(n)$, defined by $CP_vC^\dagger\doteq P_{\gamma_Cv}$ \\
$E_D \subseteq \Sp(V_3)$ &the subgroup generated by the $\gamma_{DP_vD^\dagger}$\\
\bottomrule
\end{tabular}
\caption{Notations used throughout the appendix.}
\label{tab:notation}
\end{table}

Recall that a maximal abelian $*$-subalgebra (MASA) of $\operatorname{Mat}_{2^n}(\mathbb C)$ is a unital abelian $*$-subalgebra that is maximal with respect to inclusion. For a Lagrangian subspace $L$, $\mathcal A_L$ is a MASA, which we call a \emph{Pauli MASA}, since a matrix $\sum_wc_wP_w$ commutes with every $P_v$, $v\in L$, exactly when $c_w=0$ for $w\notin L^\perp=L$. Equivalently, a Pauli MASA is the complex linear span of a maximal abelian subgroup of the Pauli group. We begin with the following proposition, which includes Lemma~\ref{lem:lagrangian}. For $C\in\C_2(n)$ we write $\gamma_C\in\Sp(V_n)$ for its symplectic action, defined by $CP_vC^\dagger\doteq P_{\gamma_Cv}$ for $v\in V_n$. Then $C\A_SC^\dagger=\A_{\gamma_CS}$ for every $S\le V_n$, and let $\mathscr L(C) \coloneqq \{(L,\gamma_CL):L\text{ Lagrangian}\}$. For a gate $U$, let
\[
\mathsf P(U)
\coloneqq
\left\{
v\in V_n :
UP_vU^\dagger \doteq P
\text{ for some } P\in\mathcal P(n)
\right\}, \quad
\mathsf C(U)
\coloneqq
\left\{
v\in V_n :
UP_vU^\dagger \doteq C
\text{ for some } C\in\mathcal C_2(n)
\right\}.
\]
These are vector subspaces of $V_n$, since $UP_{u+v}U^\dagger\doteq(UP_uU^\dagger)(UP_vU^\dagger)$.

\begin{proposition}\label{lem:lagrangian-expanded}
If $\mathcal A\subseteq\operatorname{Mat}_{2^n}(\CC)$ is a MASA, then the following are equivalent:
\begin{enumerate}
    \item $\mathcal A$ is the linear span of a maximal abelian subgroup of $\mathcal{P}(n)$; \label{item:1}
    \item for all $P\in\mathcal{P}(n)$, $P\mathcal A P^\dagger=\mathcal A$; \label{item:2}
    \item there exists a Lagrangian subspace $L\subseteq\F_2^{2n}$ such that \label{item:3}
    \[
    \mathcal A=\operatorname{span}_{\CC}\{P_v:v\in L\}
    \eqqcolon\mathcal A_L.
    \]
\end{enumerate}
In particular, $U \in \mathrm U(2^n)$ is generalized semi-Clifford if and only if there exist Lagrangian subspaces $L$, $M \subseteq \F_2^{2n}$ such that $U \mathcal{A}_L U^\dagger = \mathcal{A}_{M}$. Moreover, let $\mathcal B\subseteq\operatorname{Mat}_{2^n}(\CC)$ be a unital $*$-subalgebra with $P\mathcal BP^\dagger=\mathcal B$ for all $P\in\mathcal P(n)$. Then $\mathcal B=\A_S$ for a subspace $S\le V_n$, which is isotropic if $\mathcal B$ is abelian. If $U\mathcal BU^\dagger=\A_L$ for a gate $U$ and a Lagrangian $L$, then $\mathcal B=\A_M$ for a Lagrangian $M$.
\end{proposition}

\begin{proof}
Suppose first $\mathcal{A} = \operatorname{span}_{\CC} \mathcal{H}$ for a maximal abelian subgroup $\mathcal{H} \leq \mathcal{P}(n)$. Let $L = q(\mathcal{H})$. Since $q$ is a group homomorphism, $q(\mathcal{H})$ is an additive subgroup of $\F^{2n}_2$. Every additive subgroup of an $\F_2$-vector space is an $\F_2$-linear subspace, so $L$ is a vector subspace.

For $v$, $v' \in L$, let $P_v$, $P_{v'} \in \mathcal{H}$ such that $q(P_v)= v$ and $q(P_{v'}) = v'$. Since $\mathcal H$ is abelian, $[P_v, P_{v'}] = 0$, so $\omega(v, v')=0$. Thus, $L$ is isotropic. Suppose $w \in \F^{2n}_2$ satisfies $\omega(v,w) = 0$ for all $v \in L$. Then $[P_w, P] = 0$ for all $P \in \mathcal{H}$ since $\omega(w, q(P)) = 0$. By maximality of $\mathcal{H}$, $P_w \in \mathcal{H}$. Thus, $w \in L$, so $L$ is Lagrangian. It is straightforward to check that $\mathcal{H} = q^{-1}(L)$, i.e., $\mathcal{H} = \{\lambda P_v: v \in L,\, \lambda = \pm 1,\, \pm i\}$. Consequently, $\mathcal{A} = \operatorname{span}_{\CC}(\mathcal H) = \operatorname{span}_{\CC}\{P_v:v \in L\} = \mathcal{A}_L$.

Conversely, let $L \subseteq \F^{2n}_2$ be a Lagrangian subspace. An analogous computation to above shows that $q^{-1}(L)$ is maximal abelian and hence $\operatorname{span}_{\CC} q^{-1}(L) = \operatorname{span}_{\CC}\{P_v:v \in L\} = \mathcal{A}_L$. This concludes the equivalence between~\ref{item:1} and~\ref{item:3}.

We now prove the equivalence between~\ref{item:2} and~\ref{item:3}. Let $L \subseteq \F^{2n}_2$ be a Lagrangian subspace. For every $v,w\in\F_2^{2n}$, $P_wP_vP_w^\dagger=(-1)^{\omega(v,w)}P_v$, so $P_w\mathcal A_LP_w^\dagger=\mathcal A_L$.

Conversely, suppose $P\mathcal A P^\dagger=\mathcal A$ for all $P\in\mathcal{P}(n)$. For $v\in\F_2^{2n}$, consider
\[
\Pi_v(X)
\coloneqq
\frac{1}{2^{2n}}
\sum_{w\in\F_2^{2n}}
(-1)^{\omega(v,w)}P_w XP_w^\dagger.
\]
For $u\in\F_2^{2n}$,
\[\Pi_v(P_u)
=
\frac{1}{2^{2n}}
\sum_{w\in\F_2^{2n}}
(-1)^{\omega(v,w)}
(-1)^{\omega(u,w)}P_u
=
\frac{1}{2^{2n}}
\sum_{w\in\F_2^{2n}}
(-1)^{\omega(u+v,w)}P_u = \delta_{u,v}P_u,\] so $\Pi_v$ is a projection onto $\CC P_v$.

Since $P_v \mathcal AP_v^\dagger =\mathcal A$, $\Pi_v(\mathcal A)\subseteq\mathcal A$ for all $v \in \F_2^{2n}$. For $X\in\mathcal A$, we expand it in the Pauli basis
\[
X=\sum_{u\in\F_2^{2n}}c_uP_u.
\]
Then $\Pi_v(X)=c_vP_v\in\mathcal A$. Therefore, $P_v \in \mathcal{A}$ for all $v$ with $c_v \neq 0$. It follows that
\[
\mathcal A
=
\operatorname{span}_{\CC}\{P_v:v\in T\}, \quad \text{where } T=\{v\in\F_2^{2n}:P_v\in\mathcal A\}.
\]

Since $\mathcal A$ is abelian, $\omega(u,v)=0$ for all $u,v \in T$. Thus $L \coloneqq \operatorname{span}_{\F_2}T$ is an isotropic subspace. On the other hand, since the Pauli matrices are linearly independent, $\lvert T \rvert =\dim_{\CC}\mathcal A=2^n$, as a MASA of $\operatorname{Mat}_{2^n}(\CC)$ is unitarily conjugate to the diagonal matrices. Since $T\subseteq L$,
\[
2^n=|T|\le |L|=2^{\dim L}\le 2^n.
\]
Therefore $\dim L=n$ and $T= L$.

The assertion on generalized semi-Clifford gates follows from the equivalence of (1) and (3) and the definition in Section~\ref{sec:intro}.

For the last assertion, the argument above with $\mathcal B$ in place of $\mathcal A$ gives $\mathcal B=\operatorname{span}_{\CC}\{P_v:v\in S\}$ with $S=\{v:P_v\in\mathcal B\}$. As $\mathcal B$ is a unital algebra and $P_vP_w\doteq P_{v+w}$, $S$ is a subspace, and it is isotropic if $\mathcal B$ is abelian. If $U\mathcal BU^\dagger=\A_L$, then $\mathcal B$ is a MASA, since $\A_L$ is one, so $S$ is Lagrangian by the implication from (2) to (3).
\end{proof}

Recall from Section~\ref{sec:symp} that a pair $(L,M)$ of Lagrangians is a \emph{generalized semi-Clifford Lagrangian pair} of $U$ if $U\A_LU^\dagger=\A_M$. We abbreviate this as \emph{GSC Lagrangian pair}, or \emph{GSC pair} when there is no ambiguity, write $\mathscr L(U)$ for the set of GSC pairs of $U$, and call $L$ and $M$ the \emph{input} and \emph{output Lagrangians} of $(L,M)$.

\begin{proof}[Proof of Lemma~\ref{lem:lagrangian}]
This is the assertion on generalized semi-Clifford gates in Proposition~\ref{lem:lagrangian-expanded}.
\end{proof}

For a gate $U$ and $v\in V_n$, we call $h_v \coloneqq UP_vU^\dagger$ the \emph{correction} of $U$ at $P_v$.

\begin{proof}[Proof of Lemma~\ref{lem:main-local-global}]
Suppose first that $U\in\mathsf{GSC}(n)$. By Lemma~\ref{lem:lagrangian}, there exist Lagrangian subspaces $L_0,M_0\leq \F_2^{2n}$ such that $U\A_{L_0}U^\dagger =\A_{M_0}$. For every $v\in \F_2^{2n}$ and every $w\in L_0$, $P_vP_wP_v^\dagger = (-1)^{\omega(v,w)}P_w$, so $P_v\A_{L_0} P_v^\dagger =\A_{L_0}$. Therefore, for $h_v=UP_vU^\dagger$,
\[
h_v\A_{M_0}h_v^\dagger=
UP_vU^\dagger\A_{M_0}UP_v^\dagger U^\dagger
=
UP_v\A_{L_0}P_v^\dagger U^\dagger
=
U\A_{L_0}U^\dagger
=
\A_{M_0}.
\]

Conversely, let $M\leq \F_2^{2n}$ be a Lagrangian subspace such that $h_v\A_Mh_v^\dagger=\A_M$ for every $v \in \F_2^{2n}$. For every $v\in V_n$,
\[
P_vU^\dagger\A_MUP_v^\dagger=
U^\dagger
(UP_vU^\dagger)\A_M(UP_v^\dagger U^\dagger)
U=
U^\dagger h_v\A_Mh_v^\dagger U=
U^\dagger\A_MU.
\]
Hence $\mathcal B\coloneqq U^\dagger\A_MU$ is invariant under conjugation by every Pauli. Moreover, $\mathcal B$ is unitarily conjugate to the Pauli MASA $\A_M$. By Proposition~\ref{lem:lagrangian-expanded}, there is a Lagrangian subspace $L\leq V_n$ such that
\[
U^\dagger\A_MU=\mathcal B=\A_L.
\]
Lemma~\ref{lem:lagrangian} therefore implies $U\in\mathsf{GSC}(n)$. This also proves the final assertion that, for any such $M$, $U^\dagger\A_MU=\A_L$ for some Lagrangian $L\leq V_n$.
\end{proof}

\begin{lemma}\label{lem:lagrangian-count}
The number of Lagrangians in $V_n$ is $N_n=\prod_{j=1}^n(2^j+1)$, which is odd.
\end{lemma}
\begin{proof}
Counting $(v,L)$ with $0\ne v\in L$, using the Lagrangians of $v^\perp/\Span\{v\}$, gives $(2^n-1)N_n=(2^{2n}-1)N_{n-1}$. Now $N_0=1$. See also~\cite[Lemma~9.4.1]{BCN1989}.
\end{proof}

Recall that the conductor of $U$, denoted $\operatorname{cond}(U)$, is the least $r\ge2$ such that $\lambda U$ has entries in $K_r$ for some $\lambda\in\CC^\times$. We call $U$ \emph{Gaussian} if $\operatorname{cond}(U)=2$. As before, $K_M=\Q(\zeta_{2^M})$.

\begin{proposition}\label{lem:conductor-plane}
Let $U$ be an $n$-qubit gate of finite conductor $r=\operatorname{cond}(U)$, and let $h_v=UP_vU^\dagger$ and $m=\max_{v\in V_n}\operatorname{cond}(h_v)$. Then
\[
 \max\{2,r-1\}\leq m\leq r.
\]
Every $U\in\C_k(n)$ with $k\geq2$ has finite conductor, and $2\leq\operatorname{cond}(U)\leq k$. Moreover, for every $n\geq1$ and $k\geq2$,
\[
 \{\operatorname{cond}(U):U\in\C_k(n)\}=\{2,\ldots,k\}.
\]
\end{proposition}
\begin{proof}
First, choose $W\doteq U$ with entries in $K_r$. Then $h_v=WP_vW^{-1}$ has entries in $K_r$, so $m\le r$, and $m\ge2$ by definition. Suppose $m\le r-2$, and write $h_v=c_vB_v$ with $B_v$ over $K_m$ and $c_v\in K_r^\times$. Since $h_v^2=I$, we have $c_v^2\in K_m$, so every $\sigma \in\operatorname{Gal}(K_r/K_m)$ satisfies $\sigma(h_v)=\pm h_v$ and $\sigma^2(h_v)=h_v$. As $m\ge2$, this Galois group is cyclic, and its squares form the subgroup of index two, with fixed field $K_{m+1}$. Hence every $h_v$ has entries in $K_{m+1}$. The linear system $XP_v=h_vX$, $v\in V_n$, then has coefficients in $K_{m+1}$, and its solutions over $\CC$ are the multiples of $U$, because a solution $X$ makes $U^\dagger X$ commute with every Pauli. A linear system has the same rank over $K_{m+1}$ as over $\CC$, so some $X\doteq U$ has entries in $K_{m+1}$. Then $r\le m+1\le r-1$, a contradiction. Therefore, $\max\{2,r-1\}\le m\le r$.

Next, we show that every Hermitian involution $g\in\C_\ell(n)$, $\ell\ge1$, has entries in $K_{\ell+1}$. Fix $v\in V_n$, and let $g_0=g$ and $g_{j+1}=g_jP_vg_j$. Each $g_j$ is a Hermitian involution, $g_j\in\C_{\ell-j}(n)$ for $j<\ell$ since $g_{j+1}=g_jP_vg_j^\dagger$, and $g_{j+1}P_v=(g_jP_v)^2$, so $g_jP_v=(gP_v)^{2^j}$. As $g_{\ell-1}$ is a Pauli up to sign, $(g_{\ell-1}P_v)^4=I$, and hence $(gP_v)^{2^{\ell+1}}=I$. Thus $\operatorname{tr}(gP_v)$ is a sum of $2^{\ell+1}$-th roots of unity, and $g=2^{-n}\sum_{v\in V_n}\operatorname{tr}(gP_v)P_v$ has entries in $K_{\ell+1}$.

Now let $U\in\C_k(n)$ with $k\ge2$. Each $h_v\in\C_{k-1}(n)$ is a Hermitian involution, so it has entries in $K_k$. As the $P_v$ span the matrices over $K_k$, the matrix $UAU^\dagger$ has entries in $K_k$ for every matrix $A$ over $K_k$. Choose $U_{ab}\ne0$, and let $E_{jb}$ denote the matrix units. Then $U_{ij}\overline{U_{ab}}=(UE_{jb}U^\dagger)_{ia}$ and $|U_{ab}|^2=(UE_{bb}U^\dagger)_{aa}$ lie in $K_k$ for all $i,j$, so $U/U_{ab}$ has entries in $K_k$, and $\operatorname{cond}(U)\le k$.

Finally, we show that every $s\in\{2,\dots,k\}$ is the conductor of some $U\in\C_k(n)$. Let $D_s=\operatorname{diag}(1,\zeta_{2^s})\otimes I^{\otimes(n-1)}$. Then $D_sX_1D_s^\dagger=\zeta_{2^s}^{-1}D_{s-1}X_1$, $D_sZ_1D_s^\dagger=Z_1$, and $D_s$ commutes with the Paulis on the other qubits, so every Pauli correction of $D_s$ is a scalar times $D_{s-1}^aP$ with $a\in\F_2$ and $P$ a Pauli. Since $D_2=S\otimes I^{\otimes(n-1)}$ is Clifford and right multiplication by a Pauli preserves each level, induction on $s$ gives $D_s\in\C_s(n)\subseteq\C_k(n)$. The ratio of the entries of $D_s$ is $\zeta_{2^s}$, which lies in $K_r$ exactly when $r\ge s$, so $\operatorname{cond}(D_s)=s$.
\end{proof}

\begin{proof}[Proof of Lemma~\ref{lem:basic-operations}]
A Clifford gate $C$ maps Paulis to Paulis up to sign, and $C\A_SC^\dagger=\A_{\gamma_CS}$. By induction on $j$, conjugation by $C$ preserves every $\C_j(n)$.

We use that $U\in\SC(n)$ exactly when some $(L,M)\in\mathscr L(U)$ has $L\subseteq\mathsf P(U)$: such an $L$ makes $U$ map the maximal abelian subgroup $\{\lambda P_v:v\in L,\ \lambda\in\mu_4\}$ onto a group of Paulis, which is again maximal abelian, and conversely (cf.~\cite[Proposition~5]{sc-ZCC2008}). Now $(L,M)\in\mathscr L(U)$ exactly when $(\gamma_{C_2}^{-1}L, \gamma_{C_1}M)\in\mathscr L(\lambda C_1UC_2)$, and $L\subseteq\mathsf P(U)$ exactly when $\gamma_{C_2}^{-1}L\subseteq\mathsf P(\lambda C_1UC_2)$. The correction of $\lambda C_1UC_2$ at a Pauli $P$ is $C_1hC_1^\dagger$, where $h\in\C_{k-1}(n)$ is the correction of $U$ at the Pauli $C_2PC_2^\dagger$, and hence $\lambda C_1UC_2\in\C_k(n)$ (cf.~\cite[Proposition~3]{sc-ZCC2008}). This proves~\eqref{itm:op1}.

Up to a scalar, every Pauli on $n+m$ qubits is $P\otimes Q$, and the correction of $U\otimes V$ at it is $UPU^\dagger\otimes VQV^\dagger$, with $UPU^\dagger\in\C_{\max\{k-1,1\}}(n)$ and $VQV^\dagger\in\C_{\max\{\ell-1,1\}}(m)$. Induction on $\max\{k,\ell\}$, starting from $\Pau(n)\otimes\Pau(m)=\Pau(n+m)$, gives $U\otimes V\in\C_{\max\{k,\ell\}}(n+m)$. If $(L,M)\in\mathscr L(U)$ and $(L',M')\in\mathscr L(V)$, then $(L\oplus L',M\oplus M')\in\mathscr L(U\otimes V)$, since $\A_{L\oplus L'}=\A_L\otimes\A_{L'}$, and $L\oplus L'\subseteq\mathsf P(U\otimes V)$ if $L\subseteq\mathsf P(U)$ and $L'\subseteq\mathsf P(V)$. A diagonal $V$ has the GSC pair $(L_Z,L_Z)$ with $L_Z=\Span\{z_1,\dots,z_m\}\subseteq\mathsf P(V)$, so $U\in\mathsf X(n)$ implies $U\otimes V\in\mathsf X(n+m)$.

Conversely, let $V$ be diagonal and $(L,M)\in\mathscr L(U\otimes V)$. Let $R=\Span\{z_{n+1},\dots,z_{n+m}\}$. For $w\in R^\perp$ we have $P_w=P_{\bar w}\otimes P_{w'}$ with $P_{w'}$ diagonal. For $y\in\F_2^m$ let $E_y=I\otimes\bra y$. Then $E_yP_wE_y^\dagger=\pm P_{\bar w}$ if $w\in R^\perp$, and $E_yP_wE_y^\dagger=0$ otherwise, so $E_y\A_LE_y^\dagger=\A_{\bar L}$ with $\bar L=\{\bar w:w\in L\cap R^\perp\}\cong(L\cap R^\perp)/(L\cap R)$. The space $\bar L$ is isotropic, since the $P_{w'}$ commute, and it has dimension $n$, since $(L\cap R^\perp)^\perp=L+R$. Thus $\bar L$ is Lagrangian, and likewise $\bar M$. As $E_y(U\otimes V)=V_{yy}UE_y$ with $|V_{yy}|=1$,
\[
U\A_{\bar L}U^\dagger=E_y(U\otimes V)\A_L(U\otimes V)^\dagger E_y^\dagger=E_y\A_ME_y^\dagger=\A_{\bar M},
\]
so $(\bar L,\bar M)\in\mathscr L(U)$. If moreover $L\subseteq\mathsf P(U\otimes V)$, then $\bar L\subseteq\mathsf P(U)$, since $(U\otimes V)P_w(U\otimes V)^\dagger=UP_{\bar w}U^\dagger\otimes P_{w'}$ for $w\in R^\perp$. This concludes the proof of~\eqref{itm:op2}. For the semi-Clifford case with $V=I$, see also~\cite[Appendix~A]{sc-HRT2025arxiv}.

Finally, if $U\in\C_k(n)\setminus\mathsf X(n)$, then $U\otimes I\in\C_k(N)\setminus\mathsf X(N)$ by~\eqref{itm:op2} with $V=I$, and $\C_k(N)\subseteq\C_r(N)$ for $r\ge k$.
\end{proof}

\begin{lemma}\label{lem:root-witness}
If $U\A_SU^\dagger=\A_T$ for isotropic $S,T$, and $U\A_LU^\dagger=\A_M$ for Lagrangians $L,M$, then
\[
U\A_{S+(L\cap S^\perp)}U^\dagger
=\A_{T+(M\cap T^\perp)}.
\]
Both spaces $S+(L\cap S^\perp)$ and $T+(M\cap T^\perp)$ are Lagrangian and contain $S$ and $T$, respectively.
\end{lemma}
\begin{proof}

For any $K$, let
\[
\A_K'
\coloneqq \{B:[B,P_k]=0\text{ for all }k\in K\}
=\operatorname{span}_{\mathbb C}\{P_v:v\in K^\perp\}.
\]
Hence
\[
\A_L\cap\A_S'=\A_{L\cap S^\perp},\qquad
\A_{S+(L\cap S^\perp)}
=\operatorname{alg}(\A_S,\A_L\cap\A_S'),
\]
where $\operatorname{alg}$ denotes the generated algebra. Conjugation by $U$ sends $\A_S,\A_L,\A_S'$ to $\A_T,\A_M,\A_T'$, respectively. Thus
\[
U\A_{S+(L\cap S^\perp)}U^\dagger
=\operatorname{alg}(\A_T,\A_M\cap\A_T')
=\A_{T+(M\cap T^\perp)}.
\]

Since $S$ and $L$ are isotropic and every element of $L\cap S^\perp$ is orthogonal to $S$, the space $S+(L\cap S^\perp)$ contains $S$ and is isotropic. Let $s=\dim S$ and $d=\dim(S\cap L)$. Because $L^\perp=L$, the map $L\to\operatorname{Hom}(S,\F_2)$, $\ell\mapsto\omega(\ell,\cdot)|_S$, has kernel $L\cap S^\perp$ and an image of dimension $s-\dim(S\cap L^\perp)=s-d$. Consequently,
\[
\dim(L\cap S^\perp)=n-s+d,\qquad
\dim\bigl(S+(L\cap S^\perp)\bigr)=s+(n-s+d)-d=n.
\]
Thus this space is Lagrangian. The same argument applies to $T$ and $M$.

\end{proof}

\begin{proposition}\label{lem:controlled-clifford}
For $C\in \C_2(m)$ and $k\ge2$, $\operatorname{diag}(I,C)\in\C_k(m+1)$ if and only if $C^{2^{k-2}}\in\mu_4\mathcal P(m)$.
\end{proposition}
This is the controlled jump criterion of Xu and Wang~\cite[Theorem~3]{XuWang2026ControlledJump}, and a necessary condition of this form is also given in~\cite[Corollary~2.4.1]{AW2024ControlledGates}. We include a proof for completeness.
\begin{proof}

Let $D(C)=\diag(I,C)$, $X_c=X\otimes I$, and $Z_c=Z\otimes I$. First consider $k=2$:
\[
D(C)X_cD(C)^\dagger
=|0\rangle\langle1|\otimes C^\dagger
 +|1\rangle\langle0|\otimes C.
\]
If $D(C)$ is Clifford, this correction is a Hermitian Pauli anticommuting with $Z_c$. It must equal $\pm X\otimes P_p$ or $\pm Y\otimes P_p$ for some $p\in V_m$. Comparing the lower-left blocks gives $C=\xi P_p$ with $\xi\in\mu_4$. Conversely, for $C=\xi P_p$ and any $q\in V_m$,
\[
\begin{aligned}
D(C)X_cD(C)^\dagger
&=(\bar\xi|0\rangle\langle1|+\xi|1\rangle\langle0|)
  \otimes P_p\in\mathcal P(m+1),\\
D(C)(I\otimes P_q)D(C)^\dagger
&=Z_c^{\omega(p,q)}(I\otimes P_q)\in\mathcal P(m+1).
\end{aligned}
\]
Also $D(C)Z_cD(C)^\dagger=Z_c$. Since these cover all the Pauli generators, $D(C)\in\C_2$ if and only if $C\in\mu_4\mathcal P(m)$.

Now let $k>2$. For every target Pauli $Q$, the Clifford property of $C$ gives
\[
D(C)(I\otimes Q)D(C)^\dagger
=(I\otimes Q)D(Q^\dagger CQC^\dagger)\in\C_2(m+1),
\]
since $Q^\dagger CQC^\dagger$ is Pauli. The control correction is
\[
D(C)X_cD(C)^\dagger=(I\otimes C^\dagger)D(C^2)X_c.
\]
Every other Pauli correction is a product of these. Multiplication by a Clifford preserves $\C_{k-1}$, so $D(C)\in\C_k\Longleftrightarrow D(C^2)\in\C_{k-1}$. Iterating this equivalence and applying the base case gives
\[
D(C)\in\C_k
\Longleftrightarrow D(C^{2^{k-2}})\in\C_2
\Longleftrightarrow C^{2^{k-2}}\in\mu_4\mathcal P(m). \qedhere
\]
\end{proof}

\section{\texorpdfstring{$\mathcal{C}_3(n)\subseteq\mathsf{SC}(n)$ for $n\le6$}{C3(n) in SC(n) for n<=6}}\label{app:sc-boundary}

We now show that every third-level gate on at most six qubits is semi-Clifford. We note that this result is sharp in the sense that there exists a seven-qubit third-level gate that is not semi-Clifford~\cite{sc-BS2010}.

Recall that by~\cite{sc-BS2010} and~\cite{sc-HRT2025arxiv}, a third-level gate on $n\le6$ qubits agrees, up to Clifford gates on both sides, with a gate $M\ket{x,y}= e^{i\theta(x,y)}\ket{x,y+q(x)}$ on $r$ control and $m$ target qubits, where $r+m=n$ and $q$ is quadratic. Without the phase $\theta$, Hadamard gates on the targets make $M$ diagonal, so only $\theta$ can keep $M$ from being semi-Clifford. We begin with three lemmas.

\begin{lemma}\label{lem:sc-sign-gates}
For $g:\F_2^r\to\F_2$ let $D_g=\sum_x(-1)^{g(x)}\ket{x}\bra{x}$. Then the following hold.
\begin{enumerate}
\item $D_g$ is Pauli if and only if $g$ is affine. \label{itm:scsg1}
\item $D_g\in\C_2(r)$ if and only if $\deg g\le2$.
\end{enumerate}
\end{lemma}
\begin{proof}
Suppose first $D_g$ is Pauli, i.e., $D_g = \lambda X^a Z^b$ with $\lambda \in \{\pm 1, \pm i\}$. Since $D_g$ is Hermitian diagonal in the computational basis, $a = 0$ and there exists $c\in\F_2$ such that $\lambda =(-1)^c$. Thus, $D_g \ket x = (-1)^{g(x)}\ket x = (-1)^{b \cdot x +c}\ket x$ for all $x$. Hence $g(x) = b \cdot x +c$ is affine. Conversely, $D_{b\cdot x+c}=(-1)^cZ^b$.

Next, since $D_g$ commutes with every $Z^b$, it is Clifford if and only if $D_gX^aD_g^\dagger$ is Pauli for every $a$. Notice that
\[
D_gX^aD_g^\dagger
=X^a\sum_x(-1)^{g(x+a)+g(x)}\ket{x}\bra{x}.
\]
By~\eqref{itm:scsg1}, this is Pauli for every $a$ if and only if $g(x+a)+g(x)$ is affine for every $a$, i.e., when $\deg g\leq2$.
\end{proof}

\begin{lemma}\label{lem:sc-forms}
Let $R=\F_2^r$. Then the following hold.
\begin{enumerate}
\item Let $\phi:R\times R\to\F_2^{2m}$ be alternating bilinear with $\omega(\phi(b,a),\phi(b,c))=0$ for all $a,b,c\in R$. If $r\le3$ or $m\le2$, then $W_\phi=\Span\{\phi(a,b):a,b\in R\}$ is isotropic.
\item Let $W=\F_2^m$ and let $\beta:R\times R\to W$ be alternating bilinear with values spanning $W$; write $G_b\in\operatorname{Hom}(R,W)$ for $G_ba=\beta(b,a)$. Let $w\mapsto E_w\in\operatorname{Hom}(W,R)$ be linear, with $E_wG_b$ symmetric for all $w\in W$, $b\in R$, and $E_ww'=E_{w'}w$ for all $w,w'\in W$. If $r+m\le6$, then $E_w=0$ for all $w$.
\end{enumerate}
\end{lemma}
\begin{proof}
Since $\phi$ is alternating and $\omega(\phi(b,a),\phi(b,c))=0$, the $4$-form
\[
\tau(a,b,c,d)\coloneqq \omega(\phi(a,b),\phi(c,d)) \quad a,b,c,d \in \F^{r}_2
\]
is alternating. Therefore, if $\tau (a,b,c,d) = 1$, then $a,b,c,d$ are linearly independent over $\F_2$ and hence $r \geq 4$. Consequently, for
\[
(p_1,\ldots,p_6)
=\bigl(\phi(a,b),\phi(c,d),\phi(a,c), \phi(b,d),\phi(a,d),\phi(b,c)\bigr),
\]
we obtain that
\[
\bigl(\omega(p_i,p_j)\bigr)_{i,j=1}^6
=\begin{pmatrix}0&1\\1&0\end{pmatrix}^{\!\oplus3}.
\]
This matrix is nonsingular, so $p_1,\dots,p_6$ are linearly independent, and $2m\geq6$. Hence $\tau\ne0$ forces $r\ge4$ and $m\ge3$, and under the hypothesis $\tau=0$, so $\omega(x,y) = 0$ for all $x,y \in W_\phi$ and $W_\phi$ is isotropic.

Next, for $E\in\operatorname{Hom}(W,R)$, define $\Phi_\beta(E)\in R^*\otimes\wedge^2R^*$ by
\[
\Phi_\beta(E)(x,a,b)=x^{\mathsf T}E\beta(a,b).
\]
Then $\Phi_\beta$ is injective by surjectivity of $\beta$. Consider $\alpha_w(x,a,b) \coloneqq x^{\mathsf T}E_w\beta(a,b) = \Phi_\beta(E_w)(x,a,b)$. By the bilinearity of $\beta$, we have that $\alpha_w$ is trilinear. Moreover, since $\beta(a,a)=0$ and $E_wG_a$, $E_wG_b$ are symmetric,
\[
\alpha_w(x,a,a)=0, \quad
\alpha_w(a,a,b)
=
a^{\mathsf T}E_wG_ab
=
b^{\mathsf T}E_wG_aa
=0, \quad
\alpha_w(b,a,b)
=0 \Longrightarrow \alpha_w\in\wedge^3R^*.
\]

Suppose for contradiction that $E_w\neq0$. Since $\Phi_\beta$ is injective, $\alpha_w=\Phi_\beta(E_w)$ is a nonzero alternating three-form. Consider its contraction map
\[
\kappa_{\alpha_w}:\wedge^2R\longrightarrow R^*,
\qquad
\kappa_{\alpha_w}(a\wedge b)(x)
=
\alpha_w(x,a,b).
\]
Identifying $R\cong R^*$, it is straightforward to check that $\kappa_{\alpha_w}$ is the composition
\[
\wedge^2R
\xrightarrow{\ \beta\ }
W
\xrightarrow{\ E_w\ }
R
\cong R^*.
\]
Since $\beta$ is surjective, $\rank \kappa_{\alpha_w}=\rank E_w$. Let $x,a,b\in R$ such that $\alpha_w(x,a,b)=1$. Then the three functionals $\kappa_{\alpha_w}(a\wedge b)$, $\kappa_{\alpha_w}(b\wedge x)$, and $\kappa_{\alpha_w}(x\wedge a)$ are linearly independent. Therefore $\rank E_w = \rank\kappa_{\alpha_w} \geq3$. Since $E_w:W\to R$, we have $r\geq3$ and $m\geq3$. Since $r+m\leq6$, $r=m=3$.

Since $\dim (\Phi_\beta^{-1}(\wedge^3R^*)) \leq 1$ as $\dim\wedge^3R^*=1$ for $r=3$, linearity gives $E_w=\lambda(w)E_*$ for a nonzero $\lambda \in W^*$ and $E_*$ invertible. Choose $w$ with $\lambda(w)=1$ and nonzero $w'\in\ker\lambda$. Then $E_ww'=E_{w'}w$ gives $E_*w'=0$, a contradiction. Therefore, $E_w = 0$.
\end{proof}

\begin{lemma}\label{lem:sc-pauli-trivial}
If $Q\in\mathrm U(2^m)$ and $QP_vQ^\dagger\doteq P_v$ for every $v\in\F_2^{2m}$, then $Q\doteq P$ for a Pauli $P$.
\end{lemma}
\begin{proof}
Since $QX_jQ^\dagger$ and $QZ_jQ^\dagger$ are Hermitian with square $I$,
\[
QX_jQ^\dagger=\epsilon_jX_j,
\qquad
QZ_jQ^\dagger=\delta_jZ_j \qquad \text{where } \epsilon_j,\delta_j\in\{\pm1\}.
\]
Choose $a,b\in\F_2^m$ such that $(-1)^{b_j}=\epsilon_j$ and $(-1)^{a_j}=\delta_j$ for every $j$, and let $P=X^aZ^b$. Then $P$ has the same conjugation action as $Q$ on every $X_j$ and $Z_j$. Hence $P^\dagger Q$ commutes with all Pauli operators, and therefore with the full matrix algebra. It follows that $P^\dagger Q$ is scalar.
\end{proof}
The main tool is the following proposition on block-diagonal gates whose blocks are scalar multiples of Pauli gates. We reduce $M$ to this form in the proof of Theorem~\ref{thm:sc-six}.
\begin{proposition}\label{prop:sc-controlled-pauli}
Let $r,m\geq0$, $v:\F_2^r\to\F_2^{2m}$, and $\zeta:\F_2^r\to\T$. Consider
\[
N\coloneqq \sum_{x\in\F_2^r} \zeta(x) \ket{x}\bra{x} \otimes P_{v(x)}.
\]
If $N\in\C_3(r+m)$ and either $r\leq3$ or $m\leq2$, then $N \in \mathsf{SC}(r+m)$.
\end{proposition}
\begin{proof}
Let $u\in\F_2^{2m}$. For
\[
D_u\coloneqq \sum_x(-1)^{\omega(v(x),u)} \ket{x}\bra{x},
\]
we have that
\[
N(I\otimes P_u)N^\dagger=D_u\otimes P_u \in \mathcal{C}_2(r+m).
\]
Thus, $D_u\otimes I = (I\otimes P_u)(D_u\otimes P_u) \in \mathcal{C}_2(r+m)$ and hence $D_u\in\C_2(r)$. By Lemma~\ref{lem:sc-sign-gates}(2), $x\mapsto\omega(v(x),u)$ is quadratic.

Next, we show that we may without loss of generality assume $v(x)=\sum_{i<j}x_ix_jv_{ij}$. Since $\omega$ is nondegenerate, we have
\[
v(x)=v_0+\ell(x)+\sum_{i<j}x_ix_jv_{ij},
\]
where $\ell:\F_2^r\to\F_2^{2m}$ is linear. Let $A$ be the product of $I\otimes P_{v_0}$ and the gates $\sum_{x_i\in\F_2}\ket{x_i}\bra{x_i}\otimes P_{\ell(e_i)}^{x_i}$, $1\le i\le r$, each controlled by one qubit and Clifford by Proposition~\ref{lem:controlled-clifford} with $k=2$. Then $A=\sum_x\lambda(x)\ket{x}\bra{x}\otimes P_{v_0+\ell(x)}$ with $\lambda(x)\in\mu_4$, so $A^\dagger N$ is again of the form of $N$, with labels $v(x)+v_0+\ell(x)$ and phases absorbed into $\zeta$, and by Lemma~\ref{lem:basic-operations}(1) it satisfies the hypotheses exactly when $N$ does and is semi-Clifford exactly when $N$ is. We may therefore assume $v_0=0$ and $\ell=0$. Let
\[
\phi(a,b) \coloneqq v(a+b)+v(a)+v(b),\qquad
W_\phi \coloneqq \Span\{\phi(a,b):a,b\in\F_2^r\}.
\]
Since $v_{ij}=\phi(e_i,e_j)$, $v(x)\in W_\phi$ for every $x$.

For $b\in\F_2^r$, let $K_b \coloneqq N(X^b\otimes I)N^\dagger \in \mathcal{C}_2(r+m)$. Explicitly, $K_b=\sum_xc_b(x)\ket{x+b}\bra{x}\otimes P_{v(x+b)}P_{v(x)}$ with $c_b(x)\in\T$, and $P_{v(x+b)}P_{v(x)}\doteq P_{v(b)+\phi(b,x)}$. Consider $\F_2^r \oplus \F_2^{2m} \oplus \F_2^r$, where a triple $(a,u,c)$ represents, up to phase, the Pauli $X^aZ^c\otimes P_u$. Since $K_b \in \mathcal{C}_2(r+m)$, $K_bP_wK_b^\dagger\doteq P_{F_b(w)}$ for all $w \in \F_2^{2(r+m)}$. Hence $F_b$ is a linear map which, with respect to the decomposition $\F_2^r\oplus\F_2^{2m}\oplus\F_2^r$, is represented by the matrix
\[
F_b=
\begin{pmatrix}
I&0&0\\
\phi_b&I&0\\
\varepsilon_b&\phi_b^*&I
\end{pmatrix},
\qquad
\phi_b(a)=\phi(b,a),
\]
where $\phi_b^*$ is the adjoint of $\phi_b$ for $\omega$, that is, $(\phi_b^*u)\cdot a=\omega(\phi(b,a),u)$, and the block $\varepsilon_b$ is not needed. Since $K_b^2=I$, $F_b^2=I$ and therefore $\phi_b^*\phi_b=0$ and $\omega(\phi(b,a),\phi(b,c))=0$. By Lemma~\ref{lem:sc-forms}(1), $W_\phi$ is isotropic.

Every isotropic subspace lies in a Lagrangian, so $W_\phi\subseteq L$ for a Lagrangian $L\leq \F^{2m}_2$. Since $\Sp(2m,2)$ acts transitively on Lagrangian subspaces and every binary symplectic map is induced by a Clifford (the group $\Sp(2m,2)$ is generated by the transvections $t_p(v)=v+\omega(v,p)p$, and $t_p$ is induced by the Clifford $(I+iP_p)/\sqrt2$), choose $C\in\C_2(m)$ whose symplectic action sends $L$ to $L_Z=\{(0,z):z\in\F_2^m\}$. Then $(I\otimes C)N(I\otimes C^\dagger)$ is diagonal in the computational basis and hence semi-Clifford, and $N \in \mathsf{SC}(r+m)$ follows from Lemma~\ref{lem:basic-operations}(1).
\end{proof}

We are now ready to prove the main theorem of this section.

\begin{theorem}\label{thm:sc-six}
For $1\leq n\leq6$, $\C_3(n)\subseteq\SC(n)$.
\end{theorem}
\begin{proof}

Let $U\in\C_3(n)$. By~\cite{sc-BS2010} and He--Robitaille--Tan~\cite[Lemma~3.8, Appendix~A, and Corollary~B.6]{sc-HRT2025arxiv}, $U\doteq C_1MC_2$ for some $C_1,C_2\in\C_2(n)$, so $M\in\C_3(n)$ by Lemma~\ref{lem:basic-operations}(1), and it suffices to prove $M \in \mathsf{SC}(n)$ where
\begin{equation}\label{eq:sc-quadratic-permutation}
M\ket{x,y}=e^{i\theta(x,y)}\ket{x,y+q(x)},\qquad
x\in\F_2^r,\quad y\in\F_2^m,\quad r+m=n,
\end{equation}
where $\theta:\F_2^r\times\F_2^m\to\mathbb R$ and $q(x)=\sum_{i<j}x_ix_jq_{ij}$ with $q_{ij}\in\F_2^m$. Note that $\beta(a,b)=q(a+b)+q(a)+q(b)$ is an alternating bilinear map. Let $S=\Span\{\beta(a,b):a,b\in\F_2^r\} = \Span\{q_{ij}: i < j\}$ and $\F_2^m=S\oplus T$. In a basis of $\F_2^m$ adapted to $S\oplus T$, write $y=(y_S,y_T)$; this change of basis of the second register is a CNOT circuit, hence Clifford, so by Lemma~\ref{lem:basic-operations}(1) we may use it. Since $q(x)\in S$,~\eqref{eq:sc-quadratic-permutation} becomes $M\ket{x,y_S,y_T}=e^{i\theta}\ket{x,y_S+q(x),y_T}$, which does not change $y_T$, so $(x,y_T)$ may serve as the first register, with the same $q$. Thus, after redefining $r$ and $m$ without changing $r+m=n$, we may assume that the values of $\beta$ span $\F_2^m$. If $m=0$, then $M$ is diagonal and the result follows.

Assume $m>0$. It remains only to prove the following claim.

\noindent\emph{Claim.}
There exists $C\in\C_2(m)$ such that
\begin{equation}\label{eq:sc-pauli-field}
M(I\otimes C)
=\sum_{x\in\F_2^r}\zeta(x)\ket{x}\bra{x}\otimes P_{v(x)},\qquad
v(x)\in\F_2^{2m},\quad\zeta(x)\in\T.
\end{equation}

Indeed, if the claim holds, then $r+m=n\leq6$ implies $r\leq3$ or $m\leq2$, so $M(I\otimes C)$ is semi-Clifford by Proposition~\ref{prop:sc-controlled-pauli}. Since $U$, $M$, and $M(I\otimes C)$ differ by left and right Clifford multiplication, $U \in \mathsf{SC}(n)$ by Lemma~\ref{lem:basic-operations}.

We now prove the claim. Let $R=\F_2^r$, $W=\F_2^m$, and $M=\sum_x\ket{x}\bra{x}\otimes M_x$. Regard $\beta$ as a surjection $\wedge^2R\to W$. For $w\in W$, since $M \in \mathcal{C}_3(n)$, $h_w= M(I\otimes X^w)M^\dagger \in \mathcal{C}_2(n)$. A direct calculation in the computational basis shows that $h_w = D_w(I\otimes X^w)$, where $D_w \in \mathcal{C}_2(n)$ is a diagonal gate. As $D_w$ is a diagonal Clifford gate, there is $B_w \in \F_2^{(r+m) \times (r+m)}$ such that
\[
D_wX^uD_w^\dagger\doteq X^uZ^{B_wu},\qquad
B_w=\begin{pmatrix}A_w&E_w\\E_w^{\mathsf T}&F_w\end{pmatrix}.
\]
The matrix $B_w$ is symmetric, since $X^u$ and $X^{u'}$ commute and so do their images, which gives $u^{\mathsf T}B_wu'=u'^{\mathsf T}B_wu$. For $b\in R$, let $g_b\coloneqq M(X^b\otimes I)M^\dagger$. A direct computation gives
\[
g_b\ket{x,y}
\doteq
\ket{x+b,\;y+q(x)+q(x+b)}
=
\ket{x+b,\;y+q(b)+\beta(b,x)}.
\]
Its linear part is therefore $(a,w)\mapsto \bigl(a,\;w+\beta(b,a)\bigr)$, which, relative to the decomposition $R\oplus W$, is represented by
\[
L_b=
\begin{pmatrix}
I&0\\
G_b&I
\end{pmatrix}, \quad G_b: R \to W, \quad G_b a \coloneqq \beta(b,a).
\]
Since $h_wh_{w'}=h_{w+w'}$ and $g_b$, $h_w$, $h_{w'}$ pairwise commute,
\[
B_{w+w'}=B_w+B_{w'},\qquad
L_b^{\mathsf T}B_wL_b=B_w,\qquad
B_w\bar w'=B_{w'}\bar w,
\]
where $\bar w=(0,w)\in R\oplus W$. Indeed, $D_{w+w'}=D_w\cdot X^{\bar w}D_{w'}X^{\bar w}$, and a shift does not change the matrix of a diagonal Clifford. The gate $g_b$ is an affine permutation with linear part $L_b$ times a diagonal Clifford, so conjugation by $g_b$ fixes $X^{\bar w}$ up to a diagonal Pauli and replaces $B_w$ by $L_b^{\mathsf T}B_wL_b$. Finally, $h_wh_{w'}=h_{w'}h_w$ equates $Z^{B_{w'}\bar w}$ and $Z^{B_w\bar w'}$ up to a scalar. In particular, $w\mapsto E_w$ is linear.

Since
\[
L_b^{\mathsf T}B_wL_b
=
\begin{pmatrix}
A_w+E_wG_b+G_b^{\mathsf T}E_w^{\mathsf T}
   +G_b^{\mathsf T}F_wG_b
&
E_w+G_b^{\mathsf T}F_w
\\
E_w^{\mathsf T}+F_wG_b
&
F_w
\end{pmatrix} = B_w = \begin{pmatrix}A_w&E_w\\E_w^{\mathsf T}&F_w\end{pmatrix},
\]
we obtain that
\[
F_wG_b=0, \quad \text{and } E_wG_b+G_b^{\mathsf T}E_w^{\mathsf T}=0.
\]
Consequently, $F_w=0$ since the images of the $G_b$ span $W$. The second identity says that $E_wG_b$ is symmetric for every $b\in R$, and $B_w\bar w'=B_{w'}\bar w$ with $F_w=0$ gives $E_ww'=E_{w'}w$. Lemma~\ref{lem:sc-forms}(2) gives $E_w=0$.

Therefore, $B_w(0,a)=0$ for all $a,w\in W$. Fix $w\in W$. Since $D_w$ is a diagonal Clifford and $B_w(0,a)=0$, we have that for $a\in W$, there exists $\chi_w(a)\in\{\pm1\}$ such that
\begin{equation}\label{eq:conj}
D_w(I\otimes X^a)D_w^\dagger
=
\chi_w(a)(I\otimes X^a).
\end{equation}
Since $X^{a+a'}=X^aX^{a'}$, $\chi_w$ is a character, so $\chi_w(a)=(-1)^{t_w\cdot a}$ for some $t_w\in W$. Let
\[
D_w
=
\sum_{x\in R}\sum_{y\in W}
d_w(x,y)\ket{x,y}\bra{x,y},
\qquad
d_w(x,y)\in\mathbb T.
\]
By~\eqref{eq:conj}, $d_w(x,y+a) = (-1)^{t_w\cdot a}d_w(x,y)$ for $x \in R$ and $y, a \in W$. Taking $y=0$ and $a=y$ yields $d_w(x,y) = d_w(x,0)(-1)^{t_w\cdot y}$ and hence
\[
D_w
=
\sum_{x\in R}
d_w(x,0)\ket{x}\bra{x}\otimes Z^{t_w}.
\]

On the other hand,
\[
D_w(I\otimes X^w) = h_w
=
M(I\otimes X^w)M^\dagger
=
\sum_{x\in R}
\ket{x}\bra{x}\otimes
M_xX^wM_x^\dagger.
\]
Comparing the $x$-blocks gives
\[
M_xX^wM_x^\dagger
=
d_w(x,0)Z^{t_w}X^w
\doteq X^wZ^{t_w}
\qquad \text{for } x\in R,
\]
where the right hand side is independent of $x$ up to global phase. Similarly, a direct computation from $M_x\ket{y} = e^{i\theta(x,y)}\ket{y+q(x)}$ gives
\[
M_xZ^wM_x^\dagger
=
(-1)^{w\cdot q(x)}Z^w,
\]
so $M_xZ^wM_x^\dagger$ is also independent of $x$ up to a global phase. Since the operators $X^w$ and $Z^w$, $w\in W$, generate the $m$-qubit Pauli group, it follows that $M_x\in\mathcal C_2(m)$ for every $x \in R$. Moreover, for every $x\in R$ and $w\in W$,
\[
M_xX^wM_x^\dagger
\doteq X^wZ^{t_w}
\doteq M_0X^wM_0^\dagger, \quad M_xZ^wM_x^\dagger
\doteq Z^w
\doteq M_0Z^wM_0^\dagger.
\]
Since the $X^w$ and $Z^w$ generate the Pauli group, it follows that $M_xPM_x^\dagger \doteq M_0PM_0^\dagger$. Therefore, for $Q_x\coloneqq M_xM_0^\dagger$, we have
\[
Q_xP_vQ_x^\dagger\doteq P_v
\qquad
\text{for every }v\in\F_2^{2m}.
\]
By Lemma~\ref{lem:sc-pauli-trivial}, there exist $\zeta(x)\in\mathbb T$ and $v(x)\in\F_2^{2m}$ such that $M_xM_0^\dagger = \zeta(x)P_{v(x)}$.

Finally, letting $C\coloneqq M_0^\dagger\in\mathcal C_2(m)$ gives
\[
M(I\otimes C)=
\sum_{x\in R}
\ket{x}\bra{x}\otimes M_xM_0^\dagger=
\sum_{x\in R}
\zeta(x)\ket{x}\bra{x}\otimes P_{v(x)}. \qedhere
\]
\end{proof}

\section{\texorpdfstring{$\mathcal{C}_4(n)\subseteq\mathsf{GSC}(n)$}{C4(n) in GSC(n)}}
\label{app:finite-group}

In this section we work in $\mathrm{PU}(2^n)$, with the projective Pauli group $\overline{\Pau}(n)=\Pau(n)/\langle iI\rangle$. For a gate $U$, let
\[
\Gamma(U) \coloneqq \langle\overline{\Pau}(n),U\overline{\Pau}(n)U^\dagger\rangle
\subseteq\mathrm{PU}(2^n).
\]

\begin{theorem}\label{thm:two-group-chamber}
Every finite $2$-subgroup of $\mathrm{PU}(2^n)$ containing $\overline{\Pau}(n)$ preserves a Pauli MASA\@. Consequently, if $\Gamma(U)$ is a finite $2$-group, then $U$ is generalized semi-Clifford.
\end{theorem}
\begin{proof}
Let $\Omega_n=\{\A_L:L\leq V_n\text{ Lagrangian}\}$ be the set of Pauli MASAs. By Lemma~\ref{lem:lagrangian-count}, $|\Omega_n|=N_n=\prod_{j=1}^n(2^j+1)$ is odd.

For any finite projective $2$-group $G\subseteq \mathrm{PU}(2^n)$ containing $\overline{\Pau}(n)$, set
\[
\mathcal X=\{g\mathcal A g^\dagger:g\in G,\ \mathcal A\in\Omega_n\}.
\]
By Proposition~\ref{lem:lagrangian-expanded} we have $\mathcal X^{\overline{\Pau}(n)}=\Omega_n$. Since every nontrivial $2$-group orbit has even size,
\[
|\mathcal X^G|\equiv|\mathcal X|
\equiv|\mathcal X^{\overline{\Pau}(n)}|=N_n\equiv1\pmod2.
\]
Thus $G$ must fix a Pauli MASA, which we denote by $\mathcal A$.

Now take $G=\Gamma(U)$ and let $\A=\A_M$ be the Pauli MASA fixed by $G$. For every Pauli $P$, the gate $UPU^\dagger\in G$ fixes $\A$, so $PU^\dagger\A UP^\dagger=U^\dagger\A U$. By Proposition~\ref{lem:lagrangian-expanded}, $U^\dagger\A U=\A_L$ with $L$ Lagrangian, which completes the proof.
\end{proof}

The next theorem implies Theorem~\ref{thm:main-k4}.

\begin{theorem}\label{thm:level-four-gsc}
For every $n$ and every $U\in\mathcal C_4(n)$, the group $\Gamma(U)$ is a finite $2$-group. In particular,
\[
\mathcal C_4(n)\subseteq\mathsf{GSC}(n).
\]
Furthermore, if $U^2= \lambda I$ for some $\lambda \in \mathbb{T}$, then there is a Pauli MASA $\mathcal A$ such that $U \mathcal A U^\dagger=\mathcal A$.
\end{theorem}
\begin{proof}
For $U\in\mathcal C_4(n)$, let
\[
H=U\overline{\Pau}(n)U^\dagger,\qquad
G=\langle h\overline{\Pau}(n)h^\dagger:h\in H\rangle.
\]
We have that $H\subseteq \mathcal{C}_3(n)$ is a $2$-group of order $4^n$. Therefore, we have $h\overline{\Pau}(n) h^\dagger\subseteq \mathcal C_2(n)$. Hence $G \subseteq \mathcal{C}_2(n)/\mathbb{T}$. Moreover, $I\in H$, so $\overline{\Pau}(n)\leq G$. Since $g\overline{\Pau}(n) g^\dagger=\overline{\Pau}(n)$ for every $g \in \mathcal{C}_2(n)$, $\overline{\Pau}(n) \trianglelefteq G$.

For $k,h\in H$, $k\bigl(h\overline{\Pau}(n)h^\dagger\bigr)k^\dagger = (kh)\overline{\Pau}(n)(kh)^\dagger$. Since $kh\in H$, $kGk^\dagger=G$, so $H$ normalizes $G$. Hence $\Gamma(U)=\langle\overline{\Pau}(n),H\rangle\leq GH$. Conversely, each subgroup $h\overline{\Pau}(n)h^\dagger$ generating $G$ is contained in $\langle\overline{\Pau}(n),H\rangle=\Gamma(U)$, so $G\leq\Gamma(U)$. Therefore $GH=\Gamma(U)$.

Since $H$ normalizes $G$, conjugation by $H$ acts on the set $\operatorname{Syl}_2(G)$ of Sylow $2$-subgroups of $G$. Since $H$ is a $2$-group, every $H$-orbit on $\operatorname{Syl}_2(G)$ has cardinality a power of two. By Sylow's theorem, $|\operatorname{Syl}_2(G)|$ is odd, so $H$ fixes a Sylow $2$-subgroup $R$ of $G$. Since every Sylow $2$-subgroup contains the normal $2$-subgroup $\overline{\Pau}(n)$,
\[
h\overline{\Pau}(n)h^\dagger\subseteq hRh^\dagger = R \quad
\text{for all $h\in H$} \quad \Longrightarrow \quad
G=R.
\]
Since $G\trianglelefteq GH$, $\Gamma(U)=GH$ is a finite $2$-group. By Theorem~\ref{thm:two-group-chamber}, $U \in \mathsf{GSC}(n)$.

Finally, suppose that $U^2=\lambda I$ for $\lambda\in\mathbb T$. Then $[U]^2=I \in \mathrm{PU}(2^n)$. Moreover,
\[
U\overline{\Pau}(n)U^\dagger=H, \quad UHU^\dagger
=
U^2\overline{\Pau}(n)(U^\dagger)^2
=
\overline{\Pau}(n).
\]
Hence $U\Gamma(U)U^\dagger=\Gamma(U)$. Therefore, $\langle\Gamma(U),U\rangle = \Gamma(U)\langle U\rangle$ is a finite $2$-group with $[\langle\Gamma(U),U\rangle:\Gamma(U)]\leq2$. By Theorem~\ref{thm:two-group-chamber}, $\langle\Gamma(U),U\rangle$ fixes a Pauli MASA $\mathcal A$. Since $U\in\langle\Gamma(U),U\rangle$, $U\mathcal AU^\dagger=\mathcal A$ as desired.
\end{proof}

\section{A common strategy}\label{app:common-strategy}

This section collects the tools used in Appendices~\ref{app:three-qubit} and~\ref{app:four-qubit-results}.

The case analyses there keep returning to one question: given gates $U_0, \dots, U_s$, find a common GSC pair $(L,M)\in\bigcap_{i=0}^s\mathscr L(U_i)$. We use four techniques, combined as in Fig.~\ref{fig:map}.
\begin{itemize}
    \item \phantomsection\label{tool:U}\emph{(U) Uniqueness.} If for some $U_0$ there exists a unique GSC pair $\mathscr L(U_0)=\{(L,M)\}$ and $\mathscr L(U_0)\cap\mathscr L(U_i)\ne\varnothing$ for every $i$, then $(L,M)\in\bigcap_i\mathscr L(U_i)$.
    \item \phantomsection\label{tool:P}\emph{(P) Parity.} Suppose $U_0,\dots,U_s\in\C_2(n)$ and write $\gamma_i=\gamma_{U_i}$. They share a GSC pair if and only if the maps $\gamma_0^{-1}\gamma_i$ fix a common Lagrangian $L$. In this case, $(L,\gamma_0L)\in\bigcap_i\mathscr L(U_i)$.
\item \phantomsection\label{tool:G}\emph{(G) Galois.} Let $\sigma\in \mathrm{Gal}(K_r/\Q(i))$ and $U$ be a scalar multiple of a matrix over $K_r$. We call $U^{-1}\sigma(U)$, which is defined up to scalar, the \emph{Galois cocycle} of $U$ at $\sigma$ (Section~\ref{sec:gsc}). Then for every $(L,M)\in\mathscr L(U)$,
\[
U^{-1}\sigma(U)\in\A_L,\qquad \sigma(U)U^{-1}\in\A_M,
\]
which we show in Proposition~\ref{lem:tool-galois}(2). This Galois condition does not produce a common GSC pair by itself, but it constrains each $U_i$ and $\mathscr{L}(U_i)$ separately, for use with the other tools. We illustrate this strategy in Example~\ref{ex:strategy-galois}.
\item \phantomsection\label{tool:F}\emph{(F) Finite certificate.} There are two steps to make this question a finite computational check. First, as we show in Lemma~\ref{lem:tool-transport}(1), for $C_1,C_2\in\C_2(n)$,
\[
(L,M)\in\mathscr L(U)\iff(\gamma_{C_2}^{-1}L,\gamma_{C_1}M)\in\mathscr L(C_1UC_2).
\]
The size of $\mathscr L(U)$, and whether $U_0,\dots,U_s$ share a GSC pair, remains invariant when every $U_i$ is replaced by $C_1U_iC_2$ with the same $C_1,C_2$. Second, up to such multiplication, the gates that occur have finitely many normal forms, as shown in Proposition~\ref{lem:tool-monomial} and Proposition~\ref{app:lem-R_n}. A property that is invariant under multiplication by Clifford gates can therefore be decided on finitely many representatives by exhaustive computation. The finite facts \textup{(\ref{cert:stationary-cs-ccz-gates-and-complete-affine-flags})--(\ref{cert:block-pair-domains})} obtained this way are stated in boxes where they are used, and each names the folder of the companion repository \certrepo{} that recomputes it in exact arithmetic. Each either settles a case directly or supplies a hypothesis of another tool.
\end{itemize}

\begin{figure}[htb]
    \centering
    \begin{tikzpicture}
\node[pfnode] (U) at (-5,1.3) {(G) Galois};
\node[pfnode] (f) at (-5,-1.3) {(F) finite certificate};
\node[pfnode] (s) at (0,0) {$U_0,\dots,U_s$};
\node[pfnode] (r) at (5,1.3) {(U) uniqueness};
\node[pfnode] (c) at (5,-1.3) {(P) parity};
\node[pfgoal] (goal) at (9.6,0) {$\bigcap_i\mathscr L(U_i)\ne\varnothing$};
\draw[pfedge] (U) -- node[pflab,above right,pos=.2] {constrains $U_i$ and $\mathscr L(U_i)$} (s);
\draw[pfedge] (f) -- node[pflab,below right,pos=.2] {normal forms, hypotheses} (s);
\draw[pfedge,dashed] (U) -- node[pflab,left] {fixes\\ phases} (f);
\draw[pfedge] (s) -- node[pflab,above left,pos=.8] {$\exists i:\ |\mathscr L(U_i)|=1$} (r);
\draw[pfedge] (s) -- node[pflab,below left,pos=.8] {all $U_i\in\C_2(n)$,\\ $\langle\gamma_0^{-1}\gamma_i\rangle$ a $2$-group} (c);
\draw[pfedge] (r) -- (goal);
\draw[pfedge] (c) -- (goal);
\draw[pfedge] (f) |- (0,-2.6) -| (goal);
\node[pflab,below] at (2.5,-2.6) {or settles a case directly};
\end{tikzpicture}
\caption{A map of the relationship among the reduction techniques.}
\label{fig:map}
\end{figure}
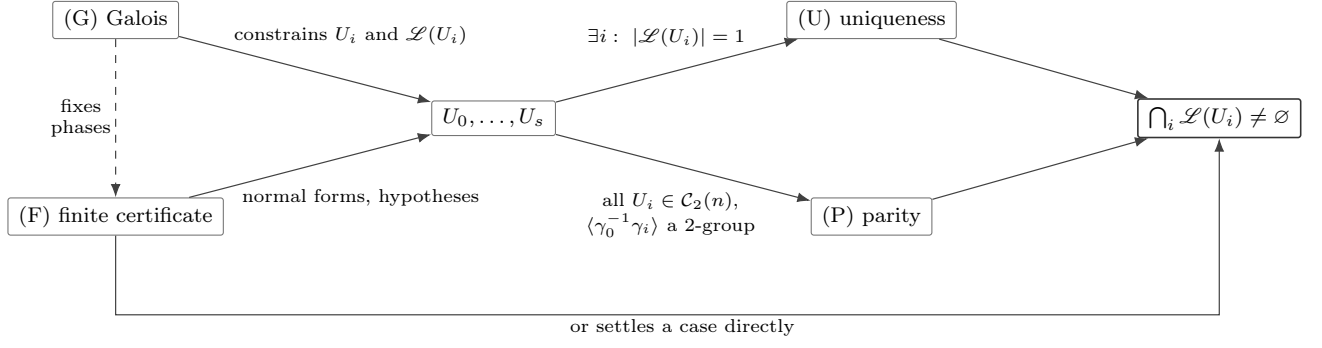

We first record how GSC pairs transform under Clifford gates.
\begin{lemma}\label{lem:tool-transport}
Let $U$ be a gate and $C_1,C_2,K,K'\in\C_2(n)$.
\renewcommand{\labelenumi}{(\theenumi)}
\begin{enumerate}
\item $(L,M)\in\mathscr L(U)$ if and only if $(\gamma_{C_2}^{-1}L,\gamma_{C_1}M)\in\mathscr L(C_1UC_2)$.
\item If $UK\doteq K'U$ and $(L,M)\in\mathscr L(U)$, then $(\gamma_KL,\gamma_{K'}M)\in\mathscr L(U)$.
\end{enumerate}
In particular, if $UK\doteq K'U$, then
\[
\gamma_K\{L:\exists\,M,\ (L,M)\in\mathscr L(U)\}
=
\{L:\exists\,M,\ (L,M)\in\mathscr L(U)\}.
\]
If moreover $\mathscr L(U)=\{(L,M)\}$, then $\gamma_KL=L$ and $\gamma_{K'}M=M$.
\end{lemma}
\begin{proof}
Conjugating $\A_{\gamma_{C_2}^{-1}L}$ by $C_1UC_2$ gives $C_1(U\A_LU^\dagger)C_1^\dagger$, which equals $\A_{\gamma_{C_1}M}$ if and only if $U\A_LU^\dagger=\A_M$. The rest follows from the observation that $U\A_{\gamma_KL}U^\dagger =UK\A_LK^\dagger U^\dagger=K'U\A_LU^\dagger K'^\dagger =\A_{\gamma_{K'}M}$.
\end{proof}

\begin{lemma}\label{lem:tool-blocks}
Let $1\le j<n$, and let $W$ commute with $Z_1,\dots,Z_j$, so that $W=\bigoplus_{x\in\F_2^j}B_x$, where $B_x$ acts on qubits $j+1,\dots,n$ in the joint eigenspace $\mathcal H_x$ of $Z_1,\dots,Z_j$ with eigenvalues $(-1)^{x_1},\dots,(-1)^{x_j}$. Let $S=\Span\{z_1,\dots,z_j\}$, and write $V_{n-j}$ for the labels of qubits $j+1,\dots,n$.
\renewcommand{\labelenumi}{(\theenumi)}
\begin{enumerate}
\item For Lagrangians $L,M\le V_{n-j}$, $(S\oplus L,S\oplus M)\in\mathscr L(W)$ if and only if $(L,M)\in\mathscr L(B_x)$ for every $x$. Moreover, $\mathscr L(W)\ne\varnothing$ if and only if $\bigcap_x\mathscr L(B_x)\ne\varnothing$. \label{itm:block1}
\item If $W\in\C_k(n)$, then $B_x\in\C_k(n-j)$ for every $x$, and $\diag(B_x,B_y)\in\C_k(n-j+1)$ for $x\ne y$.
\item Let $e\in\F_2^j$, $P$ be a Pauli on qubits $j+1,\dots,n$, and $h=W(X^e\otimes P)W^\dagger$. Then $(X^e\otimes I)h=\bigoplus_xB_{x+e}PB_x^\dagger$. If $\langle\Pau(n),h\rangle$ is a finite $2$-group modulo scalars, then so is $\langle\Pau(n-j),B_{x+e}PB_x^\dagger\rangle$ for every $x$. \label{item:block3}
\end{enumerate}
\end{lemma}
\begin{proof}
The algebra $\A_{S\oplus L}$ consists of the block-diagonal operators $\bigoplus_xa_x$ with all $a_x\in\A_L$, so $W\A_{S\oplus L}W^\dagger=\A_{S\oplus M}$ exactly when $B_x\A_LB_x^\dagger=\A_M$ for every $x$. Next, let $(L',M')\in\mathscr L(W)$. Since $W$ commutes with $Z_1,\dots,Z_j$, we have $W\A_SW^\dagger=\A_S$, so Lemma~\ref{lem:root-witness} with $T=S$ gives
\[
W\A_{S+(L'\cap S^\perp)}W^\dagger=\A_{S+(M'\cap S^\perp)},
\]
where $S+(L'\cap S^\perp)$ and $S+(M'\cap S^\perp)$ are Lagrangians containing $S$. These lie in $S^\perp=S\oplus V_{n-j}$, so they have the form $S\oplus L$ and $S\oplus M$. Indeed, a Lagrangian $N$ with $S\le N\le S\oplus V_{n-j}$ equals $S\oplus(N\cap V_{n-j})$, and $N\cap V_{n-j}$ is Lagrangian in $V_{n-j}$. The first statement then gives $(L,M)\in\bigcap_x\mathscr L(B_x)$. Conversely, a common GSC pair $(L,M)$ of the $B_x$ gives $(S\oplus L,S\oplus M)\in\mathscr L(W)$. This concludes the proof of~\eqref{itm:block1}.

Recall that $W$ acts on $\mathcal H_x$ by $B_x$. Let $G$ be a group of commuting Paulis, $\mathcal Q$ a joint eigenspace of $G$, and identify $\mathcal Q$ with $m$ qubits so that the restrictions to $\mathcal Q$ of the Paulis commuting with $G$ are exactly the elements of $\Pau(m)$, up to scalar. We claim that
\[
U\in\C_k(n),\ U\text{ commutes with }G\quad\Longrightarrow\quad U\mathcal Q=\mathcal Q\ \text{ and }\ U|_{\mathcal Q}\in\C_k(m).
\]
For $k=1$ the claim is the identification. For $k>1$, write an element of $\Pau(m)$ as $P|_{\mathcal Q}$ with $P$ a Pauli commuting with $G$. Then $UPU^\dagger\in\C_{k-1}(n)$ commutes with $G$, and by induction
\[
U|_{\mathcal Q}\,P|_{\mathcal Q}\,U|_{\mathcal Q}^\dagger=(UPU^\dagger)|_{\mathcal Q}\in\C_{k-1}(m).
\]
We apply the claim to $U=W$ twice:
\begin{itemize}
\item $G=\langle Z_1,\dots,Z_j\rangle$, $\mathcal Q=\mathcal H_x$, $m=n-j$: then $W|_{\mathcal Q}=B_x$.
\item $G=\langle Z^c:c\cdot(x+y)=0\rangle$, $m=n-j+1$: the joint eigenspace containing $\mathcal H_x$ is $\mathcal Q=\mathcal H_x\oplus\mathcal H_y$, with logical $X=X^{x+y}$ and logical $Z=Z_i$ for an $i$ with $x_i\ne y_i$. Then $W|_{\mathcal Q}=\diag(B_x,B_y)$, up to relabeling the logical qubit.
\end{itemize}

We now proceed to prove~\eqref{item:block3}. The first identity can be verified directly. For the second conclusion, let $G=\langle\Pau(n),h\rangle$ modulo scalars, and let $G_0\le G$ be the subgroup of elements preserving every $\mathcal H_y$, that is, of the form $u=\bigoplus_yu_y$. For each $x$, restriction to $\mathcal H_x$ is a homomorphism
\[
\rho_x:G_0\to\mathrm{PU}(2^{n-j}),\qquad \textstyle\bigoplus_yu_y\mapsto u_x.
\]
Since $G_0$ contains $I\otimes Q$ for every $Q\in\Pau(n-j)$, and contains $(X^e\otimes I)h=\bigoplus_yB_{y+e}PB_y^\dagger$,
\[
\langle\Pau(n-j),\,B_{x+e}PB_x^\dagger\rangle\le\rho_x(G_0).
\]
The right hand side is the image of a subgroup of the finite $2$-group $G$, so it is a finite $2$-group, and so is the left hand side.
\end{proof}

Recall that a \textit{complete affine flag} is the sequence of coset partitions associated to $0=W_0<\cdots<W_n=\F_2^n$, $\dim W_i=i$. For any permutation $\pi$, if $\pi(x+W_i)=\pi(x)+W_i$ for all $x,i$, then the flag is invariant under $\pi$.

\begin{lemma}\label{lem:tool-flags}
Let $\pi\in\operatorname{Sym}(\F_2^n)$, and $\mathcal T_n\coloneqq\{x\mapsto x+a:a\in\F_2^n\}$ be the translation subgroup.
\renewcommand{\labelenumi}{(\theenumi)}
\begin{enumerate}
\item If $\pi$ preserves a complete affine flag $\mathcal F$, then $\pi\in\operatorname{Stab}(\mathcal F)$ where $\operatorname{Stab}(\mathcal F)$ is a $2$-group of order $2^{2^n-1}$ containing $\mathcal T_n$. \label{itm:perm1}
\item If $\pi\in G$ for a $2$-group $\mathcal T_n\le G\le\operatorname{Sym}(\F_2^n)$, then there exists a complete affine flag $\mathcal F$ preserved by every element of $G$, and hence by $\pi$. \label{itm:perm2}
\item $\pi$ preserves a complete affine flag if and only if there exists a $2$-group $G\le\operatorname{Sym}(\F_2^n)$ such that $\mathcal T_n\le G$ and $\pi\in G$. \label{itm:perm3}
\end{enumerate}
\end{lemma}
\begin{proof}
Part~\eqref{itm:perm3} follows from~\eqref{itm:perm1} and~\eqref{itm:perm2}. For~\eqref{itm:perm1}, translations preserve every coset partition, hence every flag. Let $Q_n = \operatorname{Stab}(\mathcal{F})$ for the complete affine flag $0=W_0<\cdots<W_n \eqqcolon \mathcal{F}$. Each $\pi\in Q_n$ permutes the cosets of $W_1$ and induces a permutation $\bar\pi$ of $\F_2^n/W_1$ preserving the complete affine flag $(W_i/W_1)_{i\ge1}$. Every permutation of $\F_2^n/W_1$ preserving this flag is some $\bar\pi$, since any $\pi$ mapping each coset of $W_1$ onto its image preserves the flag. The kernel of $\pi\mapsto\bar\pi$ fixes or swaps each of the $2^{n-1}$ cosets of $W_1$ independently. Hence $|Q_n|=2^{2^{n-1}}|Q_{n-1}|$, and $|Q_0|=1$ gives $|Q_n|=2^{2^n-1}$, so $Q_n$ is a Sylow $2$-subgroup of $\operatorname{Sym}(\F_2^n)$ (cf.~\cite{DixonMortimer1996}).

We now prove~\eqref{itm:perm2}. Take $q\ne1$ in the center of $G$. Since $q$ commutes with all translations, $q(x)=x+v$ with $v=q(0)\ne0$. Every $g\in G$ commutes with $q$, so
\[
g(x+v)=g(x)+v\qquad\text{for all }x\in\F_2^n.
\]
Hence $g$ maps cosets of $\Span\{v\}$ to cosets of $\Span\{v\}$ and induces a permutation $\bar g$ of $\bar V=\F_2^n/\Span\{v\}\cong\F_2^{n-1}$. The map $g\mapsto\bar g$ is a homomorphism, so $\bar G=\{\bar g:g\in G\}$ is a $2$-group, and it contains all translations of $\bar V$, the images of the translations of $\F_2^n$. By induction on $n$, $\bar G$ preserves a complete affine flag $0=\bar W_0<\cdots<\bar W_{n-1}=\bar V$. Let $W_0=0$ and let $W_i$ be the preimage of $\bar W_{i-1}$ for $1\le i\le n$, so that $W_1=\Span\{v\}$ and $\dim W_i=i$. For $g\in G$, the above formula gives $g(x+W_1)=g(x)+W_1$. For $i\ge2$, $x+W_i$ is the preimage of $\bar x+\bar W_{i-1}$, so $g(x+W_i)$ is the preimage of $\bar g(\bar x)+\bar W_{i-1}$, which is $g(x)+W_i$.
\end{proof}

\begin{proposition}\label{lem:tool-monomial}
Let $n\le4$, $U\in\C(n)$ and $(L,M)\in\mathscr L(U)$. There are $C_1,C_2\in\C_2(n)$ with $C_2\A_LC_2^\dagger=\A_Z$ and $C_1\A_ZC_1^\dagger=\A_M$, a permutation $\pi$ of $\F_2^n$ and a diagonal gate $D$ such that $U\doteq C_1\Pi_\pi DC_2$, where $\A_Z$ is the algebra of diagonal matrices and $\Pi_\pi$ is the permutation gate $\ket x\mapsto\ket{\pi x}$, with the following properties:
\renewcommand{\labelenumi}{(\theenumi)}
\begin{enumerate}
\item The permutation $\pi$ preserves a complete affine flag. If $n\le2$, then $\pi=1$. If $n=3$, then $\pi=1$ or $\pi(x,y,z)=(x,y,z+xy)$. \label{itm:monomial1}
\item The ratios of the diagonal entries of $D$ lie in $\mu_{2^s}$ for some $s$, in $\mu_8$ if $U$ is defined over $K_3$ up to scalar, and in $\mu_4$ if $U$ is Gaussian.
\end{enumerate}
\end{proposition}
\begin{proof}
Since the Clifford group acts transitively on Lagrangians, the desired $C_1$ and $C_2$ exist. Then $C_1^\dagger UC_2^\dagger\,\A_Z\, C_2U^\dagger C_1 = \A_Z$. The rank-one projections in $\A_Z$ are the $\ket{x}\bra{x}$, and conjugation by $C_1^\dagger UC_2^\dagger$ maps $\A_Z$ onto itself and preserves rank, so it permutes them. Hence $C_1^\dagger UC_2^\dagger=\Pi_\pi D$ for some permutation $\pi$ of $\F_2^n$ and diagonal unitary $D$. This gate has finite level, so by~\cite[Corollary~A.11.1]{Anderson_2024} both $\Pi_\pi$ and $D$ have finite level.

Every permutation of $\F_2^n$ with $n\le2$ is affine. For $n=3,4$ the affine classifications~\cite[Appendix~C]{Anderson_2024} and~\cite[Section~3]{sc-AC2025} give affine $\alpha,\beta$ with $\alpha\pi\beta$ preserving a complete affine flag. Affine permutation gates are Clifford and normalize $\A_Z$, so replacing $C_1,C_2$ by $C_1\Pi_\alpha^{-1}$, $\Pi_\beta^{-1}C_2$ and $D$ by $\Pi_\beta^{-1}D\Pi_\beta$ we may assume that $\pi$ preserves the flag. On three qubits such a permutation is $(x+c,\,y+g_1(x),\,z+h(x,y))$, which is affine or the composition of an affine map with $(x,y,z)\mapsto(x,y,z+xy)$. Absorbing the affine part once more gives~\eqref{itm:monomial1}.

The diagonal classification~\cite[Appendix~F]{Anderson_2024} puts the ratios of the entries of $D$ in some $\mu_{2^s}$. Clifford gates are defined over $K_2=\Q(i)$ up to scalar, so $\Pi_\pi D$ is defined over the same field as $U$. Under the extra hypotheses the ratios therefore lie in $\mu_{2^s}\cap K_3=\mu_8$ and respectively $\mu_{2^s}\cap\Q(i)=\mu_4$.
\end{proof}

The next lemma gives the \emph{(U) uniqueness} technique.
\begin{lemma}\label{lem:tool-unique}
Let $U$ be an $n$-qubit gate.
\renewcommand{\labelenumi}{(\theenumi)}
\begin{enumerate}
\item If $L\le V_n$ is Lagrangian and $L\subseteq\mathsf P(U)$, then there exists a Lagrangian $M\le V_n$ such that $(L,M)\in\mathscr L(U)$. \label{itm:unique1}
\item If $n\le2$, then for every Lagrangian $L\le V_n$, there exists a Lagrangian $M \le V_n$ such that $(L,M)\in\mathscr L(U)$ if and only if $L\subseteq\mathsf P(U)$. \label{itm:unique2} In particular, if $\dim\mathsf P(U)\le n$, then $|\mathscr L(U)|\le1$.
\item If $n=3$ and $U\in\C(3)$, then $L\subseteq\mathsf C(U)$ for every $(L,M)\in\mathscr L(U)$. \label{itm:unique3}
\item A one-qubit gate that is not Clifford has at most one GSC pair.
\end{enumerate}
\end{lemma}
\begin{proof}
If $L\subseteq\mathsf P(U)$, then $U\A_LU^\dagger$ is spanned by the $2^n$ commuting Paulis $UP_vU^\dagger$, $v\in L$, where the $v$'s form a Lagrangian $M$, so $(L,M)\in\mathscr L(U)$.

We next prove~\eqref{itm:unique2} and~\eqref{itm:unique3} together. Let $(L,M)\in\mathscr L(U)$. As in the proof of Proposition~\ref{lem:tool-monomial}, write $U\doteq C_1\Pi_\pi DC_2$ with $C_2\A_LC_2^\dagger=\A_Z$, and in case~\eqref{itm:unique3} choose this form as in Proposition~\ref{lem:tool-monomial}. For $u\in L$ we have $C_2P_uC_2^\dagger=\pm Z^{u'}$, and $D$ commutes with $Z^{u'}$, so
\[
UP_uU^\dagger\doteq C_1\Pi_\pi Z^{u'}\Pi_\pi^\dagger C_1^\dagger,
\qquad
\Pi_\pi Z^{u'}\Pi_\pi^\dagger=\diag_x(-1)^{u'\cdot\pi^{-1}(x)}.
\]
If $n\le2$, then $\pi$ is affine, since every permutation of $\F_2^n$ is, so $u'\cdot\pi^{-1}(x)$ is affine in $x$ and the diagonal factor is a Pauli. Hence $L\subseteq\mathsf P(U)$, and~\eqref{itm:unique1} gives the converse. If $\dim\mathsf P(U)\le n$, a Lagrangian contained in $\mathsf P(U)$ equals $\mathsf P(U)$, so there is at most one input Lagrangian, and the input determines the output through $\A_M=U\A_LU^\dagger$. In case~\eqref{itm:unique3}, Proposition~\ref{lem:tool-monomial}(1) gives $\pi=1$ or $\pi(x,y,z)=(x,y,z+xy)$, and $\pi^{-1}=\pi$ has degree at most two. Then the diagonal factor is a diagonal Clifford, so $L\subseteq\mathsf C(U)$.

Lastly, recall that $\mathsf P(U)$ is a subspace of $V_1$. But $\mathsf P(U) \subsetneq V_1$, since otherwise $U$ would be Clifford. Hence $\dim\mathsf P(U)\le1$, and~\eqref{itm:unique2} gives $|\mathscr L(U)|\le1$.
\end{proof}

The next lemma gives the \emph{(P) parity} technique.
\begin{lemma}\label{lem:tool-two-group}
\renewcommand{\labelenumi}{(\theenumi)}
The following statements hold.
\begin{enumerate}
\item Every $2$-subgroup of $\Sp(V_n)$ fixes a Lagrangian. Hence, if $U_0,\dots,U_s\in\C_2(n)$ and the maps $\gamma_{U_0}^{-1}\gamma_{U_a}$ generate a $2$-group, then $\bigcap_a\mathscr L(U_a)\ne\varnothing$.
\item Let $U$ be a gate such that $\langle\Pau(n),U\rangle$ is a finite $2$-group modulo scalars. Then $U$ normalizes a Pauli MASA\@. If $U$ is Clifford, then $\gamma_U$ is unipotent, so $\gamma_U$ fixes a nonzero vector in every nonzero $\gamma_U$-invariant subspace. If $U$ normalizes $\A_Z$, its permutation preserves a complete affine flag.
\end{enumerate}
\end{lemma}
\begin{proof}
The orbits of a $2$-group have $2$-power size, and the number of Lagrangians is odd by Lemma~\ref{lem:lagrangian-count}. Then a Lagrangian $L$ fixed by every $\gamma_{U_0}^{-1}\gamma_{U_a}$ gives the common GSC pair $(L,\gamma_{U_0}L)$.

By Theorem~\ref{thm:two-group-chamber}, the group $\langle\Pau(n),U\rangle$ fixes a Pauli MASA\@. If $U$ is Clifford, $U^{2^k}$ is scalar for some $k$, so $(\gamma_U-I)^{2^k}=\gamma_U^{2^k}-I=0$ over $\F_2$. If $U$ normalizes $\A_Z$, so does every element of the group, and its permutation image is a $2$-group containing the translations $x\mapsto x+c$ of the gates $X^c$. Applying Lemma~\ref{lem:tool-flags}(2) completes the proof.
\end{proof}

The following lemma treats gates on one qubit.
\begin{lemma}\label{lem:tool-one-qubit}
\renewcommand{\labelenumi}{(\theenumi)}
The following statements hold.
\begin{enumerate}
\item Let $1 = \pi_0,\pi_1,\dots,\pi_s$ be permutations of a three-element set, any two of which agree at some point. Then they have a common fixed point.
\item If $U_0,\dots,U_s\in\mathrm U(2)$ satisfy $\mathscr L(U_a)\cap\mathscr L(U_b)\ne\varnothing$ for all $a,b$, then $\bigcap_a\mathscr L(U_a)\ne\varnothing$. \label{itm:one2}
\end{enumerate}
\end{lemma}
\begin{proof}
First, notice that two permutations $\pi$ and $\pi'$ agree at a point exactly when $\pi'^{-1}\pi$ fixes it. On three points the permutations without a fixed point are the two $3$-cycles, and the product of two distinct transpositions is a $3$-cycle. Taking $\pi'=\pi_0=1$, every $\pi_a$ is the identity or a transposition, and two transpositions among the $\pi_a$ agree somewhere, so they are equal. Hence all $\pi_a$ lie in $\{1,\tau\}$ for one transposition $\tau$, and the point fixed by $\tau$ is fixed by every $\pi_a$.

Now we prove~\eqref{itm:one2}. Let $\Lambda$ be the set of the three Lagrangians of $V_1$, the lines spanned by the labels of $X$, $Y$ and $Z$. The hypothesis with $b=a$ gives $\mathscr L(U_a)\ne\varnothing$ for every $a$. Consider the following cases.
\begin{itemize}
\item \emph{Some $U_a$ is not Clifford:} then $\mathscr L(U_a)$ has at most one element by Lemma~\ref{lem:tool-unique}(4), so $\mathscr L(U_a)=\{(L,M)\}$. For every $b$ the nonempty set $\mathscr L(U_a)\cap\mathscr L(U_b)$ is $\{(L,M)\}$, so $(L,M)\in\bigcap_b\mathscr L(U_b)$.
\item \emph{All $U_a$ are Clifford:} let $\gamma_a=\gamma_{U_a}$, so that $\mathscr L(U_a)=\{(L,\gamma_aL):L\in\Lambda\}$, and let $\pi_a=\gamma_0^{-1}\gamma_a$, a permutation of $\Lambda$ with $\pi_0=1$. A common GSC pair of $U_a$ and $U_b$ is a pair $(L,\gamma_aL)=(L,\gamma_bL)$, that is, a point $L\in\Lambda$ with $\pi_aL=\pi_bL$. By (1) some $L\in\Lambda$ is fixed by every $\pi_a$. Then $\gamma_aL=\gamma_0L$ for all $a$, and $(L,\gamma_0L)\in\bigcap_a\mathscr L(U_a)$. \qedhere
\end{itemize}
\end{proof}

\begin{lemma}\label{lem:tool-two-paulis}
Let $U\in\C(3)$, and consider commuting Pauli gates $P_1=P_{v_1}$, $P_2 = P_{v_2}$ where $v_1$ and $v_2$ are linearly independent and $UP_jU^\dagger\in\mathcal P(3)$ for $j=1,2$. Then $U\in\GSC(3)$.
\end{lemma}
\begin{proof}
Choose Clifford gates $C_1,C_2$ with $C_2Z_jC_2^\dagger=P_j$ and $C_1UP_jU^\dagger C_1^\dagger=Z_j$. Then $V=C_1UC_2$ commutes with $Z_1,Z_2$, so $V=\bigoplus_{u\in\F_2^2}V_u$ with one-qubit blocks $V_u$ on qubit $3$. For $u\ne u'$, Lemma~\ref{lem:tool-blocks}(2) gives $\diag(V_u,V_{u'})\in\C(2)$, which has a GSC pair by the two-qubit theorem~\cite{sc-ZCC2008}, and Lemma~\ref{lem:tool-blocks}(1) turns it into a common GSC pair of $V_u$ and $V_{u'}$. By Lemma~\ref{lem:tool-one-qubit}(2) the four blocks share a GSC pair, so $V\in\GSC(3)$ by Lemma~\ref{lem:tool-blocks}(1), and $U\in\GSC(3)$.
\end{proof}

The next proposition gives the \emph{(G) Galois} technique.
\begin{proposition}\label{lem:tool-galois}
Let $r\ge3$ and let $\sigma\in\operatorname{Gal}(K_r/\Q)$ fix $i$. For a gate $U$ that is a scalar multiple of a matrix over $K_r$, let $\sigma(U)$ be the entrywise image of that matrix, which is defined up to scalar.
\renewcommand{\labelenumi}{(\theenumi)}
\begin{enumerate}
\item $\sigma(C)\doteq C$ for every $C\in\C_2(n)$. \label{itm:galois1}
\item If $(L,M)\in\mathscr L(U)$, then $U^{-1}\sigma(U)\in\A_L$ and $\sigma(U)U^{-1}\in\A_M$. In particular, if $U^{-1}\sigma(U)\doteq P_a$, then $a\in L$.
\item If $\sigma(UP_vU^\dagger)\doteq UP_vU^\dagger$ for every $v\in V_n$, then $\sigma(U)\doteq UP_a$ for some $a\in V_n$.
\end{enumerate}
\end{proposition}
\begin{proof}
Every Clifford gate is a scalar multiple of a matrix over $K_2=\Q(i)$ since each Hadamard can be scaled by $\sqrt2$, while phase and $\CNOT$ gates already have entries in $K_2$. As $\sigma$ fixes $K_2$, $\sigma(C)\doteq C$.

Next, choose $C_1,C_2\in\C_2(n)$ with $C_2\A_LC_2^\dagger=\A_Z$ and $C_1\A_ZC_1^\dagger=\A_M$. Then $C_1^\dagger UC_2^\dagger=\Pi D$ is monomial, with $\Pi$ a permutation matrix and $D$ diagonal, and $\sigma(\Pi D)=\Pi\sigma(D)$. By~\eqref{itm:galois1},
\[
U^{-1}\sigma(U)\doteq C_2^\dagger\bigl(D^{-1}\sigma(D)\bigr)C_2
\in C_2^\dagger\A_ZC_2=\A_L,
\]
and $\sigma(U)U^{-1}=U\bigl(U^{-1}\sigma(U)\bigr)U^{-1}\in U\A_LU^\dagger=\A_M$. Finally, $P_a\in\A_L$ exactly when $a\in L$.

Lastly, Pauli matrices have entries in $\Z[i]$, so $\sigma(P_v)=P_v$ and $\sigma(U)P_v\sigma(U)^{-1}\doteq UP_vU^{-1}$. Thus $Y=U^{-1}\sigma(U)$ satisfies $YP_vY^{-1}=c_vP_v$ with $c_v^2=1$, and $v\mapsto c_v$ is a character, $c_v=(-1)^{\omega(a,v)}$. Then $P_aY$ commutes with every Pauli, so it is scalar, and $\sigma(U)\doteq UP_a$.
\end{proof}

The following example applies the Galois technique.
\begin{example}\label{ex:strategy-galois}
Let $\mathrm{CT}=\operatorname{diag}(1,1,1,\zeta_8)$, and let $\sigma(\zeta_8)=-\zeta_8$, which fixes $i=\zeta_8^2$. Then $\sigma(\mathrm{CT})=\mathrm{CT}\cdot\CZ$, so by Proposition~\ref{lem:tool-galois}(2), every $(L,M)\in\mathscr L(\mathrm{CT})$ has
\[
\CZ=\tfrac12(I+Z_1+Z_2-Z_1Z_2)\in\A_L\cap\A_M.
\]
Hence $z_1,z_2\in L\cap M$, and $\mathscr L(\mathrm{CT})=\{(\Span\{z_1,z_2\},\Span\{z_1,z_2\})\}$. The Galois constraint alone shows that the GSC pair of $\mathrm{CT}$ is unique. Later, in Lemma~\ref{app:return-block}(3) and in Cases~1b and~1c of the proof of Lemma~\ref{app:return-two}, the Galois cocycle is a Pauli $P_a$, which puts $a$ in every input Lagrangian, and in Lemma~\ref{lem:galois-normal-form} it is the diagonal gate $\diag_x i^{p\cdot x}(-1)^{q(x)}$, a sign gate exactly when $p=0$.
\end{example}

\begin{proposition}\label{app:lem-R_n}
    For $n=3,4$, let $\mathscr R_n$ be all the gates with a Clifford conjugate of the form
\begin{equation}\label{core:eq:row}
M|x\rangle=\zeta_8^\lambda i^{f(x)}|\pi x\rangle,\quad
\lambda\in\{0,1\},\quad f:\mathbb F_2^n\to\mathbb Z/4,
\end{equation}
where $\pi$ preserves a complete affine flag, $\pi^2=1$, $f+f\circ\pi=-\lambda\pmod4$, and $\operatorname{tr}M=0$. Then $\mathscr R_n$ is the class of traceless finite-level Gaussian GSC Hermitian involutions. For every $h\in\mathscr R_n$, the group $\langle\Pau(n),h\rangle$ is a finite $2$-group modulo scalars.
\end{proposition}
\begin{proof}

Write $[\,\cdot\,]$ for the image in $\mathrm{PU}(2^n)$ and $\overline{\mathcal P}(n)=[\mathcal P(n)]$. For a complete affine flag $F$ and $t\ge2$, let $\mathcal M_F(t)$ be the group of all $[\Pi_\rho D]$ with $\rho$ preserving $F$ and $D$ diagonal with entries in $\mu_{2^t}$. By Lemma~\ref{lem:tool-flags}(1), $\mathcal M_F(t)$ is a finite $2$-group. It contains $\overline{\mathcal P}(n)$, since $X^aZ^b$ has permutation $x\mapsto x+a$ and entries in $\mu_4$.

We first show that every finite $2$-group $G\le\mathrm{PU}(2^n)$ containing $\overline{\mathcal P}(n)$ lies in $[\C_c(n)]$ for some $c$. Let $Z_1(G)\le Z_2(G)\le\cdots$ be the upper central series. We show $Z_j(G)\subseteq[\C_j(n)]$ by induction. We have $Z_1(G)\subseteq\overline{\mathcal P}(n)=[\C_1(n)]$. If $Z_j(G)\subseteq[\C_j(n)]$, $g\in Z_{j+1}(G)$ and $P\in\overline{\mathcal P}(n)$, then
\[
gPg^{-1}=[g,P]\,P,\qquad [g,P]=gPg^{-1}P^{-1}\in Z_j(G)\subseteq[\C_j(n)],
\]
and right multiplication by a Pauli preserves $\C_j(n)$. Hence $g\in[\C_{j+1}(n)]$. Since $G$ is nilpotent, $Z_c(G)=G$ for some $c$.

\begin{itemize}
\item \step{Step 1:}{every $h\in\mathscr R_n$ is a traceless finite-level Gaussian GSC Hermitian involution, and $\langle\Pau(n),h\rangle$ is a finite $2$-group modulo scalars.} Write $h=C^\dagger MC$ with $C\in\C_2(n)$ and $M$ as in~\eqref{core:eq:row}. Since $\zeta_8^{2\lambda}=i^\lambda$ and $f(x)+f(\pi x)=-\lambda$,
\[
M^2|x\rangle
=\zeta_8^{2\lambda}i^{f(x)+f(\pi x)}|\pi^2x\rangle
=|x\rangle.
\]
Thus $M$ is a unitary involution, hence Hermitian. It is traceless by assumption, $\zeta_8^{-\lambda}M$ has entries in $\Q(i)$, and $M$ is monomial, so it normalizes $\A_Z$. If $\pi$ preserves the flag $F$, then $[M]\in\mathcal M_F(3)$, so $\langle\overline{\mathcal P}(n),[M]\rangle\le\mathcal M_F(3)$.

Conjugation by $C$ preserves these properties and turns $\A_Z$ into the Pauli MASA $C^\dagger\A_ZC$, which is normalized by $h$, so $h$ is GSC\@. Moreover $[\langle\Pau(n),h\rangle]=[C]^{-1}\langle\overline{\mathcal P}(n),[M]\rangle[C]$ is a finite $2$-group, and the second paragraph of this proof gives $h\in\C(n)$.

\item \step{Step 2:}{every traceless finite-level Gaussian GSC Hermitian involution $h$ lies in $\mathscr R_n$.} We first show that $h$ normalizes a Pauli MASA\@. By Proposition~\ref{lem:tool-monomial} applied to a GSC pair of $h$, we have $h\doteq C_1NC_2$ with $C_1,C_2\in\C_2(n)$ and $N=\Pi_\pi D$, where $\pi$ preserves a complete affine flag $F$ and, $h$ being Gaussian, $[N]\in\mathcal M_F(2)$. Since $C_1^\dagger h=NC_2$ up to scalar and $C_2$ normalizes the Pauli group,
\[
[C_1]^{-1}\,\Gamma(h)\,[C_1]
=\langle\overline{\mathcal P}(n),\,[N]\overline{\mathcal P}(n)[N]^{-1}\rangle
\le\mathcal M_F(2),
\]
so $\Gamma(h)$ is a finite $2$-group. As $h^2=I$, conjugation by $[h]$ swaps $\overline{\mathcal P}(n)$ and $[h]\overline{\mathcal P}(n)[h]^{-1}$, so $[h]$ normalizes $\Gamma(h)$ and $G=\Gamma(h)\langle[h]\rangle$ is a finite $2$-group containing $\overline{\mathcal P}(n)$. By Theorem~\ref{thm:two-group-chamber}, $G$ fixes a Pauli MASA, and in particular $h$ normalizes it.

Let $G$ be as above and choose $C\in\C_2(n)$ mapping the fixed Pauli MASA to $\A_Z$. Every element of $[C]G[C]^{-1}$ normalizes $\A_Z$, so it is a projective monomial, in particular $M=ChC^\dagger=\Pi_\pi D'$.

Consider the homomorphism sending $[\Pi_\rho D]$ to $\rho$. The image of $[C]G[C]^{-1}$ is a $2$-group containing all translations, since $[C]\overline{\mathcal P}(n)[C]^{-1}=\overline{\mathcal P}(n)$. By Lemma~\ref{lem:tool-flags}(2) it preserves a complete affine flag, and so does $\pi$.

$M$ has finite level and is Gaussian, since Clifford conjugation preserves both. Therefore $D'$ has finite level~\cite[Corollary~A.11.1]{Anderson_2024}, its entry ratios lie in some $\mu_{2^t}$~\cite[Appendix~F]{Anderson_2024}, and Gaussianity puts them in $\mu_{2^t}\cap\Q(i)=\mu_4$. Thus $M|x\rangle=c\,i^{f(x)}|\pi x\rangle$ for some $c$ and $f:\F_2^n\to\Z/4$.

$M^2=I$ gives
\[
\pi^2=1,\qquad c^2i^{f(x)+f(\pi x)}=1\quad\text{for every }x.
\]
Hence $c^2\in\mu_4$ and $c\in\mu_8$. Write $c=\zeta_8^\lambda i^m$ with $\lambda\in\{0,1\}$ and replace $f$ by $f+m$; then $f+f\circ\pi=-\lambda\pmod 4$. Finally $\operatorname{tr}M=\operatorname{tr}h=0$, so $h\in\mathscr R_n$. \qedhere
\end{itemize}
\end{proof}

We say a gate $D$ is \emph{stationary} if every GSC pair of $D$ has the form $(L,L)$, that is, $D$ maps each of its input Pauli MASAs to itself. The three-qubit gates $\mathrm{CS}_{yz}^a\CCZ^b$ with $a,b\in\F_2$, that is, $I$, $\mathrm{CS}_{yz}$, $\CCZ$ and $\mathrm{CS}_{yz}\CCZ$, are stationary by \textup{(\ref{cert:stationary-cs-ccz-gates-and-complete-affine-flags})}. For a third-level gate $D$, let $E_D\le\Sp(6,2)$ be generated by the symplectic actions of the Clifford gates $DP_vD^\dagger$, $v\in V_3$.

\begin{lemma}\label{lem:tool-bridge}
Let $D$ be a Clifford conjugate of some $\mathrm{CS}_{yz}^a\CCZ^b$ and $F\in\C_2(3)$. If $\langle\Pau(3),FD\rangle$ is a finite $2$-group modulo scalars, then so is $\langle\Pau(3),F,D\rangle$. In particular $\gamma_Fe$ is unipotent for every $e\in E_D$.
\end{lemma}
\begin{proof}
Write $D=CGC^\dagger$ with $C\in\C_2(3)$ and $G=\mathrm{CS}_{yz}^a\CCZ^b$.
\begin{itemize}
\item \step{Step 1:}{$F$ and $D$ normalize a common Pauli MASA $\A_N$.} By Lemma~\ref{lem:tool-two-group}(2), $FD$ normalizes a Pauli MASA $\A_N$. Then $D\A_ND^\dagger=F^\dagger\A_NF=\A_{\gamma_F^{-1}N}$, so $(N,\gamma_F^{-1}N)\in\mathscr L(D)$, and Lemma~\ref{lem:tool-transport}(1) gives $(\gamma_C^{-1}N,\gamma_C^{-1}\gamma_F^{-1}N)\in\mathscr L(G)$. As $G$ is stationary by \textup{(\ref{cert:stationary-cs-ccz-gates-and-complete-affine-flags})}, $\gamma_F^{-1}N=N$. Hence $D\A_ND^\dagger=\A_N=F\A_NF^\dagger$, and $G$ normalizes $\A_{N'}$ with $N'=\gamma_C^{-1}N$.
\item \step{Step 2:}{in Clifford coordinates in which $\A_N$ is diagonal, the permutations of $F$ and $D$ preserve a common complete affine flag.} Choose $E\in\C_2(3)$ with $E\A_NE^\dagger=\A_Z$. By Step~1, $EFE^\dagger$ and $EDE^\dagger$ normalize $\A_Z$, so they are monomial, $\Pi_\alpha\Delta_F$ and $\Pi_\pi\Delta_D$ with $\Delta_F,\Delta_D$ diagonal, and $E(FD)E^\dagger$ has permutation $\alpha\pi$.

For $c\in\F_2^3$ the conjugate of $X^c$ by $EFE^\dagger$ is a Pauli with permutation $\alpha\tau_c\alpha^{-1}$, where $\tau_c(x)=x+c$. This is a translation, so $\alpha(x+c)+\alpha(x)$ does not depend on $x$, i.e., $\alpha$ is affine. Since $\langle\Pau(3),E(FD)E^\dagger\rangle=E\langle\Pau(3),FD\rangle E^\dagger$ is a finite $2$-group modulo scalars, and Lemma~\ref{lem:tool-two-group}(2) applies, $\alpha\pi$ preserves a complete affine flag. Since $EDE^\dagger=(EC)G(EC)^\dagger$, and $EC$ maps $\A_{N'}$ to $\A_Z$, $\pi$ is a permutation of \textup{(\ref{cert:stationary-cs-ccz-gates-and-complete-affine-flags})}. Translations preserve every complete affine flag by Lemma~\ref{lem:tool-flags}(1). For a fixed input Lagrangian $N'$ of $G$, the gate $EC$ is unique up to a left factor $\mu\in\C_2(3)$ normalizing $\A_Z$, and $\mu$ conjugates $\alpha$ and $\pi$ by an affine permutation, which preserves the claim. The data are therefore $G$ (four gates), $N'$ (among the $135$ Lagrangians of $V_3$), $\alpha$ (among the $1{,}344$ affine permutations) and the $21$ complete affine flags of $\F_2^3$. The stationarity used in Step~1 is decided on the same $135$ Lagrangians, and a search over these data gives the following.

\begin{finitefact}{cert:stationary-cs-ccz-gates-and-complete-affine-flags}{Stationary gates $\mathrm{CS}_{yz}^a\CCZ^b$ and complete affine flags}{Lemma~\ref{lem:tool-bridge}}

The gates $G=\mathrm{CS}_{yz}^a\CCZ^b$ with $a,b\in\F_2$, that is, $I$, $\mathrm{CS}_{yz}$, $\CCZ$ and $\mathrm{CS}_{yz}\CCZ$, are stationary, with $135$, $3$, $15$ and $3$ GSC pairs. Let $\pi$ be the permutation of $CGC^\dagger$, where $C\in\C_2(3)$ maps the Pauli MASA of an input Lagrangian of $G$ to the diagonal matrices. Then for every affine permutation $\alpha$ of $\F_2^3$, if $\alpha\pi$ and the translations preserve a common complete affine flag, so do $\alpha$, $\pi$ and the translations.

\tcblower

\textbf{Code.} \certlink{stationary-cs-ccz-gates-and-complete-affine-flags}
\end{finitefact}

By \textup{(\ref{cert:stationary-cs-ccz-gates-and-complete-affine-flags})}, $\alpha$, $\pi$ and the translations then preserve a common complete affine flag.
\item \step{Step 3:}{$\langle\Pau(3),F,D\rangle$ is a finite $2$-group modulo scalars.} Let $Q$ be the $2$-group of the permutations preserving the flag of Step~2 (Lemma~\ref{lem:tool-flags}(1)). Modulo scalars, the gates $\Pi_\rho\Delta$ with $\rho\in Q$ and $\Delta$ diagonal with entries in $\mu_4$ form a group of order $|Q|\cdot4^7$. It contains the Paulis $X^cZ^u$, whose permutations are translations. It contains $EFE^\dagger$ and $EDE^\dagger$ up to scalar, since both are Gaussian and of finite level, so their entries lie in $\mu_4$ up to scalar~\cite[Corollary~A.11.1, Appendix~F]{Anderson_2024}. Hence it contains $E\langle\Pau(3),F,D\rangle E^\dagger$ modulo scalars.
\item \step{Step 4:}{$\gamma_Fe$ is unipotent for every $e\in E_D$.} As $C\mapsto\gamma_C$ is a homomorphism on $\C_2(3)$, $e=\gamma_c$ for a product $c$ of the Clifford gates $DP_vD^\dagger$. The Clifford gate $Fc$ lies in $\langle\Pau(3),F,D\rangle$, so $\langle\Pau(3),Fc\rangle$ is a finite $2$-group modulo scalars by Step~3, and Lemma~\ref{lem:tool-two-group}(2) makes $\gamma_{Fc}=\gamma_Fe$ unipotent. \qedhere
\end{itemize}
\end{proof}

The following lemma is used in the proof of Lemma~\ref{app:return-balanced}.

\begin{lemma}\label{lem:quarter-turn}
Let $J=\sum_uc_uP_u$ be a finite-level Hermitian involution on $n$ qubits.
\renewcommand{\labelenumi}{(\theenumi)}
\begin{enumerate}
\item For every Hermitian Pauli $P$, the order of $JP$ is a power of two. \label{itm:turn1}
\item If $(I+iJ)/\sqrt2$ has finite level, then $\operatorname{supp}_PJ=\{u:c_u\ne0\}$ is isotropic.
\end{enumerate}
\end{lemma}
\begin{proof}
First, let $J\in\C_k(n)$, $J_0=J$ and $J_{j+1}=J_jPJ_j$. Each $J_j$ is a Hermitian involution, and if $J_j\in\C_\ell(n)$ then $J_{j+1}\in\C_{\ell-1}(n)$, so $J_{k-1}$ is a Pauli gate up to sign. Since $J_{j+1}P=(J_jP)^2$, $J_jP=(JP)^{2^j}$. As $J_{k-1}P$ is a product of two Hermitian Paulis, it has order dividing $4$. Hence the order of $JP$ divides $2^{k+1}$.

Let $G=(I+iJ)/\sqrt2$. We first show that $(JP)^4=I$, i.e., $J$ commutes with $PJP$, for every Hermitian Pauli $P$. By Jordan's lemma~\cite[p.~6]{childs2013Jordan}, $J$ and $P$ preserve an orthogonal decomposition into blocks of dimension at most two. On a block of dimension one, $JP=\pm1$. On a block of dimension two choose a basis with
\[
J=Z,\qquad P=\cos\theta\,Z+\sin\theta\,X .
\]
There $JP$ has eigenvalues $e^{\pm i\theta}$, and $G=\operatorname{diag}(\zeta_8,\zeta_8^{-1})$ gives $GPG^\dagger=\cos\theta\,Z-\sin\theta\,Y$, so $(GPG^\dagger)P$ has determinant one and trace $2\cos^2\theta$. By~\eqref{itm:turn1}, applied to $J$ and to the finite-level Hermitian involution $GPG^\dagger$, both $JP$ and $(GPG^\dagger)P$ have $2$-power order. So $\eta=e^{2i\theta}$ and the eigenvalues $\xi^{\pm1}$ of $(GPG^\dagger)P$ are roots of unity of $2$-power order, and $2\cos^2\theta=1+\cos2\theta$ gives
\[
(\xi+\xi^{-1})-\tfrac12(\eta+\eta^{-1})=1.
\]
The normalized trace of a cyclotomic field $\Q(\zeta_{2^m})$ containing $\xi$ and $\eta$ sends $\epsilon+\epsilon^{-1}$, for $\epsilon$ a root of unity of $2$-power order, to $2$, $-2$ or $0$ according as $\epsilon=1$, $\epsilon=-1$ or otherwise. Applied to the identity, it leaves only $\eta=\xi=1$, and $\eta=-1$ with $\xi\ne\pm1$. Thus $\theta\in\frac\pi2\Z$, and $(JP)^4=I$ on every block.

Now write $J=\sum_uc_uP_u$. For $v\in V_n$,
\[
4^{-n}\sum_{w\in V_n}(-1)^{\omega(v,w)}P_wJP_w=c_vP_v,
\]
and $J$ commutes with each $P_wJP_w$, hence with $c_vP_v$. If $c_v\ne0$, then $J$ commutes with $P_v$, which forces $c_u=0$ whenever $\omega(u,v)=1$. Therefore $\operatorname{supp}_PJ$ is isotropic.
\end{proof}

Lastly, the next proposition turns a GSC pair into a bound on the level $k$ in the Clifford hierarchy. We will use it for $n=4$ and $m=2$ at the end of Appendix~\ref{app:four-qubit}. For $d\ge1$, the \emph{$d$-fold corrections} of a gate $U$ are the corrections of its $(d-1)$-fold corrections, where the $0$-fold correction of $U$ is $U$ itself. By the definition of the hierarchy, a gate $U$ lies in $\C_j(n)$ as soon as every $d$-fold correction of $U$ lies in $\C_{j-d}(n)$. A monomial gate $N\ket x=\alpha_x\ket{\pi x}$ has permutation $\pi$, and products and inverses of monomial gates are monomial, with the product and inverse permutations.

\begin{proposition}\label{lem:level-bound}
Let $n\in\{3,4\}$, $U\in\C(n)$ be generalized semi-Clifford, and $m\ge2$ be such that every correction $h_v=UP_vU^\dagger$ is defined over $K_m$ up to scalar. Then $U\in\C_{m+2n-4}(n)$, that is, $U\in\C_{m+2}(3)$ or $U\in\C_{m+4}(4)$.
\end{proposition}
\begin{proof}
By Proposition~\ref{lem:tool-monomial}, $U\doteq C_1MC_2$ with $C_1,C_2\in\C_2(n)$ and $M=\Pi_\pi D$, where $D$ is diagonal and $\pi$ preserves a complete affine flag. Every $d$-fold correction of $U$ is the conjugate by $C_1$ of a $d$-fold correction of $M$, and Clifford conjugation preserves each level and, up to scalar, the field of definition, since Clifford gates are defined over $\Q(i)$ up to scalar. We may therefore assume $U=M$.
\begin{itemize}
\item \step{Step 1:}{$\Pi_\pi\in\C_n(n)$.} For $n=3$, $\pi$ is $1$ or $(x,y,z)\mapsto(x,y,z+xy)$ by Proposition~\ref{lem:tool-monomial}(1), and the Toffoli gate is third-level. For $n=4$, after an affine change of coordinates, which is a Clifford conjugation, the flag is $\Span\{e_4\}<\Span\{e_3,e_4\}<\Span\{e_2,e_3,e_4\}$ and $\pi(x,y,z,w)=(x+c,\,y+f_2(x),\,z+f_3(x,y),\,w+f_4(x,y,z))$. Composing with an affine permutation on the left, which is Clifford, we may assume
\[
 \pi(x,y,z,w)=(x,\,y,\,z+\epsilon xy,\,w+f(x,y,z)),\qquad\epsilon\in\F_2 .
\]
For $a\in\F_2^4$ let $t_a$ be the translation by $a$ and $\rho_a=\pi t_a\pi^{-1}$. A direct computation gives
\[
 \rho_a(x,y,z,w)=\bigl(x+a_1,\;y+a_2,\;z+a_3+\epsilon(a_1y+a_2x+a_1a_2),\;w+a_4+q_a(x,y,z)\bigr),
\]
with $q_a(x,y,z)=(\Delta_{(a_1,a_2,a_3)}f)(x,y,z+\epsilon xy)$ of degree at most two: a derivative of $f$ has degree at most two, and substituting $z+\epsilon xy$ for $z$ keeps the degree. The correction of $\Pi_\pi$ at $X^aZ^b$ is the monomial gate $\ket p\mapsto(-1)^{g(p)}\ket{\rho_ap}$ with $g=b\cdot\pi^{-1}$, and $g(x,y,z,w)=b_4w+g_0(x,y,z)$ with $\deg g_0\le3$. Its correction at $X^{a'}Z^{b'}$ has permutation $\rho_at_{a'}\rho_a^{-1}$, which is affine because the first three coordinates of $\rho_a$ are affine and $\deg q_a\le2$, and signs $(-1)^{g(u)+g(u+a')+b'\cdot u}$ with $u=\rho_a^{-1}(p)$, a function of degree at most two in $p$. An affine permutation gate with a sign function of degree at most two is Clifford, as in the proof of Proposition~\ref{prop:sc-controlled-pauli}. So every correction of $\Pi_\pi$ is third-level, and $\Pi_\pi\in\C_4(4)$.

\item \step{Step 2:}{every non-scalar $(n-2)$-fold correction $H$ of $M$ has the form $H\ket x=\eta\,\zeta^{\varphi(x)}\ket{\rho x}$, where $\zeta=\zeta_{2^m}$, $\varphi:\F_2^n\to\Z/2^m$, $\rho$ is affine, $\eta$ is a scalar, and $H$ is a traceless Hermitian involution.} As $n-2\ge1$, $H=gPg^\dagger$ with $g$ an $(n-3)$-fold correction of $M$ and $P\ne\pm I$ a Hermitian Pauli, so $H$ is a traceless Hermitian involution. It is monomial, and its permutation is that of the $(n-2)$-fold correction of $\Pi_\pi$ at the same Paulis, which lies in $\C_2(n)$ by Step~1. A monomial gate in $\C_2(n)$ maps each $X^a$ to a Pauli, so its permutation is affine. The corrections of $M$ are defined over $K_m$ up to scalar, and if $g\doteq B$ with $B$ over $K_m$, then $gPg^\dagger=BPB^{-1}$ is over $K_m$. Hence the ratios of the entries of $H$ lie in $K_m$. They are roots of unity of $2$-power order~\cite[Appendix~F]{Anderson_2024}, hence lie in $\mu_{2^m}$.

\item \step{Step 3:}{every such $H$ lies in $\C_{m+n-2}(n)$.} Since $H$ is a traceless Hermitian involution on $2^n\ge8$ dimensions, $\det H=1$. From $H^2=I$, $\varphi(x)+\varphi(\rho x)=c$ is constant and $\eta^2=\zeta^{-c}$. Affine permutations of $\F_2^n$, $n\ge3$, are even. Hence $1=\det H=\eta^{2^n}\zeta^{\sum_x\varphi(x)}=\zeta^{\sum_x\varphi(x)-2^{n-1}c}$, and $\sum_x\varphi(x)$ is even. Expand $\varphi=\sum_{A\subseteq\{1,\dots,n\}}c_Ax_A$ with $c_A\in\Z/2^m$ and $x_A=\prod_{i\in A}x_i$. The top coefficient $c_{\{1,\dots,n\}}=\sum_x(-1)^{n-|x|}\varphi(x)$ is even. We use the following bound on diagonal gates: $\operatorname{diag}(\zeta^G)\in\C_j(n)$ whenever every coefficient $c_A$, $A\ne\varnothing$, of $G$ is divisible by $2^{m+|A|-1-j}$. Indeed, such a gate is a product of gates $\operatorname{diag}(\zeta_{2^{m'}}^{x_A})^{c'}$ with $m'+|A|-1\le j$, each of which lies in $\C_{m'+|A|-1}(n)\subseteq\C_j(n)$, and the diagonal gates in $\C_j(n)$ form a group~\cite[Lemma~3 and Theorem~1]{CGK2017}.

For $\varphi$ and $j=m+n-2$, the exponent $m+|A|-1-j$ is at most $0$ when $|A|\le n-1$ and equals $1$ when $|A|=n$, where $c_A$ is even. So $\operatorname{diag}(\zeta^\varphi)\in\C_{m+n-2}(n)$, and $H=\eta\,\Pi_\rho\operatorname{diag}(\zeta^\varphi)$ with $\Pi_\rho$ Clifford lies in $\C_{m+n-2}(n)$.
\end{itemize}
Scalar $(n-2)$-fold corrections are Paulis, so every $(n-2)$-fold correction of $M$ lies in $\C_{m+n-2}(n)$, and $M\in\C_{m+2n-4}(n)$.
\end{proof}

\Needspace{14\baselineskip}
\section{\texorpdfstring{$\mathcal{C}(3)\subseteq\mathsf{GSC}(3)$}{C(3) in GSC(3)}}\label{app:three-qubit}

\subsection{Proof of Theorem~\ref{thm:n3-gaussian} when \texorpdfstring{$n=3$}{n=3}}\label{gauss-app:ge3}

We first reduce Theorem~\ref{gauss-thm:rigidity} to gates whose corrections are also generalized semi-Clifford.

\begin{lemma}\label{app:lem-ge-equiv}

Let $1\le n\le4$, and suppose that $U\in\GSC(n)$ for every $U\in\C(n)$ whose corrections $h_v$, $v\in V_n$, are all Gaussian and generalized semi-Clifford. Then every Gaussian gate in $\C(n)$ is generalized semi-Clifford. Consequently, if every correction of a gate $U\in\C(n)$ is Gaussian, then every correction of $U$ is generalized semi-Clifford and $U\in\GSC(n)$, that is, the first assertion of Theorem~\ref{gauss-thm:rigidity} holds for this $n$.

\end{lemma}
\begin{proof}

We prove by induction on $k$ that every Gaussian gate in $\C_k(n)$ is generalized semi-Clifford. For $k\le2$ the gate is Clifford, hence generalized semi-Clifford. Let $k\ge3$ and let $V\in\C_k(n)$ be Gaussian, say $V\doteq V_0$ with $V_0$ over $\Q(i)$. Then each $VP_vV^\dagger=V_0P_vV_0^{-1}\in\C_{k-1}(n)$ is Gaussian, hence generalized semi-Clifford by induction, and the hypothesis gives $V\in\GSC(n)$.

Now let every correction $h_v$ of $U\in\C(n)$ be Gaussian. Each $h_v$ is a Gaussian gate in $\C(n)$, hence generalized semi-Clifford, and the hypothesis gives $U\in\GSC(n)$.

\end{proof}

Now by Lemma~\ref{app:lem-ge-equiv}, for $n=3$ it suffices to show that $U\in\GSC(3)$ whenever $h_v$ is Gaussian and generalized semi-Clifford for all $v\in V_3$. From now on in this subsection, let $U\in\C(3)$ and assume $h_v$ is Gaussian and generalized semi-Clifford for all $v\in V_3$. In particular $h_v\in\mathscr R_3$ for $v\ne0$. We want to show that $h_v\in\C_3(3)$ for all $v$, which gives $U\in\C_4(3)\subseteq\GSC(3)$ by Theorem~\ref{thm:main-k4}.

\begin{proposition}\label{gauss-lem:exceptional}
Let $h\in\mathscr R_3$. Then $h\in\C_4(3)$. Suppose moreover that $h\notin\C_3(3)$.
\renewcommand{\labelenumi}{(\theenumi)}
\begin{enumerate}
\item There is a nonzero label $r$ with $hP_rh^\dagger=P_r$, and $r$ is the only nonzero label $v$ with $hP_vh^\dagger=\pm P_v$. Moreover $\mathsf C(h)=r^\perp$, and every non-Pauli correction of $h$ at a label in $r^\perp$ has the form $\frac12(\pm P_q\pm P_{q+r}\pm P_{q'}\pm P_{q'+r})$ for some labels $q,q'$.
\item $\mathscr L(h)=\{(L,L)\}$ for a Lagrangian $L\le V_3$, and $r\in L$.
\item Exactly two $E\in\A_L$ satisfy $E^2=E=E^\dagger$, $\tr E=1$ and $hEh^\dagger=E$.
\end{enumerate}
\end{proposition}

\begin{proof}

By Proposition~\ref{app:lem-R_n}, $h$ is Clifford-conjugate to a gate $M$ of the form~\eqref{core:eq:row}. Conjugating by an affine permutation, we may take the flag in that proposition to be $0<\Span\{e_z\}<\Span\{e_y,e_z\}<\F_2^3$.

We use \textup{(\ref{cert:three-qubit-flag-preserving-involution-rays})} to show that $M\in \C_3(3)$ or $M$ is Clifford-conjugate up to scalar to $H_*$, which is defined as follows: label the three qubits $x,y,z$ and let
\[
H_*=\diag(A,B),\qquad \text{where }
A=\frac{I+Z_z}{2}X_y+\frac{I-Z_z}{2}Y_y, \quad
B=\frac{I+Z_y}{2}Z_z+\frac{I-Z_y}{2}X_z.
\]
$A$ and $B$ are Clifford by Proposition~\ref{lem:controlled-clifford} with $k=2$, and they are Hermitian and unitary, hence involutions. The gate $M$ is determined by $(\lambda,f,\pi)\in\F_2\times(\Z/4)^{\F_2^3}\times Q_3$, where $Q_3$ is the group of the $2^7$ permutations preserving the flag (Lemma~\ref{lem:tool-flags}(1)). Whether $M\in \C_3(3)$ or not is decided on a basis of $V_3$, and the affine--diagonal Clifford gates are finitely many modulo scalars. Enumerating these gates gives the following.

\begin{finitefact}{cert:three-qubit-flag-preserving-involution-rays}{Three-qubit involutions with a flag-preserving permutation}{Proposition~\ref{gauss-lem:exceptional}}
Let $M$ be a gate of the form~\eqref{core:eq:row} on three qubits whose permutation preserves the flag $0<\Span\{e_z\}<\Span\{e_y,e_z\}<\F_2^3$, and let $H_*$ be as in the proof of Proposition~\ref{gauss-lem:exceptional}. Then $[M]\in\C_3(3)$, or $[M]$ is conjugate to $[H_*]$ by an affine--diagonal Clifford gate.
\tcblower
\textbf{Code.} \certlink{three-qubit-flag-preserving-involution-rays}
\end{finitefact}

 Moreover $H_*=\diag(I,BA)\,(I\otimes A)$. Let $r=z_x$ and $L_0=\Span\{z_x,z_y,z_z\}$. Identify $r^\perp/\Span\{r\}$ with the label space $V_2$ of qubits $y,z$, write $p\mapsto\bar p$ for the quotient map, and let $\alpha,\beta\in\Sp(V_2)$ be the symplectic actions of $A,B$. Set $J=\alpha+\beta$ and $N=\alpha\beta+I$. In coordinates $(s,t,u,v)$ with respect to $x_y,x_z,z_y,z_z$,
\[
J(s,t,u,v)=(0,s,s+v,t),\qquad
N(s,t,u,v)=(0,s,v,s+t),\qquad
N^3(s,t,u,v)=(0,0,s,0),\qquad N^4=0.
\]
Since $(BA)^4=-I$ and $(BA)^2\notin\mu_4\mathcal P(2)$, Proposition~\ref{lem:controlled-clifford} gives $\diag(I,BA)\in\C_4(3)\setminus\C_3(3)$, and hence, as right multiplication by a Clifford preserves each level, $H_*\in\C_4(3)\setminus\C_3(3)$. Thus $h\in\C_4(3)$.

Now suppose $h\notin\C_3(3)$. Then $h=\pm CH_*C^\dagger$ for a Clifford $C$ (the scalar is $\pm1$ because both gates are Hermitian involutions). Assertions (1)--(3) are invariant under this conjugation, with $r$, $L$ and $E$ replaced by their images under $C$, so it suffices to prove them for $H_*$, with $r=z_x$ and $L=L_0$.

\begin{itemize}
\item \textbf{Proof of (1).}
\begin{itemize}
\item \emph{The corrections $K_p=H_*P_pH_*$ at $p\in r^\perp$:} Write $P_p=Z_x^c\otimes P_{\bar p}$. Then
\[
K_p=Z_x^c\,\diag\bigl(\pm P_{\alpha\bar p},\pm P_{\beta\bar p}\bigr)
=\pm Z_x^c\,\diag\bigl(I,\pm P_{\beta\bar p}P_{\alpha\bar p}\bigr)\bigl(I\otimes P_{\alpha\bar p}\bigr),
\]
which is a Clifford by Proposition~\ref{lem:controlled-clifford} with $k=2$. As $\alpha\bar p+\beta\bar p=J\bar p$, it maps $P_w$ to $\pm P_{w+\omega(J\bar p,\bar w)r}$ for $w\in r^\perp$. It is a Pauli exactly when $\alpha\bar p=\beta\bar p$, that is, $\bar p\in\ker J=\Span\{z_y\}$. Otherwise $K_p=\frac12(\pm P_q\pm P_{q+r}\pm P_{q'}\pm P_{q'+r})$ for some labels $q,q'$.
\item \emph{$x_x\notin\mathsf C(H_*)$:} We have $H_*X_xH_*=X_x\diag(BA,AB)=X_x\diag\bigl(I,(AB)^2\bigr)(I\otimes BA)$, and $(AB)^2=((BA)^2)^{-1}=-(BA)^2\notin\mu_4\mathcal P(2)$. Hence $H_*X_xH_*$ is not a Clifford gate by Proposition~\ref{lem:controlled-clifford}.
\end{itemize}
The subspace $\mathsf C(H_*)$ contains the hyperplane $r^\perp$ but not $x_x$, so $\mathsf C(H_*)=r^\perp$, and then $\mathsf P(H_*)=\Span\{r,z_y\}$. A label $p$ with $H_*P_pH_*=\pm P_p$ lies in $\mathsf P(H_*)$. Since $H_*Z_xH_*=Z_x$ and $H_*Z_yH_*=-Z_xZ_y$, the label $r$ is the only nonzero one.

\item \textbf{Proof of (2).} Let $(L_1,L_2)\in\mathscr L(H_*)$. As $H_*$ is an involution, also $(L_2,L_1)\in\mathscr L(H_*)$, so for $p\in r^\perp$ the Clifford $K=H_*P_pH_*$ satisfies
\[
K\A_{L_i}K^\dagger=H_*P_p\A_{L_{3-i}}P_pH_*=\A_{L_i}\qquad(i=1,2).
\]
\begin{itemize}
\item \emph{$r\in L_i$:} Otherwise $L_i\cap r^\perp$ is two-dimensional and injects into $V_2$. For $w$ in it, $w+\omega(J\bar p,\bar w)r\in L_i$ forces $\omega(J\bar p,\bar w)=0$ for every $p$, so its image lies in $(\operatorname{im}J)^\perp$, which is one-dimensional because $\operatorname{rank}J=3$, a contradiction.
\item \emph{$L_i=L_0$:} Now $L_i=\Span\{r\}\oplus\bar L_i$ with Lagrangians $\bar L_i\le V_2$, and $\A_{L_i}=\diag(\A_{\bar L_i},\A_{\bar L_i})$, so $H_*\A_{L_1}H_*=\A_{L_2}$ reads
\[
\alpha\bar L_1=\beta\bar L_1=\bar L_2,\qquad\text{hence}\qquad
N\bar L_1\subseteq\bar L_1,
\]
since $\alpha^2=\beta^2=I$. As $N^3\ne0=N^4$, $N$ is a single nilpotent Jordan block, and its only invariant two-dimensional subspace is $\ker N^2=\Span\{z_y,z_z\}$, which $\beta$ fixes. Hence $L_1=L_2=L_0$.
\end{itemize}
Conversely $(L_0,L_0)\in\mathscr L(H_*)$, because $H_*$ is monomial. Finally $r=z_x\in L_0$.

\item \textbf{Proof of (3).} The rank-one projections in $\A_{L_0}$ are $|x,y,z\rangle\langle x,y,z|$, and $H_*$ permutes the computational basis by $(x,y,z)\mapsto(x,y+x+1,z+xy)$, whose fixed points are exactly $(1,0,z)$, so the two projections are $|1,0,z\rangle\langle1,0,z|$, $z\in\F_2$.
\end{itemize}
\end{proof}

We call a gate $t\in\mathscr R_3\setminus\C_3(3)$ \emph{exceptional}. We write $L_t$ for its Lagrangian $L$ and call its label $r_t=r$ from Proposition~\ref{gauss-lem:exceptional}(1) the \emph{root} of $t$. By the proof of that proposition, $t=\pm CH_*C^\dagger$ for a Clifford gate $C$, and then $\gamma_Cz_x=r_t$ and $\gamma_CL_0=L_t$.

\begin{lemma}\label{gauss-lem:three-correction-compatibility}
Let $s,v\in V_3\setminus\{0\}$ with $s\ne v$ and $h_s\notin\C_3(3)$, and let $\mathscr L(h_s)=\{(S,S)\}$ \textup{(}Proposition~\ref{gauss-lem:exceptional}\textup{)}. Then
\[
h_v\mathcal A_Sh_v^\dagger=\mathcal A_S.
\]
\end{lemma}

\begin{proof}

Conjugating every $h_w$ by one Clifford gate, that is, replacing $U$ by $CU$, changes neither the hypotheses nor the conclusion. By the proof of Proposition~\ref{gauss-lem:exceptional} we may therefore assume $h_s=\pm H_*$ and $S=L_0$. Write $H=h_s$, $h=h_v$, $g=h_{s+v}$, and $r=z_x=r_H$. Then $Hh=\pm hH$, $g\doteq Hh\doteq hH$, and $h,g\in\C_4(3)$ by Proposition~\ref{gauss-lem:exceptional}.

It suffices to show $h\A_{L_0}h=\A_{L_0}$. Since $H$ preserves $\A_{L_0}$ and $g\doteq Hh$, this is equivalent to $g\A_{L_0}g^\dagger=\A_{L_0}$, so the roles of $h$ and $g$ are symmetric. For every exceptional $t$, $\mathscr L(t)=\{(L_t,L_t)\}$, so Lemma~\ref{lem:tool-transport}(2) gives
\begin{equation}\tag{$\ast$}\label{eq:unique-witness-normalizes}
K,K'\in\C_2(3),\quad tK\doteq K't
\quad\Longrightarrow\quad K\A_{L_t}K^\dagger=\A_{L_t}.
\end{equation}
We split on whether $h,g$ are third-level; otherwise, by symmetry, $h$ is exceptional with root $r'=r_h$, and we compare $r'$ with $r$.

\begin{inlinefigure}
\begin{tikzpicture}
\node[pfnode] (s) at (0,0) {$h,g$};
\node[pfnode] (c1) at (3.8,1.1) {Case 1};
\node[pfnode] (c2) at (3.8,-1.2) {Case 2};
\node[pfnode] (c2a) at (8.6,0.2) {Case 2a};
\node[pfnode] (c2b) at (8.6,-1.2) {Case 2b};
\node[pfnode] (c2c) at (8.6,-2.6) {Case 2c};
\node[pfgoal] (goal) at (12.4,0.2) {$h\A_{L_0}h=\A_{L_0}$};
\node[pfgoal] (bot) at (12.4,-2.6) {Contradiction};
\draw[pfedge] (s) -- node[pflab,sloped,above] {$h,g\in\C_3(3)$} (c1);
\draw[pfedge] (s) -- node[pflab,sloped,below] {$h$ exceptional} (c2);
\draw[pfedge] (c2) -- node[pflab,above left,pos=.8] {$r'=r$} (c2a);
\draw[pfedge] (c2) -- node[pflab,above,pos=.55] {$r'\ne r$, $\omega(r,r')=0$} (c2b);
\draw[pfedge] (c2) -- node[pflab,below left,pos=.8] {$\omega(r,r')=1$} (c2c);
\draw[pfedge] (c1) -| (goal);
\draw[pfedge] (c2a) -- (goal);
\draw[pfedge] (c2b) -- (goal);
\draw[pfedge] (c2c) -- (bot);
\end{tikzpicture}
\caption{Case analysis in the proof of Lemma~\ref{gauss-lem:three-correction-compatibility}.}
\end{inlinefigure}

\begin{itemize}
\item \step{Case 1:}{$h,g\in\C_3(3)$.} For $p\in V_3$, apply~\eqref{eq:unique-witness-normalizes} with $t=H$, $K=hP_ph$, and $K'=gP_pg^\dagger=HKH$, which are Clifford. Thus every correction of $h$ preserves $\A_{L_0}$, and Lemma~\ref{lem:main-local-global} gives a Lagrangian $M$ with $h\A_{L_0}h=\A_M$. Since $Hh=\pm hH$ and $H$ preserves $\A_{L_0}$, we get $H\A_MH=hH\A_{L_0}Hh=\A_M$, and Proposition~\ref{gauss-lem:exceptional}(2) gives $M=L_0$.

\item \step{Case 2a:}{$h$ is exceptional and $r'=r$.} Write $h=\pm CH_*C^\dagger$ with $C$ Clifford. As $\gamma_Cz_x=r'=z_x$, $C$ is block diagonal or block antidiagonal with Clifford blocks, so $h=\diag(G_0,G_1)$, where the symplectic actions $\phi_0,\phi_1\in\Sp(V_2)$ of $G_0,G_1$ are conjugate to $\alpha,\beta$ in some order. Since $\A_{L_0}=\diag(\A_{\bar L_0},\A_{\bar L_0})$, it suffices to show that $\phi_0$ and $\phi_1$ preserve $\bar L_0=\Span\{z_y,z_z\}$. From $Hh=\pm hH$ we get $AG_0=\pm G_0A$ and $BG_1=\pm G_1B$, so $\phi_0\alpha=\alpha\phi_0$ and $\phi_1\beta=\beta\phi_1$. Moreover
\[
h\doteq\diag(I,G_1G_0)(I\otimes G_0),\qquad
g\doteq\diag(I,BG_1G_0A)(I\otimes AG_0)
\]
lie in $\C_4(3)$, so Proposition~\ref{lem:controlled-clifford} with $k=4$ gives $(\phi_1\phi_0)^4=(\beta\phi_1\phi_0\alpha)^4=1$. Taking inverses and using $\phi_1\beta=\beta\phi_1$, both $\phi_0\phi_1$ and $\alpha\phi_0\beta\phi_1$ have order dividing $4$.

Let $\mathcal E$ be the space of even subsets of $\Omega=\{1,\dots,6\}$ under symmetric difference. Then $\omega(a,b)=|a\cap b|\bmod2$ is a symplectic form on $\mathcal E/\Span\{\Omega\}$, and each nonzero class contains exactly one pair $\{i,j\}$, written $ij$. Identify $V_2$ with $\mathcal E/\Span\{\Omega\}$ by $x_y,x_z,z_y,z_z\mapsto25,34,56,13$. Permuting $\Omega$ then gives $\Sp(V_2)\cong S_6$, under which
\[
\alpha=(24)(56),\qquad\beta=(14)(23),\qquad\bar L_0=\{0,13,24,56\}.
\]
Under this isomorphism $\phi_0,\phi_1$ are double transpositions. The five commuting with $\alpha$, listed below, all preserve $\bar L_0$. Of the five commuting with $\beta$, only $(14)(56)$ and $(23)(56)$ do not; conjugation by $(13)(24)$ exchanges these two and fixes $\alpha$, $\beta$ and $\bar L_0$. So if $\phi_1$ does not preserve $\bar L_0$, we may assume $\phi_1=(14)(56)$, and then every $\phi_0$ violates an order condition:
\[
\begin{array}{c|ccccc}
\phi_0&(13)(24)&(13)(56)&(24)(56)&(25)(46)&(26)(45)\\ \hline
\operatorname{ord}(\phi_0\phi_1)&4&3&3&5&5\\
\operatorname{ord}(\alpha\phi_0\beta\phi_1)&3&4&2&5&5
\end{array}
\]
Hence $\phi_0$ and $\phi_1$ preserve $\bar L_0$.

\item \step{Case 2b:}{$h$ is exceptional, $r'\ne r$ and $\omega(r,r')=0$.} Since $h=\pm CH_*C^\dagger$ for a Clifford gate $C$ with $\gamma_CL_0=L_h$, the map $K\mapsto CKC^\dagger$ is a bijection from the Hermitian Clifford gates commuting with $H_*$ onto those commuting with $h$, and $\operatorname{supp}_P(CKC^\dagger)=\gamma_C\operatorname{supp}_PK$. A Hermitian $K\in\C_2(3)$ commuting with $H_*$ commutes with $H=\pm H_*$, so~\eqref{eq:unique-witness-normalizes} with $t=H$ and $K'=K$ gives $K\A_ZK^\dagger=\A_Z$, as $L_0=L_Z$. Hence $K$ is monomial: modulo scalars, one of the $1{,}344$ affine permutation gates times one of the $512$ diagonal Clifford gates, and testing all of them gives the following.

\begin{finitefact}{cert:cliffords-commuting-with-hstar}{Cliffords commuting with $H_*$}{Lemma~\ref{gauss-lem:three-correction-compatibility}}
With $H_*$ as above, a Hermitian gate $K\in\C_2(3)$ commutes with $H_*$ if and only if $K=\pm(-1)^{ax+bxy+c(1+x)z}$ with $a,b,c\in\F_2$. The traceless ones are $\pm Z_x$ and $\pm$ two non-Pauli gates whose supports span $L_Z$.

\tcblower

\textbf{Code.} \certlink{cliffords-commuting-with-hstar}
\end{finitefact}
By \textup{(\ref{cert:cliffords-commuting-with-hstar})} and this reduction, for every non-Pauli traceless Hermitian Clifford gate $K$, $\operatorname{supp}_PK$ spans $L_0$ if $K$ commutes with $H$, and $L_h$ if $K$ commutes with $h$. As $\omega(r,r')=0$, Proposition~\ref{gauss-lem:exceptional}(1) gives $r\in\mathsf C(h)$ and $r'\in\mathsf C(H)$, so $T=hP_rh$ and $T'=HP_{r'}H$ are traceless Hermitian Clifford gates. Since $Hh=\pm hH$, $HP_rH=P_r$ and $hP_{r'}h=P_{r'}$, $T$ commutes with $H$ and $T'$ commutes with $h$. If $T$ were a Pauli, then $T=\pm P_r$, as $r$ is the only nonzero label $v$ with $HP_vH=\pm P_v$; so $hP_rh=\pm P_r$, and $r=r'$ by the same uniqueness for $h$, a contradiction. Symmetrically, $T'$ is not a Pauli. Hence $\operatorname{supp}_PT$ spans $L_0$ and $\operatorname{supp}_PT'$ spans $L_h$. By Proposition~\ref{gauss-lem:exceptional}(1) for $H$, $\operatorname{supp}_PT'=\{q,q+r,q',q'+r\}$ for some labels $q,q'$, so $r=q+(q+r)\in L_h$. Then $P_r\in\A_{L_h}$, so $T\in h\A_{L_h}h=\A_{L_h}$ as $(L_h,L_h)\in\mathscr L(h)$, and $L_0\subseteq L_h$. Thus $L_h=L_0$ and $h\A_{L_0}h=\A_{L_0}$.

\item \step{Case 2c:}{$h$ is exceptional and $\omega(r,r')=1$.} Since $g\doteq hH\doteq Hh$, $HP_rH=P_r$ and $hP_{r'}h=P_{r'}$,
\[
gP_rg^\dagger=hP_rh,\qquad gP_{r'}g^\dagger=HP_{r'}H.
\]
Neither is Clifford, by Proposition~\ref{gauss-lem:exceptional}(1), as $r\notin r'^\perp=\mathsf C(h)$ and $r'\notin r^\perp=\mathsf C(H)$. So $g\notin\C_3(3)$, and as $g=h_{s+v}\in\mathscr R_3$, $g$ is exceptional. Let $r''=r_g$. As $r,r'\notin\mathsf C(g)=r''^\perp$, we have $\omega(r,r'')=\omega(r',r'')=1$, so $\Lambda=r^\perp\cap r''^\perp$ has dimension four. For $p\in\Lambda$, both $K_p=HP_pH$ and $hK_ph=gP_pg^\dagger$ are Clifford, as $p\in\mathsf C(H)\cap\mathsf C(g)$, so~\eqref{eq:unique-witness-normalizes} with $t=h$, $K=K_p$ and $K'=hK_ph$ gives $K_p\A_{L_h}K_p^\dagger=\A_{L_h}$. The map $p\mapsto J\bar p$ on $r^\perp$ has kernel $\mathsf P(H)=\Span\{r,z_y\}$, which meets $\Lambda$ in a line since $\omega(r,r'')=1$; hence it maps $\Lambda$ onto $\im J$. The argument for $r\in L_i$ in the proof of Proposition~\ref{gauss-lem:exceptional}(2) uses only that $K_p$ preserves $\A_{L_i}$ for $p$ with $J\bar p$ ranging over $\im J$, so it gives $r\in L_h$. Also $r'\in L_h$ by Proposition~\ref{gauss-lem:exceptional}(2), contradicting $\omega(r,r')=1$.
\end{itemize}

\end{proof}

\begin{proposition}\label{gauss-prop:n3-gsc}
Let $U\in\C(3)$, and suppose that every $h_v$, $v\in V_3$, is Gaussian and generalized semi-Clifford. Then $h_v\in\C_3(3)$ for every $v$, and $U\in\C_4(3)\cap\GSC(3)$.
\end{proposition}

\begin{proof}
Suppose that $h_s\notin\C_3(3)$ for some $s\in V_3$. Then $s\ne0$, and $h_w\in\mathscr R_3$ for every $w\ne0$. Let $\mathscr L(h_s)=\{(S,S)\}$ (Proposition~\ref{gauss-lem:exceptional}(2)). Then $h_v\A_Sh_v^\dagger=\A_S$ for every $v\in V_3$: for $v\notin\{0,s\}$ by Lemma~\ref{gauss-lem:three-correction-compatibility}, for $v=s$ since $(S,S)\in\mathscr L(h_s)$, and trivially for $v=0$.

Let $E,F$ be the two rank-one orthogonal projections in $\A_S$ fixed by $h_s$ (Proposition~\ref{gauss-lem:exceptional}(3)). They are distinct and commute, so $EF=0$, and $R=E+F$ is a projection with $\tr R=2$. For every $v$, since $h_sh_v=\pm h_vh_s$,
\[
h_s(h_vEh_v^\dagger)h_s^\dagger
=h_v(h_sEh_s^\dagger)h_v^\dagger
=h_vEh_v^\dagger,
\]
and likewise for $F$. Conjugation by $h_v$ therefore maps $E$ and $F$ to distinct rank-one orthogonal projections in $\A_S$ fixed by $h_s$. Hence it permutes $\{E,F\}$, and $h_vRh_v^\dagger=R$. Thus $U^\dagger RU$ commutes with every $P_v$, so it is scalar, and $\tr R=2$ gives $R=I/4$, contradicting $R^2=R$. Hence $h_v\in\C_3(3)$ for every $v$, so $U\in\C_4(3)$, and $U\in\GSC(3)$ by Theorem~\ref{thm:main-k4}.
\end{proof}

Proposition~\ref{gauss-prop:n3-gsc} proves Theorem~\ref{thm:n3-gaussian} for $n=3$: if every $h_v$ is Gaussian, Lemma~\ref{app:lem-ge-equiv} makes every $h_v$ generalized semi-Clifford, and Proposition~\ref{gauss-prop:n3-gsc} gives $h_v\in\C_3(3)$ and $U\in\C_4(3)\subseteq\GSC(3)$.

\subsection{Proof of Theorem~\ref{thm:conductor-descent}}

In this subsection, let $U\in\C(3)$, let $M=\operatorname{cond}(U)$ and $m=\max_v\operatorname{cond}(h_v)$, so that $m\in\{M-1,M\}$ by Proposition~\ref{lem:conductor-plane}, and suppose $m>2$, that is, some $h_v$ is non-Gaussian. As in Theorem~\ref{thm:conductor-descent}, we also assume $h_v\in\GSC(3)$ for all $v$. These assumptions hold for all the lemmas in this subsection. We want to find $U'\in\C(3)$ with
\[
\operatorname{cond}(U')<M,\qquad
U'\in\GSC(3)\Longleftrightarrow U\in\GSC(3).
\]
If $U\in\GSC(3)$, take $U'=I$. Then $\operatorname{cond}(U')=2<M$, since $m>2$ and $m\leq M$. Henceforth we assume $U\notin\GSC(3)$.

Choose $W\doteq U$ over $K_M$, and let $\sigma$ be a generator of the cyclic group $\mathrm{Gal}(K_M/K_{m-1})$: explicitly, $\sigma(\zeta_{2^M})=-\zeta_{2^M}$ if $m=M$, and $\sigma(\zeta_{2^M})=i\zeta_{2^M}$ if $m=M-1$. Since $m>2$, $\sigma$ fixes $K_2=\Q(i)$, and its restriction to $K_m$ generates $\mathrm{Gal}(K_m/K_{m-1})$.

We first show that the Galois cocycle $W^{-1}\sigma(W)$ of $U$ is, in suitable coordinates, $\diag_xi^{p\cdot x}(-1)^{q(x)}$ with $\deg q\le2$ (Lemma~\ref{lem:galois-normal-form}), and then that $U$ maps $Z^p$ to a Pauli (Lemma~\ref{lem:m-M-1}). Multiplying $U$ by a rotation about $Z^p$ then lowers the conductor (Lemma~\ref{lem:conductor-step}).

\begin{lemma}\label{lem:galois-normal-form}
After replacing $U,W$ by $UC,WC$ for some $C\in\C_2(3)$, there are $p\in\mathbb F_2^3$ and a Boolean function $q$ of degree at most two such that
\[
W^{-1}\sigma(W)\doteq\diag_{x\in\mathbb F_2^3}
 i^{p\cdot x}(-1)^{q(x)},
\]
where $p\cdot x\in\{0,1\}$ is the dot product over $\F_2$. If $m=M$, then $p=0$; if $m=M-1$, then $p\ne0$. Moreover, write the polar form of $q$ as $q(x+a)+q(x)+q(a)+q(0)=a^{\mathsf T}B_qx$, and let $T=B_q+pp^{\mathsf T}$ and $g_c=UZ^cU^\dagger$. Then $T\ne0$ is symmetric with $\operatorname{diag}T=p\in\operatorname{im}T$. For $0\ne c=Ta$ and every $b\in\F_2^3$, the gate $g_c$ is a Clifford involution, and $g_c\in\A_L\cap\A_{L'}$ for every $(L,L')\in\mathscr L(h_{(a,b)})$. If $g_c$ is not Pauli, its support spans a Lagrangian $L_c$, and $\mathscr L(h_{(a,b)})=\{(L_c,L_c)\}$.
\end{lemma}

\begin{proof}

Each $h_v=WP_vW^{-1}$ has entries in $K_M$. Denote $d_v=h_v^{-1}\sigma(h_v)$.

We first recall a simple fact (i): a commutative algebra $\mathfrak B\subseteq\operatorname{Mat}_8(\CC)$ normalized by every Pauli is spanned by the Paulis it contains, and their labels form an isotropic subspace of $V_3$. Indeed, for $B=\sum_uc_uP_u\in\mathfrak B$ the average $4^{-3}\sum_w(-1)^{\omega(v,w)}P_wBP_w^\dagger=c_vP_v$ lies in $\mathfrak B$. Hence these Paulis lie in $\mathfrak B$, and they form a commutative group.

\begin{itemize}
\item \step{Step 1:}{$d_v=\beta_vy_v$ with $\beta_v\in\mu_4$, $y_v^2=I$, and $\sigma(d_v)=d_v$. Moreover $y_v\in\A_L\cap\A_{L'}$ for every $(L,L')\in\mathscr L(h_v)$.} Let $(L,L')\in\mathscr L(h_v)$, and write $h_v\doteq C_1\Pi DC_2$ as in Proposition~\ref{lem:tool-monomial}, with Clifford gates $C_1,C_2$ such that $C_2\A_LC_2^\dagger=\A_Z$, a permutation matrix $\Pi$, and $D$ diagonal. By Proposition~\ref{lem:tool-galois}(1), $\sigma(C_j)\doteq C_j$. By Proposition~\ref{lem:tool-monomial}(2), the ratios of the entries of $D$ are roots of unity. They lie in $K_m$, hence in $\mu_{2^m}$, and $\sigma$ multiplies each of them by a sign. Therefore
\[
 d_v=\beta_vy_v,\qquad \beta_v\in\CC^\times,\qquad y_v=C_2^{-1}S_vC_2,
\]
with $S_v$ a diagonal sign matrix. Thus $y_v$ is an involution in $C_2^{-1}\A_ZC_2=\A_L$, and $\sigma(y_v)=y_v$. Since $h_v$ and $\sigma(h_v)$ are involutions, $h_vd_vh_v=\sigma(h_v)h_v=d_v^{-1}$, that is, $h_vy_vh_v=\beta_v^{-2}y_v$. Squaring gives $\beta_v\in\mu_4$, so $y_v$ also lies in $h_v\A_Lh_v=\A_{L'}$, and $\sigma(d_v)=d_v$ and $d_v^2=\epsilon_vI$ with $\epsilon_v=\beta_v^2\in\{\pm1\}$. Up to sign, $y_v=\beta_v^{-1}d_v$ does not depend on the GSC pair.

\item \step{Step 2:}{the $d_v$ commute, and after replacing $U,W$ by $UC,WC$ for a Clifford $C$, the matrix $A=W^{-1}\sigma(W)$ is diagonal.} We first show that the $d_v$ commute. Since $\sigma^2(h_v)=\sigma(h_vd_v)=h_vd_v^2=\epsilon_vh_v$ and $h_{v+w}\in\mu_4h_vh_w$, the signs $\epsilon_v$ form a character of $V_3$. Now
\[
 d_{v+w}=h_w^{-1}d_vh_w\,d_w.
\]
The factors $d_v'=h_w^{-1}d_vh_w$ and $d_w$ have squares $\epsilon_vI$ and $\epsilon_wI$, and their product has square $\epsilon_{v+w}I=\epsilon_v\epsilon_wI$. Hence $d_v'd_wd_v'd_w=d_v'^2d_w^2$, so $d_v'$ commutes with $d_w$, and therefore so does $d_{v+w}=d_v'd_w$. As $v$ is arbitrary, the $d_v$ commute pairwise. Moreover $h_w^{-1}d_vh_w=d_{v+w}d_w^{-1}$, so every $h_w$ normalizes the commutative algebra generated by the $d_v$.

Now let $\tilde d_v=W^{-1}d_vW$. The algebra generated by the $\tilde d_v$ is commutative and normalized by every $P_w=W^{-1}h_wW$, so by (i) it is spanned by commuting Paulis. The $\tilde d_v$ have eigenvalues in $\mu_4$, so their Pauli coefficients, and hence their entries, lie in $\Q(i)$. From $h_v=WP_vW^{-1}$ and $\sigma(h_v)=\sigma(W)P_v\sigma(W)^{-1}$,
\[
 \tilde d_v=P_vAP_vA^{-1},\qquad \sigma(\tilde d_v)=A^{-1}\tilde d_vA=\tilde d_v.
\]
Thus $A$ commutes with every $\tilde d_v$, and $P_vAP_v=\tilde d_vA$. Since $\sigma$ commutes with complex conjugation, $A$ is unitary up to scalar. The Pauli conjugates of $A$ therefore generate a commutative, Pauli-normalized $*$-algebra, and by (i) $\operatorname{supp}_PA$ spans an isotropic subspace. Choose a Clifford $C$ carrying $\operatorname{supp}_PA$ into the $Z$-space. By Proposition~\ref{lem:tool-galois}(1), $\sigma(C)\doteq C$, so replacing $U,W$ by $UC,WC$ conjugates $A$ by $C$ up to scalar, and it preserves $M$ and $m$. Hence $A$ is diagonal in these coordinates.

\item \step{Step 3:}{$A\doteq\operatorname{diag}_xi^{p\cdot x}(-1)^{q(x)}$, and $p=0$ if and only if $m=M$.} Normalize $A(0)=1$, and write $P_v\doteq X^aZ^b$ for $v=(a,b)$. Then $\tilde d_{(a,b)}=\operatorname{diag}_xA(x+a)/A(x)$, and $\tilde d_{(a,b)}^2=\epsilon_{(a,b)}I$ shows that the squared ratios $A(x+a)^2/A(x)^2$ are independent of $x$ and $b$ and define a character $a\mapsto(-1)^{p\cdot a}$. Thus, for a Boolean function $q$,
\[
 A\doteq\operatorname{diag}_x i^{p\cdot x}(-1)^{q(x)},\qquad
 F_a|x\rangle=(-1)^{q(x+a)+q(x)+(p\cdot a)(p\cdot x)}|x\rangle,
\]
with $d_{(a,b)}=i^{p\cdot a}UF_aU^\dagger$. The entries of $A$ lie in $\mu_4$, so $\sigma(A)\doteq A$ and $\sigma^2(W)=\sigma(WA)\doteq WA^2=WZ^p$. If $m=M$, then $\sigma^2=1$ and $p=0$. If $m=M-1$ and $p=0$, scaling $W$ by a nonzero entry makes it $\sigma^2$-fixed and hence defined over $K_{M-1}$, contrary to the minimality of $M$.

\item \step{Step 4:}{$\deg q\leq2$.} Suppose $\deg q=3$, and let $e_a(x)=q(x+a)+q(x)+(p\cdot a)(p\cdot x)$ be the exponent of $F_a$. We first show that the seven $e_a$, $a\ne0$, span the Boolean functions of degree at most two. The quadratic part of $e_a$ is the derivative of the cubic monomial of $q$ in direction $a$, which is linear and injective in $a$. Hence a relation $\sum_{a\in S}e_a=0$ with $\emptyset\ne S\subseteq\F_2^3\setminus\{0\}$ requires $\sum_{a\in S}a=0$, so $S$ is a line $\{a,b,a+b\}$, its complement, or all seven labels. These give, respectively, $e_a+e_b+e_{a+b}=\Delta_a\Delta_bq$, where $\Delta_aq(x)=q(x+a)+q(x)$, which has nonzero linear part, its sum with the constant $1$, and the constant $\sum_{y}q(y)=1$, none of which vanishes. Hence the seven exponents are linearly independent. Each has even weight, and the Boolean functions of even weight on $\F_2^3$ are exactly those of degree at most two, a seven-dimensional space.

Next we show that $I-2(E_x+E_y)\in\C_2(3)$ for $x\ne y$, where $E_x=U|x\rangle\langle x|U^\dagger$. Each $UF_aU^\dagger=\pm y_{(a,b)}$ is a sign involution in a Pauli MASA (Step~1), and its $-1$-eigenspace has even dimension $\operatorname{wt}(e_a)$. On three qubits a sign involution in a Pauli MASA is Clifford exactly when its sign function has degree at most two, that is, even weight, so $UF_aU^\dagger\in\C_2(3)$. By the previous paragraph, products of the $UF_aU^\dagger$ give $U\bigl(I-2|x\rangle\langle x|-2|y\rangle\langle y|\bigr)U^\dagger=I-2(E_x+E_y)$.

Finally we derive a contradiction. For a Clifford $C=\sum_uc_uP_u$, $\operatorname{supp}_PC$ is a coset of $\operatorname{im}(\gamma_C-I)$, with equal absolute coefficients. Indeed, comparing coefficients in $CP_w=\pm P_{\gamma_Cw}C$ gives $|c_{u+(\gamma_C+I)w}|=|c_u|$, and for $w\in\ker(\gamma_C-I)=\operatorname{im}(\gamma_C-I)^\perp$ the relation $CP_w=\pm P_wC$ fixes $\omega(u,w)$ for all $u\in\operatorname{supp}_PC$. Here the trace is $4$, so the identity coefficient is $1/2$, the support is a subspace, and there are four coefficients, each $\pm1/2$. Let $\rho_x(P)=\tr(PE_x)$. For $P\ne I$, the coefficient of $P$ is $-(\rho_x(P)+\rho_y(P))/4\in\{0,\pm1/2\}$, so
\[
 \rho_x(P)+\rho_y(P)\in2\mathbb Z\qquad(x\ne y,\ P\ne I).
\]
Using three distinct indices, $\rho_x(P)\in\mathbb Z$ with parity independent of $x$. Now $|\rho_x(P)|\leq1$ and $\sum_P \rho_x(P)^2=8$, so the same seven nonidentity Paulis have $|\rho_x(P)|=1$ for every $x$. They stabilize every $E_x$ up to signs and form a Lagrangian $L$, so $U\A_ZU^\dagger=\A_L$, contradicting $U\notin\GSC(3)$.

\item \step{Step 5:}{the assertions about $T$ and $g_c$.} Since $\deg q\le2$, the exponent of $F_a$ is $a^{\mathsf T}B_qx+(p\cdot a)(p\cdot x)$ plus a constant, that is, $(Ta)\cdot x$ plus a constant. So $F_a=\pm Z^{Ta}$ and $d_{(a,b)}=\pm i^{p\cdot a}g_{Ta}$. If $T=0$, every $d_v$ is scalar, so $\sigma(h_v)\doteq h_v$, and dividing $h_v$ by a nonzero entry gives a matrix fixed by $\sigma$, hence with entries in $K_{m-1}$, for every $v$, contrary to the definition of $m$. As $T$ is symmetric, $\operatorname{diag}T=p$, and $\operatorname{im}T$ is the orthogonal complement of $\ker T$ for the dot product. Since $u\in\ker T$ gives $p\cdot u=u^{\mathsf T}Tu=0$, we get $p\in\operatorname{im}T$. For $0\ne c=Ta$ and every $b$, $g_c=\pm UF_aU^\dagger=\pm y_{(a,b)}$ lies in $\A_L\cap\A_{L'}$ for every $(L,L')\in\mathscr L(h_{(a,b)})$ by Step~1. It is a traceless Hermitian involution, so in such a Pauli MASA its sign function has weight four, hence degree at most two, and $g_c$ is Clifford. Then $\operatorname{supp}_Pg_c$ is a single label or four labels with coefficients $\pm1/2$ spanning a Lagrangian $L_c$. In the latter case $\A_{L_c}$ is the only Pauli MASA containing $g_c$, so every $(L,L')\in\mathscr L(h_{(a,b)})$ has $L=L'=L_c$, and $\mathscr L(h_{(a,b)})\ne\varnothing$ since $h_{(a,b)}$ is generalized semi-Clifford.
\end{itemize}
\end{proof}

\begin{lemma}\label{lem:m-M-1}
For the $p,q$ of Lemma~\ref{lem:galois-normal-form}, we have $m=M-1$, $p\ne0$, and $UZ^pU^\dagger\in\mathcal P(3)$.
\end{lemma}

\begin{proof}

Let $T$ and $g_c$ be as in Lemma~\ref{lem:galois-normal-form}, and let $\mathcal G_T=\{g_c:c\in\operatorname{im}T\}$. Below we replace $U,W$ by $UC,WC$ for Clifford gates $C$ that keep $A=W^{-1}\sigma(W)$ diagonal. The proof of Lemma~\ref{lem:galois-normal-form} then gives the normal form in the new coordinates, and each $d_v$ becomes $d_{\gamma_Cv}$, so $\mathcal G_T$ is unchanged. We read $p$, $q$ and $T$ in the current coordinates. Replacing $U$ by $KU$ with $K\in\C_2(3)$, and $W$ by $K'W$ with $K'$ the scalar multiple of $K$ over $\Q(i)$, keeps $A$, $p$, $q$ and $T$, and conjugates every $g_c$ by $K$.

It suffices to show that $U\in\GSC(3)$, or $p\ne0$ and $g_p\in\mathcal P(3)$: since $U\notin\GSC(3)$, the second alternative then holds, and as $p\ne0$, Lemma~\ref{lem:galois-normal-form} gives $m=M-1$. We split on whether $\mathcal G_T$ lies in a Pauli MASA, that is, $\mathcal G_T\subseteq\A_L$ for some Lagrangian $L$. Since $g_cg_{c'}=g_{c+c'}$, the map $c\mapsto g_c$ is an isomorphism from $\operatorname{im}T\cong\F_2^{\operatorname{rank}T}$ onto $\mathcal G_T$, and the $g_c$ commute. By Lemma~\ref{lem:galois-normal-form}, every $g_c\ne I$ is a traceless Hermitian Clifford involution that is a Pauli or has four-term support spanning a Lagrangian. If every $g_c$ is Pauli, $\mathcal G_T$ lies in a Pauli MASA, since the labels of commuting Paulis span an isotropic subspace, which lies in a Lagrangian.

The group $\mathcal G_T\cong\F_2^{\operatorname{rank}T}$ has at most three generators. Modulo sign, each $g_c\ne I$ is one of the $63$ Paulis, or a sign function of degree at most two on the joint eigenspaces of some $\A_L$ (Step~5 of the proof of Lemma~\ref{lem:galois-normal-form}), that is, one of $135\cdot28$ gates. As the Clifford group is finite modulo scalars, these groups can be listed up to Clifford conjugation.

\begin{finitefact}{cert:clifford-involution-groups-in-no-pauli-masa}{Clifford involution groups in no Pauli MASA}{Lemma~\ref{lem:m-M-1}}
Let $\mathcal G\cong\F_2^r$ be a group, modulo signs, of commuting traceless Hermitian three-qubit Clifford involutions and $I$, in which every $g\ne I$ is a Pauli or has four-term support spanning a Lagrangian, and which lies in no Pauli MASA\@. Then, after a left Clifford conjugation and a change of generators, $\mathcal G$ is generated by
\[
\hat X_1,\hat X_2,\qquad\text{or}\qquad \hat X_1,\hat X_2,Z_3,\qquad\text{or}\qquad \hat X_1,\hat X_2,\hat X_3,
\]
where $\hat X_j=QX_jQ$ and $Q=\CCZ$.
\tcblower
\textbf{Code.} \certlink{clifford-involution-groups-in-no-pauli-masa}

\end{finitefact}
So in Case~2 below, $\mathcal G_T$ is one of these three groups after a left Clifford conjugation and a change of generators.

Let $Q_P=QH_1H_2$, $Q_N=QH_1H_2H_3$, and $\Theta=\left(\begin{smallmatrix}0&1&0\\1&0&0\\0&0&0\end{smallmatrix}\right)$. In Case~2 we show that $U=Q_NG$ or $U=Q_PG$ with $G$ diagonal, and we write $G(x)^2=e^{2\pi if(x)}$ with $f:\F_2^3\to\mathbb R/\mathbb Z$ and $\Delta_jf(x)=f(x+e_j)-f(x)$.

\begin{inlinefigure}
\begin{tikzpicture}
\node[pfnode] (s) at (0,0) {$\mathcal G_T$};
\node[pfnode] (c1) at (3.4,1.5) {Case 1};
\node[pfnode] (c2) at (3.4,-1.5) {Case 2};
\node[pfnode] (c1a) at (7.6,2.4) {Case 1a};
\node[pfnode] (c1b) at (7.6,0.9) {Case 1b};
\node[pfnode] (c2b) at (7.6,-0.9) {Case 2b};
\node[pfnode] (c2a) at (7.6,-2.4) {Case 2a};
\node[pfgoal] (gsc) at (11.8,1.9) {$U\in\GSC(3)$};
\node[pfgoal] (goal) at (11.8,0.1) {$p\ne0$, $g_p\in\mathcal P(3)$};
\node[pfgoal] (bot) at (11.8,-1.7) {Contradiction};
\draw[pfedge] (s) -- node[pflab,sloped,above] {in a Pauli MASA} (c1);
\draw[pfedge] (s) -- node[pflab,sloped,below] {in no Pauli MASA} (c2);
\draw[pfedge] (c1) -- node[pflab,sloped,above] {some $g_c\notin\mathcal P(3)$} (c1a);
\draw[pfedge] (c1) -- node[pflab,sloped,below] {$\mathcal G_T\subseteq\mathcal P(3)$} (c1b);
\draw[pfedge] (c2) -- node[pflab,sloped,above] {$U=Q_PG$} (c2b);
\draw[pfedge] (c2) -- node[pflab,sloped,below] {$U=Q_NG$} (c2a);
\draw[pfedge] (c1a) -- (gsc);
\draw[pfedge] (c1b) -- node[pflab,sloped,above,pos=.45] {$\operatorname{rank}T\ge2$} (gsc);
\draw[pfedge] (c1b) -- node[pflab,sloped,below,pos=.45] {$\operatorname{rank}T=1$} (goal);
\draw[pfedge] (c2b) -- node[pflab,sloped,above,pos=.45] {$\operatorname{rank}T=3$} (goal);
\draw[pfedge] (c2b) -- node[pflab,sloped,below,pos=.45] {$\operatorname{rank}T=2$} (bot);
\draw[pfedge] (c2a) -- (bot);
\end{tikzpicture}
\caption{Case analysis in the proof of Lemma~\ref{lem:m-M-1}.}
\end{inlinefigure}

\begin{itemize}
\item \step{Case 1a:}{$\mathcal G_T\subseteq\A_N$ for a Lagrangian $N$, and some $g_c$ is not Pauli.} For each non-Pauli $g_c$, the Lagrangian $L_c$ spanned by $\operatorname{supp}_Pg_c$ lies in $N$, so $L_c=N$. The $c$ with $g_c$ Pauli form a proper subgroup of $\operatorname{im}T$, so the $a$ with $g_{Ta}$ non-Pauli lie outside a proper subgroup of $\F_2^3$ and span it, and the labels $(a,b)$ with $g_{Ta}$ non-Pauli span $V_3$. By Lemma~\ref{lem:galois-normal-form}, each such $h_{(a,b)}$ normalizes $\A_N$. Hence so does every $h_v$, and $U\in\GSC(3)$ by Lemma~\ref{lem:main-local-global}.

\item \step{Case 1b:}{every $g_c$ is Pauli.} At rank three, $U\A_ZU^\dagger$ is spanned by the $g_c$, so it is a Pauli MASA\@. At rank two, $U$ maps two independent commuting Paulis $Z^c$, $c\in\operatorname{im}T$, to Paulis, and Lemma~\ref{lem:tool-two-paulis} applies. At rank one, $T=pp^{\mathsf T}$, since a rank-one symmetric binary matrix is $uu^{\mathsf T}$ with diagonal $u$. Thus $p\ne0$ and $g_p$ is Pauli.

\item \step{Case 2:}{$\mathcal G_T$ lies in no Pauli MASA.} We first show that $U=Q_NG$ if $\mathcal G_T=\langle\hat X_1,\hat X_2,\hat X_3\rangle$ and $U=Q_PG$ otherwise, with $G$ diagonal. We then derive the constraints on the phases of $G$ recorded in the table, and Cases~2a and~2b compare them with the Galois cocycle.

We have $\hat X_j=Q_NZ_jQ_N^\dagger$ for $j\le3$, $\hat X_j=Q_PZ_jQ_P^\dagger$ for $j\le2$, and $Z_3=Q_PZ_3Q_P^\dagger$. In coordinates where the listed generators are $g_{e_1},g_{e_2}$ (and $g_{e_3}$), the gate $G=Q_N^{-1}U$, respectively $G=Q_P^{-1}U$, commutes with $Z_j$ for $e_j\in\operatorname{im}T$, so at rank three $G$ is diagonal. For $0\ne c\in\Span\{e_1,e_2\}$, the Pauli MASA of $g_c$ is $Q_P\A_{\Span\{z^c,x^{\Theta c},z_3\}}Q_P^\dagger$.

At rank two, $\mathcal G_T=\langle\hat X_1,\hat X_2\rangle$, $\ker T=\Span\{e_3\}$, and $G=\bigoplus_{u\in\F_2^2}G_u$ with one-qubit blocks on qubit~$3$. Let $v=(a,b)$ with $a'=(a_1,a_2)\ne0$, and write $P_v\doteq X^{a'}Z^{b'}\otimes P$ with $P$ a Pauli on qubit~$3$. By Lemma~\ref{lem:galois-normal-form}, $h_v$ normalizes the Pauli MASA of $g_{Ta}$, so $GP_vG^\dagger=Q_P^\dagger h_vQ_P$ normalizes $\A_{\Span\{z^c,x^{\Theta c},z_3\}}$, $c=Ta$, which contains $Z_3$. The elements of this algebra commuting with $Z_1,Z_2$ restrict to each block as elements of $\Span_{\CC}\{I,Z_3\}$, so the blocks $\pm G_{u+a'}PG_u^\dagger$ of $GP_vG^\dagger$ send $Z_3$ to $\pm Z_3$. With $P=I$ this makes $G_u^\dagger Z_3G_u$ independent of $u$ up to sign, and with $P=X,Y,Z$ this common involution is $\pm$ a Pauli. After a right Clifford on qubit~$3$, $G=\Pi_tD$ with $\Pi_t(u,z)=(u,z+t(u))$ and $D$ diagonal, and a right Clifford removes the affine part of $t$. If $t(u)=u_1u_2$, then
\[
UZ_3U^\dagger=\diag_{z_3}\bigl(C_0,\,-\CZ_{12}C_0\CZ_{12}\bigr),\qquad C_0=H_1H_2\CZ_{12}H_1H_2,
\]
and the symplectic action of $\CZ_{12}C_0\CZ_{12}C_0^{-1}$ has order three, so none of its $2^j$-th powers lies in $\mu_4\mathcal P(2)$, contrary to Proposition~\ref{lem:controlled-clifford}. Hence $G$ is diagonal. These right Clifford gates are fixed by $\sigma$ up to scalar and keep $A\doteq G^{-1}\sigma(G)$ diagonal.

Now we constrain the phases. Let $Q=Q_N$ or $Q=Q_P$, so that $U=QG$, and let $v=(a,b)$ with $g_{Ta}$ not a Pauli. By Lemma~\ref{lem:galois-normal-form}, $h_v$ normalizes the Pauli MASA $Q\A_SQ^\dagger$ of $g_{Ta}$, where
\[
S=\Span\{z^{Ta},x^k:k\cdot Ta=0\}\ \text{ for }Q=Q_N,\qquad
S=\Span\{z^c,x^{\Theta c},z_3\}\ \text{ for }Q=Q_P,
\]
with $c$ the $\Span\{e_1,e_2\}$-part of $Ta$ (at rank three, $g_{Ta}=\pm g_cZ_3$ lies in the Pauli MASA of $g_c$). Hence $GP_vG^\dagger$, which is $P_v$ times $R_a=\diag_xG(x+a)/G(x)$, normalizes $\A_S$, and so does $R_a$. For each $X^k$ in $\A_S$, $R_aX^kR_a^\dagger$ lies in $\A_S$ and commutes with $X^k$. As $R_a^2=\diag_xe^{2\pi i\Delta_af(x)}$, this says $\Delta_k\Delta_af=0$. The unknowns are the eight values $f(x)\in\mathbb R/\mathbb Z$, and the equations form an integer system $Rf=0$. By Pontryagin duality for $\mathbb R/\mathbb Z$, a relation $wf=0$ with $w\in\Z^8$ holds for every solution exactly when $w$ lies in the integer row span of $R$. Each consequence is therefore decided by integer row reduction, for each of the $64$ symmetric binary $3\times3$ matrices $T$, with the following outcome.

\begin{finitefact}{cert:second-difference-equations-for-phases}{Second-difference equations for phases}{Lemma~\ref{lem:m-M-1}}

Let $f:\F_2^3\to\mathbb R/\mathbb Z$, let $T$ be a symmetric binary $3\times3$ matrix, and let $\Theta$ be as in the proof of Lemma~\ref{lem:m-M-1}. Then:
\[
\begin{array}{c|c|c}
T&\text{equations}&\text{consequence}\\\hline
\operatorname{rank}T=3&\Delta_k\Delta_af=0\ (k\cdot Ta=0)&4\Delta_jf=0\ (j=1,2,3)\\
\operatorname{rank}T=3\text{ or }\operatorname{im}T=\Span\{e_1,e_2\},\ T\ne\Theta&\Delta_{\Theta Ta}\Delta_af=0\ (\text{all }a)&4\Delta_jf=0\ (j=1,2)\\
T=\Theta&\Delta_{\Theta Ta}\Delta_af=0\ (\text{all }a)&2\Delta_jf=0\ (j=1,2)
\end{array}
\]

\tcblower

\textbf{Code.} \certlink{second-difference-equations-for-phases}

\end{finitefact}
By \textup{(\ref{cert:second-difference-equations-for-phases})}, $4\Delta_jf=0$ for the listed $j$, and $2\Delta_jf=0$ if $T=\Theta$. Thus the ratios of $G$ along the listed $e_j$ lie in $\mu_8$, and in $\mu_4$ if $T=\Theta$. As $\sigma(Q_P)\doteq Q_P$ and $\sigma(Q_N)\doteq Q_N$, we have $A\doteq G^{-1}\sigma(G)$, and $\sigma^2(W)\doteq WZ^p$ reads $G^{-1}\sigma^2(G)\doteq Z^p$.

\item \step{Case 2a:}{$U=Q_NG$.} Then $G^8$ is scalar, so $M\le3$. But $T$ is nonsingular of odd size, so $p\ne0$ and $m=M-1>2$, a contradiction.

\item \step{Case 2b:}{$U=Q_PG$.} If $p=0$, then $T=B_q=\Theta$, the ratios along $\Span\{e_1,e_2\}$ are fixed by $\sigma$, so $(-1)^q$ is invariant under $\Span\{e_1,e_2\}$ up to sign and $B_qe_1=B_qe_2=0$, a contradiction. If $p\ne0$, then $M\ge4$, the ratios along $e_1,e_2$ are fixed by $\sigma^2$, and $G^{-1}\sigma^2(G)\doteq Z^p$ gives $p\in\Span\{e_3\}$. At rank two this contradicts $p\in\operatorname{im}T=\Span\{e_1,e_2\}$. At rank three, $p=e_3$, and $U=Q_PG$ commutes with $Z_3$, so $g_p=Z_3$.
\end{itemize}

\end{proof}

We use the notation
\[
R_A(t)=\frac{I+A}{2}+t\frac{I-A}{2},
\]
and let $\xi=\zeta_{2^M}^{-1}$.

\begin{lemma}\label{lem:conductor-step}
Let $U'=UR_{Z^p}(\xi)$. Then $U'\in\C(3)$, $\operatorname{cond}(U')<M$, and
\[
U'\in\GSC(3)\quad\Longleftrightarrow\quad U\in\GSC(3).
\]
\end{lemma}

\begin{proof}

By Lemma~\ref{lem:m-M-1}, $m=M-1$, $p\ne0$, and $P\coloneqq UZ^pU^\dagger$ is a Pauli. Hence
\[
U'=UR_{Z^p}(\xi)=R_P(\xi)U.
\]

\begin{itemize}
\item \step{Step 1:}{$\operatorname{cond}(U')<M$.} Let $W\doteq U$ be the matrix over $K_M$ chosen above. Since $m=M-1$, the automorphism $\delta=\sigma^2$ generates $\operatorname{Gal}(K_M/K_{M-1})$, sends $\xi$ to $-\xi$, and satisfies $\delta(W)\doteq WZ^p$ by Step~3 of the proof of Lemma~\ref{lem:galois-normal-form}. Therefore
\[
\delta(WR_{Z^p}(\xi))\doteq WZ^p R_{Z^p}(-\xi)=WR_{Z^p}(\xi),
\]
because $Z^pR_{Z^p}(-\xi)=R_{Z^p}(\xi)$. Dividing $WR_{Z^p}(\xi)$ by one of its nonzero entries gives a matrix fixed entrywise by $\delta$. Its entries lie in $K_{M-1}$, so $\operatorname{cond}(U')\leq M-1$.

\item \step{Step 2:}{$U'\in\C(3)$.} Fix $v\in V_3$. Since $P_vZ^p=(-1)^{\omega(v,(0,p))}Z^pP_v$, expanding $R_{Z^p}(\xi)P_vR_{Z^p}(\xi)^\dagger$ gives
\[
U'P_vU'^\dagger=
\begin{cases}
h_v,&\omega(v,(0,p))=0,\\
\xi^{-1}R_P(\xi^2)h_v,&\omega(v,(0,p))=1.
\end{cases}
\]
In the first case, $h_v\in\C(3)$ since $U\in\C(3)$. Suppose now $\omega(v,(0,p))=1$. Since $h_vPh_v^\dagger=-P$ and $h_v$ is generalized semi-Clifford, Lemma~\ref{lem:root-witness} gives $(L_v,L_v')\in\mathscr L(h_v)$ with the label of $P$ in $L_v$ and in $L_v'$. Choose Clifford gates $C_1,C_2$ with $C_1PC_1^\dagger=Z_1$ and $C_1h_vC_2=\Pi_vD_v$ monomial. Set
\[
D'_v=\Pi_v^\dagger R_{Z_1}(\xi^2)\Pi_vD_v,
\qquad C_1R_P(\xi^2)h_vC_2=\Pi_vD'_v.
\]
By~\cite[Corollary~A.11.1]{Anderson_2024}, $\Pi_v,D_v\in\C(3)$. The gate $D'_v$ is diagonal with dyadic relative phases, so it also lies in $\C(3)$ by the diagonal classification~\cite[Appendix~F]{Anderson_2024}. The same corollary then gives $R_P(\xi^2)h_v\in\C(3)$. There are finitely many $v$, so their corrections have a common finite level. Hence $U'\in\C(3)$.

\item \step{Step 3:}{$U'\in\GSC(3)$ if and only if $U\in\GSC(3)$.} Suppose that one of the gates $U,U'$, call it $V$, has a GSC pair $(L,L')\in\mathscr L(V)$. Both gates send $Z^p$ to $P$, since $U'Z^pU'^\dagger=R_P(\xi)PR_P(\xi)^\dagger=P$. Hence Lemma~\ref{lem:root-witness} lets us assume that the labels of $Z^p$ and $P$ lie in $L$ and $L'$, respectively. Then $R_{Z^p}(\xi)\in\A_L$, so
\[
U'\A_LU'^\dagger=UR_{Z^p}(\xi)\A_LR_{Z^p}(\xi)^\dagger U^\dagger=U\A_LU^\dagger.
\]
Thus $(L,L')$ is a GSC pair of both gates, proving the equivalence.
\end{itemize}

\end{proof}

This proves Theorem~\ref{thm:conductor-descent}. Indeed, the gate $U$ of Lemma~\ref{lem:conductor-step} is $KU_0C$ for the given gate $U_0$ and Clifford gates $K,C$ (Lemmas~\ref{lem:galois-normal-form} and~\ref{lem:m-M-1}). Multiplication by Clifford gates preserves the conductor (since Clifford gates are defined over $K_2$ up to scalar) and whether a gate is in $\GSC(3)$ or not (Lemma~\ref{lem:basic-operations}(1)). So $K^{-1}U'C^{-1}$ has the required properties for $U_0$. The case $U\in\GSC(3)$ was treated at the start of this subsection.

\section{Results on \texorpdfstring{$\mathcal{C}(4)$}{C(4)}}\label{app:four-qubit-results}

\subsection{Proof of Theorem~\ref{thm:n4-gaussian} when \texorpdfstring{$n=4$}{n=4}}\label{app:four-qubit}
\phantomsection\label{gauss-app:diagram}

In this subsection, let $U\in\C(4)$, and suppose that every $h_v=UP_vU^\dagger$, $v\in V_4$, is Gaussian. By Lemma~\ref{app:lem-ge-equiv}, we may also assume that the $h_v$ are generalized semi-Clifford, so that $h_v\in\mathscr R_4$ for $v\ne0$ by Proposition~\ref{app:lem-R_n}. These assumptions hold for all the lemmas in this subsection except Lemma~\ref{app:return-two}, which states its own hypotheses because the proof of Lemma~\ref{app:return-balanced} applies it to a gate other than $U$. We want to find a pair $(L, M)\in \mathscr L(U)$. A \textit{Pauli pair} of a gate $W$ is a pair $(s,t)$ of nonzero labels with $WP_tW^\dagger=\pm P_s$. For a gate $K$, let $X(K)=\{L:(L,M)\in\mathscr L(K)\text{ for some }M\}$. Figure~\ref{fig:F} shows how the lemmas below combine, and the proofs of Lemmas~\ref{app:return-single}--\ref{app:return-two} begin with a tree of their cases.

\begin{figure}
    \centering
    \begin{tikzpicture}
\node[pfnode] (endpoint) at (0,0)
  {\hyperref[app:finite-endpoint]{Lemma~\ref*{app:finite-endpoint}}};
\node[pfnode] (block) at (0,-1.5)
  {\hyperref[app:return-block]{Lemma~\ref*{app:return-block}}:
   $W=\operatorname{diag}(A,B)$};
\node[pfnode] (single) at (-3.4,-3.1)
  {\hyperref[app:return-single]{Lemma~\ref*{app:return-single}}};
\node[pfnode] (multiple) at (3.4,-3.1)
  {\hyperref[app:return-multiple]{Lemma~\ref*{app:return-multiple}}};
\node[pfnode] (balanced) at (-3.4,-4.8)
  {\hyperref[app:return-balanced]{Lemma~\ref*{app:return-balanced}}};
\node[pfnode] (two) at (1.6,-4.8)
  {\hyperref[app:return-two]{Lemma~\ref*{app:return-two}}};
\node[pfgoal] (goal) at (0,-6.3) {$\mathscr L(U)\ne\varnothing$};
\draw[pfedge] (endpoint) -- node[pflab,right] {Pauli pair} (block);
\draw[pfedge] (block) -- node[pflab,above left] {$\exists q:\ |X(K_q)|=1$} (single);
\draw[pfedge] (block) -- node[pflab,above right] {$\forall v:\ |X(K_v)|\ge2$} (multiple);
\draw[pfedge] (single) -- node[pflab,left]
  {$W^\dagger(I\otimes P_r)W$\\$=(P_a\pm P_b)/\sqrt2$} (balanced);
\draw[pfedge] (single) -- node[pflab,sloped,above=2pt] {$[P_t,Z_1]=0$, $P_t\ne Z_1$} (two);
\draw[pfedge] (multiple) -- node[pflab,right=2pt,pos=.55] {$[P_t,Z_1]=0$, $P_t\ne Z_1$} (two);
\draw[pfedge,dashed] (balanced) -- node[pflab,below=2pt] {uses} (two);
\draw[pfedge] (balanced) -- (goal);
\draw[pfedge] (two) -- (goal);
\draw[pfedge] (endpoint.east) -- (7,0) |- (goal.east);
\draw[pfline] (multiple.east) -- (7,-3.1);
\draw[pfedge] (single.west) -- (-7,-3.1) |- (goal.west);
\node[pflab,above left] at (7,0) {$\mathscr L(U)\ne\varnothing$};
\node[pflab,above left] at (7,-3.1) {$\mathscr L(W)\ne\varnothing$};
\node[pflab,above right] at (-7,-3.1) {$\mathscr L(W)\ne\varnothing$};
\end{tikzpicture}
    \caption{Proof of Theorem~\ref{thm:n4-gaussian} when $n=4$.}
    \label{fig:F}
\end{figure}
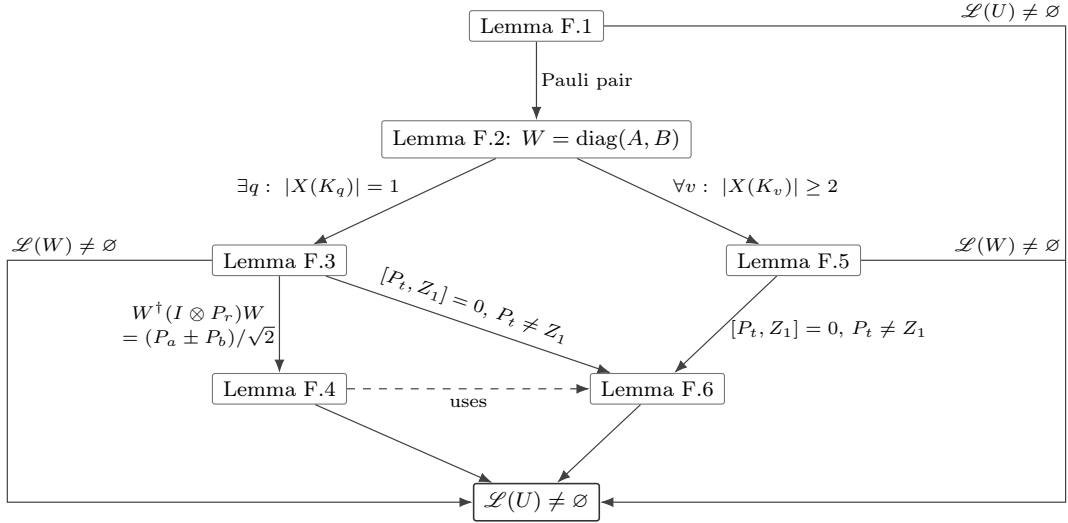

\begin{lemma}\label{app:finite-endpoint}
For $U \in \mathcal{C}(4)$, either $\mathscr L(U)\ne\varnothing$ or there exist $s,t\in V_4\setminus\{0\}$ such that $UP_tU^\dagger=\pm P_s$.
\end{lemma}
\begin{proof}
For $h\in\mathscr R_4$, let
\[
N(h)=\{s\in V_4:hP_sh^\dagger=\pm P_s\},\qquad
F(h)=(\operatorname{span}\operatorname{supp}_P h)^{\perp}, \qquad \rho(h)=(|N(h)|-1,|F(h)|-1,t(h)),
\]
where
\[
t(h)=
\begin{cases}
1, & \text{if $h \in \mathcal{C}_1(4)$},\\
2, & \text{if $h \in \mathcal{C}_2(4) \setminus \mathcal{C}_1(4)$},\\
3, & \text{if $h \in \mathcal{C}_3(4) \setminus \mathcal{C}_2(4)$},\\
4, & \text{if $h \in \mathcal{C}_k(4) \setminus \mathcal{C}_3(4)$ for $k \geq 4$}.
\end{cases}
\]
We order $\rho(h)$ lexicographically. Each entry of $\rho$ is invariant under Clifford conjugation. For $M\in\mathscr R_4$, let $\mathcal K(M)$ be the set of $h\in\mathscr R_4$ for which some $\epsilon\in\F_2$ satisfies
\[
hM=(-1)^\epsilon Mh,\qquad g=i^\epsilon Mh\in\mathscr R_4,\qquad
\rho(h)\ge\rho(M),\qquad \rho(g)\ge\rho(M).
\]

By Proposition~\ref{app:lem-R_n}, every $h\in\mathscr R_4$ is Clifford conjugate to some $M$ of the form~\eqref{core:eq:row}, determined by $(\lambda,f,\pi)\in\F_2\times(\Z/4)^{\F_2^4}\times Q_4$. After an affine permutation, which is Clifford, $Q_4$ is the group of order $2^{15}$ preserving a fixed complete affine flag (Lemma~\ref{lem:tool-flags}(1)). The set $N(M)$ is found by conjugating the $256$ Paulis, and an enumeration gives the following.

\begin{finitefact}{cert:r4-representatives-up-to-clifford-conjugacy}{Representatives of $\mathscr R_4$ up to Clifford conjugacy}{Lemma~\ref{app:finite-endpoint}}
A list of $1{,}020$ gates of the form~\eqref{core:eq:row} on four qubits contains a Clifford conjugate of every $h\in\mathscr R_4$. Exactly ten listed $M$ have $|N(M)|-1\ge63$, and they lie in $\C_2(4)$.

\tcblower

\textbf{Code.} \certlink{r4-representatives-up-to-clifford-conjugacy}

\end{finitefact}
The set $\mathcal K(M)$ lies in $\mathscr R_4$, which is finite: its elements are the $CM'C^\dagger$ with $M'$ of the form~\eqref{core:eq:row} and $C\in\C_2(4)$, which is finite modulo scalars. The sets $N(h)$ and the entries of $\rho(h)$ are found by conjugating Paulis, and $V_4$ has finitely many planes. A search over the listed $M$ gives the following.

\begin{finitefact}{cert:plane-meeting-fixed-pauli-labels}{A plane meeting the fixed Pauli labels}{Lemma~\ref{app:finite-endpoint}}

For each of the other $1{,}010$ listed $M$ there is a plane $\Pi_M\le V_4$ with
\[
N(h)\cap\Pi_M\ne\{0\}\qquad\text{for all }h\in\{M\}\cup\mathcal K(M),
\]
where $N(h)=\{s:hP_sh^\dagger=\pm P_s\}$.

\tcblower

\textbf{Code.} \certlink{plane-meeting-fixed-pauli-labels}

\end{finitefact}

By \textup{(\ref{cert:plane-meeting-fixed-pauli-labels})}, every listed $M$ with $|N(M)|-1<63$ has a two-dimensional subspace $\Pi_M\le V_4$ with
\begin{equation}\label{eq:finite-plane-lock}
N(h)\cap\Pi_M\ne\{0\}
\qquad(h\in\{M\}\cup\mathcal K(M)).
\end{equation}
Fix $v\ne0$ with $\rho(h_v)$ minimal among the $\rho(h_w)$, $w\ne0$. By \textup{(\ref{cert:r4-representatives-up-to-clifford-conjugacy})}, there exists $C\in\C_2(4)$ with $Ch_vC^\dagger=M$. Replacing $U$ by $CU$ replaces each $h_w$ by $Ch_wC^\dagger$, which preserves $\rho$, the hypotheses, and both alternatives of the lemma. We may therefore assume $h_v=M$.
\begin{itemize}
\item \step{Case 1:}{$|N(M)|-1\ge63$.} By minimality, $|N(h_w)|-1\ge63$ for every $w\ne0$. Since each $h_w$ is Clifford conjugate to a listed representative with the same $N$-count, \textup{(\ref{cert:r4-representatives-up-to-clifford-conjugacy})} gives $h_w\in\C_2(4)$. Hence, $U\in\C_3(4)\subseteq\SC(4)$ by Theorem~\ref{thm:sc-six}, and $\mathscr L(U)\ne\varnothing$.

\item \step{Case 2:}{$|N(M)|-1<63$.} We first show that $N(h_w)\cap \Pi_M\neq\{0\}$ for every $w\neq 0$. Let $w\ne0$ and $w \neq v$, $h=h_w$, and $\epsilon=\omega(v,w)$, so that $hM=(-1)^\epsilon Mh$. Then $g=i^\epsilon Mh=\pm h_{v+w}\in\mathscr R_4$, and minimality gives $\rho(h)\ge\rho(M)$ and $\rho(g)\ge\rho(M)$, so $h\in\mathcal K(M)$. With $h_v=M$,~\eqref{eq:finite-plane-lock} gives $N(h_w)\cap\Pi_M\ne\{0\}$ for every $w\ne0$.

Now, for $0\ne s\in\Pi_M$, let $B_s=U^\dagger P_sU=\sum_ub_uP_u$, and $V_s=\{w\in V_4:P_wB_sP_w^\dagger=\pm B_s\}$. Conjugating by $U$ shows that $s\in N(h_w)$ if and only if $w\in V_s$, so the three sets $V_s$ cover $V_4$. Conjugation by $P_w$ multiplies $b_u$ by $(-1)^{\omega(w,u)}$, so
\[
V_s=\operatorname{span}(\operatorname{supp}_PB_s+\operatorname{supp}_PB_s)^\perp
\]
is a subspace. If all three $V_s$ were proper, they would be three distinct hyperplanes, the only way three proper subspaces cover a vector space over $\F_2$. Then each $B_s$ would have exactly two Pauli terms, with labels $u_s$ and $u_s+d_s$ and distinct $d_s$. Moreover, for $s\ne s'$ the product $B_sB_{s'}$ would have the four labels $u_s+u_{s'}+\{0,d_s,d_{s'},d_s+d_{s'}\}$, while $B_sB_{s'}\doteq B_{s+s'}$ has two. Hence $V_s=V_4$ for some $s$, so $B_s$ has a single Pauli term. As $B_s$ is a Hermitian unitary and $s\ne0$, $B_s=\pm P_t$ with $t\ne0$, that is, $UP_tU^\dagger=\pm P_s$.
\end{itemize}
\end{proof}

\begin{lemma}\label{app:return-block}
Suppose $U$ has a Pauli pair. Then there are $C_1,C_2\in\C_2(4)$ such that $W=C_1UC_2^\dagger$ has the following properties.
\renewcommand{\labelenumi}{(\theenumi)}
\begin{enumerate}
\item $WZ_1W^\dagger=Z_1$, so $W=\operatorname{diag}(A,B)$ with $A,B\in\C_4(3)\cap\GSC(3)$.
\item For each $v\in V_3$, let $K_v=BP_vA^\dagger$. Then $\langle\Pau(3),K_v\rangle$ is a finite $2$-group modulo scalars, and it normalizes $\A_L$ for some Lagrangian $L\le V_3$.
\item After replacing $W$ by $W(I\otimes C)$ for a single $C\in\C_2(3)$, the gate $W$ is a scalar multiple of a matrix over $\Q(\zeta_8)$, and for the automorphism $\sigma$ of $\Q(\zeta_8)$ with $\sigma(\zeta_8)=-\zeta_8$ there are $e\in\F_2$ and $p\in\F_2^3$ with
\[
\sigma(W)\doteq W(Z_1^e\otimes Z^p),\qquad
\sigma(A)\doteq AZ^p,\qquad \sigma(B)\doteq BZ^p.
\]
In particular $(0,p)\in L$ for every $(L,M)\in\mathscr L(A)\cup\mathscr L(B)$.
\item For every $(L,M)\in\mathscr L(A)$, we have $A=\lambda C_3DC_4$ with $\lambda\in\T$, $C_3,C_4\in\C_2(3)$, $C_4\A_LC_4^\dagger=\A_Z$ and
\[
D=\pi_0\,i^{axy+bxz+cyz+dxyz}\,T_{p'},
\qquad \pi_0\in\{I,\pi\},\quad
a,b,c\in\F_2,\ d\in\Z/4,
\]
where $Z^{p'}\doteq C_4Z^pC_4^\dagger$. The same holds for $B$. Here $i^{f}$ is the diagonal gate $\ket{x,y,z}\mapsto i^{f(x,y,z)}\ket{x,y,z}$, $\pi$ is the permutation gate of $(x,y,z)\mapsto(x,y,z+xy)$, and $T_{p'}=(I+Z^{p'})/2+\zeta_8(I-Z^{p'})/2$.
\end{enumerate}
\end{lemma}
\begin{proof}
Choose $0\neq t\in\mathsf P(U)$. Then, for some $0\neq s\in V_4$ and $\eta\in\{\pm1\}$, we have $UP_tU^\dagger=\eta P_s$. Choose $C_1,C_2\in\C_2(4)$ such that $C_2P_tC_2^\dagger=Z_1$ and $C_1(\eta P_s)C_1^\dagger=Z_1$. Then $WZ_1W^\dagger=Z_1$, so $W=\operatorname{diag}(A,B)$ on the eigenspaces of $Z_1$, with blocks on qubits $2,3,4$. As $\mathscr R_4$ is closed under Clifford conjugation and sign,
\[
W(I\otimes P_v)W^\dagger=\operatorname{diag}(AP_vA^\dagger,\,BP_vB^\dagger)\in\mathscr R_4\qquad(0\ne v\in V_3).
\]
Its blocks, the corrections of $A$ and $B$, satisfy the hypotheses of Proposition~\ref{gauss-prop:n3-gsc}: they are Gaussian, they are generalized semi-Clifford by Lemma~\ref{lem:tool-blocks}(1) applied to a GSC pair of the left side, and they have finite level by Lemma~\ref{lem:tool-blocks}(2). Therefore, $A,B\in\C_4(3)\cap\GSC(3)$.

For $v\in V_3$ the correction with $X$ on the first qubit is
\[
W(X_1\otimes P_v)W^\dagger
=\begin{pmatrix}0&K_v^\dagger\\K_v&0\end{pmatrix},
\]
which anticommutes with $Z_1$. This correction lies in $\mathscr R_4$, so it generates, together with $\Pau(4)$, a finite $2$-group by Proposition~\ref{app:lem-R_n}. By Lemma~\ref{lem:tool-blocks}(3) with $j=1$ and $e=1$, so is $\langle\Pau(3),K_v\rangle$, and it fixes a Pauli MASA $\A_L$ by Theorem~\ref{thm:two-group-chamber}.

Let $\sigma$ be the automorphism of $\Q(\zeta_8)$ with $\sigma(\zeta_8)=-\zeta_8$, which fixes $i$. By Proposition~\ref{app:lem-R_n}, every $h'_v=WP_vW^\dagger$ has entries in $\Q(\zeta_8)$, so $W$, determined up to scalar by the $h'_v$, may be taken over $\Q(\zeta_8)$. Each $h'_v$ is Gaussian, so $\sigma(h'_v)\doteq h'_v$, and Proposition~\ref{lem:tool-galois}(3) gives $\sigma(W)\doteq WP_a$ for some $a\in V_4$. Applying this to $Z_1=\sigma(WZ_1W^\dagger)$ shows that $P_a$ commutes with $Z_1$, so $P_a\doteq Z_1^e\otimes P_{a'}$ with $e\in\F_2$ and $a'\in V_3$. Replacing $W$ by $W(I\otimes C)$ for one $C\in\C_2(3)$ with $C^\dagger P_{a'}C\doteq Z^p$, which keeps (1) and (2) with $A,B$ replaced by $AC,BC$, and using $\sigma(C)\doteq C$ via Proposition~\ref{lem:tool-galois}(1), we get $\sigma(W)\doteq W(Z_1^e\otimes Z^p)$, so $\sigma(A)\doteq AZ^p$ and $\sigma(B)\doteq BZ^p$. If $(L,M)\in\mathscr L(A)$, then Proposition~\ref{lem:tool-galois}(2) gives $Z^p\doteq A^{-1}\sigma(A)\in\A_L$, that is, $(0,p)\in L$. The same argument applies to $B$.

Lastly, let $(L,M)\in\mathscr L(A)$. By Proposition~\ref{lem:tool-monomial} applied with $n=3$ to $A$, which is defined over $K_3$ up to scalar, $A\doteq C_3\pi_0\Delta C_4$ with $C_4\A_LC_4^\dagger=\A_Z$, $\pi_0\in\{I,\pi\}$, and $\Delta$ diagonal with entries $\zeta_8^{\,g(x)}$ after a scalar, $g:\F_2^3\to\Z/8$. Let $Z^{p'}\doteq C_4Z^pC_4^\dagger$, a diagonal Pauli by (3). By Proposition~\ref{lem:tool-galois}(1), $\Delta^{-1}\sigma(\Delta)\doteq C_4A^{-1}\sigma(A)C_4^\dagger\doteq Z^{p'}$, and $\Delta^{-1}\sigma(\Delta)=\operatorname{diag}((-1)^{g(x)})$, so $g(x)\equiv p'\cdot x+c_0\pmod2$ for a constant $c_0$. Hence $\zeta_8^{\,g}=\zeta_8^{\,c_0}\,T_{p'}\,i^{g'}$ for some $g':\F_2^3\to\Z/4$. Write $g'$ as a multilinear polynomial with coefficients in $\Z/4$. The diagonal Cliffords $i^{\ell}(-1)^{q}$, with $\ell$ linear and $q$ quadratic, remove its linear part and reduce its quadratic coefficients modulo $2$. Being diagonal, they commute with $T_{p'}$ and $Z^{p'}$, so absorbing them into $C_4$ keeps $C_4\A_LC_4^\dagger=\A_Z$ and $Z^{p'}\doteq C_4Z^pC_4^\dagger$. This gives $D=\pi_0\Delta$ with $\Delta\doteq i^{axy+bxz+cyz+dxyz}\,T_{p'}$. The same argument applies to $B$.
\end{proof}

\begin{lemma}\label{app:return-single}
Let $W=\operatorname{diag}(A,B)$ be as in Lemma~\ref{app:return-block}, with $K_v=BP_vA^\dagger$, and suppose that $|X(K_q)|=1$ for some $q\in V_3$. Then at least one of the following holds.
\renewcommand{\labelenumi}{(\theenumi)}
\begin{enumerate}
\item $\mathscr L(W)\ne\varnothing$;
\item $W$ has a Pauli pair $(s,t)$ with $P_t\ne Z_1$ and $P_tZ_1=Z_1P_t$;
\item there exist $0\ne r\in V_3$ and $a,b\in V_4$ with $\omega(a,b)=1$ and $W^\dagger(I\otimes P_r)W=(P_a\pm P_b)/\sqrt2$.
\end{enumerate}
\end{lemma}
\begin{proof}
For $v\in V_3$, let $a_v=AP_vA^\dagger$ and $b_v=BP_vB^\dagger$, and let $\Sigma=\mathsf C(A)+\mathsf C(B)$ and $\delta=\dim\Sigma$. Fix $(L_A,M_A)\in\mathscr L(A)$ and $(L_B,M_B)\in\mathscr L(B)$. By Lemma~\ref{lem:tool-unique}(3), $L_A\subseteq\mathsf C(A)$ and $L_B\subseteq\mathsf C(B)$, so $3\le\delta\le6$.

We first show that every $K_v$ with $v\in q+\Sigma$ has the same unique GSC pair $(N,N)$. By Lemma~\ref{app:return-block}(2), $K_q$ normalizes some $\A_N$, so $\mathscr L(K_q)=\{(N,N)\}$. For $s\in\mathsf C(A)$ and $t\in\mathsf C(B)$, $K_{q+s}\doteq K_qa_s$ and $K_{q+t}\doteq b_tK_q$ with $a_s,b_t\in\C_2(3)$. By Lemma~\ref{lem:tool-transport}(1), each $K_v$ with $v\in q+\Sigma$ has exactly one GSC pair, which is of the form $(N_v,N_v)$ by Lemma~\ref{app:return-block}(2). Comparing with $(\gamma_{a_s}^{-1}N,N)$ and $(N,\gamma_{b_t}N)$ gives $N_v=N$, $\gamma_{a_s}N=N$ and $\gamma_{b_t}N=N$.

We split into cases on $\delta$, and for $\delta\in\{4,5\}$ on the number of Pauli terms of $H=W^\dagger(I\otimes P_r)W$, with $r$ constructed in Case~3. Cases~1, 2 and~3c use \tool{U}: the unique GSC pair of $K_q$ becomes a GSC pair of $W$.

\begin{inlinefigure}
\begin{tikzpicture}
\node[pfnode] (s) at (0,0) {$K_q$, $\delta=\dim\Sigma$};
\node[pfnode] (c1) at (4,1.5) {Case 1};
\node[pfnode] (c2) at (4,0.3) {Case 2};
\node[pfnode] (c3) at (4,-1.5) {Case 3};
\node[pfnode] (c3c) at (8.4,-0.4) {Case 3c};
\node[pfnode] (c3a) at (8.4,-1.5) {Case 3a};
\node[pfnode] (c3b) at (8.4,-2.6) {Case 3b};
\node[pfgoal] (o1) at (12,0.6) {(1)};
\node[pfgoal] (o2) at (12,-1.5) {(2)};
\node[pfgoal] (o3) at (12,-2.6) {(3)};
\draw[pfedge] (s) -- node[pflab,sloped,above] {$\delta=3$} (c1);
\draw[pfedge] (s) -- node[pflab,sloped,below] {$\delta=6$} (c2);
\draw[pfedge] (s) -- node[pflab,sloped,below] {$\delta\in\{4,5\}$} (c3);
\draw[pfedge] (c3) -- node[pflab,sloped,above,pos=.6] {$4$ terms} (c3c);
\draw[pfedge] (c3) -- node[pflab,above,pos=.6] {$1$ term} (c3a);
\draw[pfedge] (c3) -- node[pflab,sloped,below,pos=.6] {$2$ terms} (c3b);
\draw[pfedge] (c1) -- (o1);
\draw[pfedge] (c2) -- (o1);
\draw[pfedge] (c3c) -- (o1);
\draw[pfedge] (c3a) -- (o2);
\draw[pfedge] (c3b) -- (o3);
\end{tikzpicture}
\caption{Case analysis in the proof of Lemma~\ref{app:return-single}.}
\end{inlinefigure}

\begin{itemize}
\item \step{Case 1:}{$\delta=3$.} In this case, $\Sigma=L_A=L_B$, since both Lagrangians lie in the three-dimensional $\Sigma$. So $K_q\A_{M_A}K_q^\dagger=BP_q\A_{L_A}P_qB^\dagger=\A_{M_B}$, and uniqueness gives $M_A=M_B=N$. Thus $(\Sigma,N)\in\mathscr L(A)\cap\mathscr L(B)$, and $(\Span\{z_1\}\oplus\Sigma,\Span\{z_1\}\oplus N)$ is a GSC pair for $W$ by Lemma~\ref{lem:tool-blocks}(1). This is (1).

\item \step{Case 2:}{$\delta=6$.} Since every $K_v$ normalizes $\A_N$, so does $a_v=K_0^\dagger K_v$ for every $v\in V_3$. So the MASA $A^\dagger\A_NA$ is invariant under conjugation by every Pauli, and Proposition~\ref{lem:lagrangian-expanded} gives $A^\dagger\A_NA=\A_L$ for a Lagrangian $L$. Then $(L,N)\in\mathscr L(A)$, and $B\A_LB^\dagger=K_0\A_NK_0^\dagger=\A_N$. As before this gives (1).

\item \step{Case 3:}{$\delta\in\{4,5\}$.} We construct $0\ne r\in N$ such that $K_q$, every $a_s$ with $s\in\mathsf C(A)$ and every $b_t$ with $t\in\mathsf C(B)$ map $P_r$ to $\pm P_r$, and show that $H=W^\dagger(I\otimes P_r)W$ has one, two or four Pauli terms.

For a gate $K$ normalizing $\A_N$, let $R(K)=\{r\in N:KP_rK^\dagger\in\pm\Pau(3)\}$, and let $R_0=R(K_q)$. Every $K_v$ with $v\in q+\Sigma$ normalizes $\A_N$, and $\langle\Pau(3),K_v\rangle$ is a finite $2$-group by Lemma~\ref{app:return-block}(2). In Clifford coordinates in which $\A_N$ is diagonal, $K_v$ is monomial, and its permutation preserves a complete affine flag by Lemma~\ref{lem:tool-two-group}(2). In coordinates adapted to the flag it reads $(x+c,\,y+g(x),\,z+h(x,y))$, and its inverse has the same form. As $K_vZ^rK_v^\dagger=(-1)^{r\cdot\pi_v^{-1}(x)}$ and $K_v^\dagger Z^rK_v=(-1)^{r\cdot\pi_v(x)}$, both are Paulis exactly when $r$ has zero coefficient on the last flag coordinate or the last coordinate of $\pi_v$ is affine. Hence
\[
R(K_v)=\{r\in N:K_v^\dagger P_rK_v\in\pm\Pau(3)\},
\]
which is $N$ or a plane, and $\dim R_0\in\{2,3\}$. Hence $K_q$ maps $R_0$ onto itself by a linear map $\Phi$.

Next we show that for $v\in\Sigma$, $a_v$ acts on $R_0$ by an elementary abelian group $\mathcal E$ such that $\Phi e$ is unipotent for every $e\in\mathcal E$. For $v\in\Sigma$, write $v=s+t$ with $s\in\mathsf C(A)$ and $t\in\mathsf C(B)$. Then $a_v\doteq a_sa_t$ and $a_t\doteq K_q^\dagger b_tK_q$, so every $a_v$ normalizes $\A_N$, and $K_{q+v}\doteq K_qa_v$. For $s\in\mathsf C(A)$, the Clifford $a_s$ acts on $N$ by a linear map $\bar a_s$, and $K_{q+s}\doteq K_qa_s$ gives $R(K_{q+s})=\bar a_s^{-1}R_0$ from inputs and $R(K_{q+s})=R_0$ from outputs. Hence $\bar a_sR_0=R_0$. Likewise, $K_{q+t}\doteq b_tK_q$ gives $\bar b_tR_0=R_0$ for $t\in\mathsf C(B)$. Consequently, $a_t=K_q^\dagger b_tK_q$ maps each $P_r$ with $r\in R_0$ to $\pm P_{\tilde a_tr}$ with $\tilde a_t=\Phi^{-1}\bar b_t\Phi$. As a result, every $a_v$ acts on $R_0$ by a linear map $\tilde a_v$, with $\tilde a_s=\bar a_s$ for $s\in\mathsf C(A)$. Since $a_va_w=\pm a_wa_v$ and $a_v^2\doteq I$, the maps $\tilde a_v$ commute and square to the identity, so $\mathcal E=\{\tilde a_v:v\in\Sigma\}$ is elementary abelian. Finally, $\Phi\tilde a_v$ is the action of $K_qa_v\doteq K_{q+v}$ on $R_0$, which has $2$-power order because $\langle\Pau(3),K_{q+v}\rangle$ is a finite $2$-group. So every $\Phi e$ is unipotent.

Now we find a common fixed vector $r\ne0$ of $\Phi$ and $\mathcal E$ in $R_0$. For $e\in\mathcal E$, $(e-I)^2=0$, so $\operatorname{rank}(e-I)\le1$ as $\dim R_0\le3$. Commuting rank-one square-zero maps have zero product, so $\{e-I:e\in\mathcal E\}$ is a linear space of maps of rank at most one. Hence these maps have a common image line or a common kernel hyperplane. Write $e=I+u_e\varphi_e$. Since $\Phi$ and $\Phi e$ are both unipotent,
\[
1=\frac{\det(tI-\Phi e)}{\det(tI-\Phi)}=1-\varphi_e(tI-\Phi)^{-1}\Phi u_e,
\]
so $\varphi_e\Phi^ju_e=0$ for all $j\ge1$, and $\varphi_eu_e=0$. In either case, the $\Phi$-invariant span of the vectors $\Phi^ju_e$ lies in every $\ker\varphi_e$, hence in the fixed space of $\mathcal E$, and a unipotent $\Phi$ has a nonzero fixed vector $r$ there (if $\mathcal E$ is trivial, take any fixed vector of $\Phi$).

Let $R=I\otimes P_r$, so that $H=W^\dagger RW$. Write $K_qP_rK_q^\dagger=(-1)^\varepsilon P_r$ and $A^\dagger P_rA=\sum_s\alpha_sP_s$. Since $K_q=BP_qA^\dagger$,
\[
B^\dagger P_rB=(-1)^\varepsilon P_q(A^\dagger P_rA)P_q,
\qquad
H=\sum_s\alpha_s Z_1^{\varepsilon+\omega(q,s)}\otimes P_s.
\]
Let $S=\operatorname{supp}_P(A^\dagger P_rA)$. By the display, also $S=\operatorname{supp}_P(B^\dagger P_rB)$, and the terms of $H$ correspond to the labels in $S$. Since $r$ is fixed by $\mathcal E$, every $a_s$ with $s\in\mathsf C(A)$ maps $P_r$ to $\pm P_r$, so $P_s(A^\dagger P_rA)P_s=\pm A^\dagger P_rA$ and $\omega(s,\cdot)$ is constant on $S$. Likewise, $b_t=K_qa_tK_q^\dagger$ maps $P_r$ to $\pm P_r$ for $t\in\mathsf C(B)$, so $\omega(t,\cdot)$ is constant on $S$. Hence $S$ lies in one coset of $\Sigma^\perp$, and $H$ has at most $2^{6-\delta}\le4$ Pauli terms.

Finally, the number of terms is one, two, or four. By Lemma~\ref{app:return-block}(4), $A=\lambda C_3DC_4$ with $C_3,C_4$ Cliffords and $D$ monomial with eighth-root entries. After a scalar, all three have entries in $\Z[\zeta_8,1/2]$. Hence $\alpha_s=2^{-3}\tr(P_sA^\dagger P_rA)\in\Z[\zeta_8,1/2]$, and $\alpha_s$ is real because $H$ is Hermitian, so $\alpha_s\in\Z[\sqrt2,1/2]$. By Lemma~\ref{app:return-block}(3), $\sigma(A^\dagger P_rA)=Z^pA^\dagger P_rAZ^p$, so $\sigma(\alpha_s)=(-1)^{\omega((0,p),s)}\alpha_s$. Since $(0,p)\in L_A\subseteq\Sigma$ and $S$ lies in one coset of $\Sigma^\perp$, this sign is the same for all $s\in S$. Thus the $\alpha_s$ lie either all in $\Z[1/2]$ or all in $\sqrt2\,\Z[1/2]$, and $\sum_s\alpha_s^2=2^{-4}\tr H^2=1$. Write $\alpha_s=m_s/2^k$, respectively $\alpha_s=\sqrt2\,m_s/2^k$, with $m_s\in\Z$ and $k\ge0$ minimal. Then $\sum_sm_s^2=4^k$, respectively $2^{2k-1}$, with at most four terms. Odd squares are $1$ modulo $8$, and minimality makes some $m_s$ odd. Reducing modulo $8$ leaves only $k=0$ with one term $\pm1$, $k=1$ with four terms $\pm1/2$, and, in the second case, $k=1$ with two terms $\pm1/\sqrt2$. This gives the three subcases below.

\item \step{Case 3a:}{$H$ has one term.} Write $H=\pm P_t$. Then $WP_tW^\dagger=\pm I\otimes P_r$, so $(s,t)$ with $P_s=I\otimes P_r$ is a Pauli pair of $W$. Since $W$ commutes with $Z_1$ and $I\otimes P_r$ commutes with $Z_1$, so does $P_t$. Also $P_t\ne Z_1$ because $WZ_1W^\dagger=Z_1\ne\pm I\otimes P_r$. This is (2).

\item \step{Case 3b:}{$H$ has two terms.} Then $H=(P_a\pm P_b)/\sqrt2$, and $H^2=I$ forces $P_a$ and $P_b$ to anticommute, that is, $\omega(a,b)=1$. This is (3).

\item \step{Case 3c:}{$H$ has four terms.} Then $\delta=4$, since $2^{6-\delta}\ge4$. Write $A=\lambda C_3DC_4$ as in Lemma~\ref{app:return-block}(4) for the GSC pair $(L_A,M_A)$, with $D=D'T_{p'}$, $Z^{p'}\doteq C_4Z^pC_4^\dagger$ and $D'=\pi_0i^{axy+bxz+cyz+dxyz}$. The form $D'$ is determined by $(\pi_0,a,b,c,d)\in\{I,\pi\}\times\F_2^3\times\Z/4$, which gives $64$ forms, and $D$ also by $p'\in\F_2^3$, which gives $512$. The labels $v$ and the Paulis $P$ range over $V_3$, and Clifford membership is decided on six Paulis. Running over these forms gives the following.

\begin{finitefact}{cert:pauli-terms-and-cubic-coefficient}{Pauli terms and the cubic coefficient}{Lemma~\ref{app:return-single}}
Let $D=\pi_0\,i^{axy+bxz+cyz+dxyz}$ with $\pi_0\in\{I,\pi\}$, $a,b,c\in\F_2$ and $d\in\Z/4$. (a) If $d$ is odd, then $D^\dagger PD$ has at least eight Pauli terms for every non-diagonal three-qubit Pauli $P$. (b) If $d$ is even, then for every $p\in\F_2^3$ the labels $v$ with $DT_pP_vT_p^\dagger D^\dagger\in\C_2(3)$ form a subspace of dimension at least five.
\tcblower
\textbf{Code.} \certlink{pauli-terms-and-cubic-coefficient}
\end{finitefact}
If $d$ were even, \textup{(\ref{cert:pauli-terms-and-cubic-coefficient})}(b) would make the labels $v$ with $DP_vD^\dagger$ Clifford a subspace of dimension at least five. As $\mathsf C(A)$ is its image under the symplectic action of $C_4^{-1}$, this would give $\delta\ge5$. So $d$ is odd, and by \textup{(\ref{cert:pauli-terms-and-cubic-coefficient})}(a), $D'^\dagger PD'$ has at least eight Pauli terms for every non-diagonal Pauli $P$. The same holds for $B$. Since the coefficients of $H$ are $\pm1/2$, $\sigma(H)=H$. By Lemma~\ref{app:return-block}(3), $\sigma(H)=(Z_1^e\otimes Z^p)H(Z_1^e\otimes Z^p)$. Since $H$ commutes with $Z_1$, it commutes with $I\otimes Z^p$, and $A^\dagger P_rA$ commutes with $Z^p$. So $C_4(A^\dagger P_rA)C_4^\dagger$ commutes with $C_4Z^pC_4^\dagger\doteq Z^{p'}$, hence with $T_{p'}$, and $C_4(A^\dagger P_rA)C_4^\dagger=D'^\dagger(C_3^\dagger P_rC_3)D'$ has four Pauli terms. By \textup{(\ref{cert:pauli-terms-and-cubic-coefficient})}(a), $C_3^\dagger P_rC_3$ is a diagonal Pauli $Z^u$. The same holds for $B$. The diagonal factor of $D'$ commutes with $Z^u$, so
\[
D'^\dagger Z^uD'=\pi_0^\dagger Z^u\pi_0,\qquad
\pi^\dagger Z^u\pi=Z^u\CZ_{xy}^{u_z}.
\]
Four terms force $\pi_0=\pi$ and $u_z=1$. The four labels of $Z^u\CZ_{xy}$ then span $L_Z$. As $C_4\A_{L_A}C_4^\dagger=\A_Z$, the support of $A^\dagger P_rA$ spans $L_A$, and likewise that of $B^\dagger P_rB$ spans $L_B$. By the display in Case~3 the two supports agree, so $L_A=L_B$. As in Case~1, uniqueness of $\mathscr L(K_q)$ gives a common GSC pair of $A$ and $B$, hence (1). \qedhere
\end{itemize}
\end{proof}

\begin{lemma}\label{app:return-balanced}
Let $W=\operatorname{diag}(A,B)$ be as in Lemma~\ref{app:return-block}. If there exist $0\ne r\in V_3$, $a,b\in V_4$ with $\omega(a,b)=1$, and $W^\dagger(I\otimes P_r)W=(P_a\pm P_b)/\sqrt2$, then $\mathscr L(W)\ne\varnothing$.
\end{lemma}
\begin{proof}
Steps~1 and~2 write $W=\Lambda(Q_1\otimes V)$ with $V$ and $I-2E$ monomial, for a projector $E$ on qubits $2,3,4$, and the cases split on $E$.

\begin{inlinefigure}
\begin{tikzpicture}
\node[pfnode] (s) at (0,0) {Steps 1--2};
\node[pfnode] (c1) at (3.8,1.2) {Case 1};
\node[pfnode] (c2) at (3.8,-1.2) {Case 2};
\node[pfnode] (c2a) at (7.8,-0.4) {Case 2a};
\node[pfnode] (c2b) at (7.8,-2) {Case 2b};
\node[pfnode] (two) at (11.4,-0.4)
  {\hyperref[app:return-two]{Lemma~\ref*{app:return-two}}};
\node[pfgoal] (goal) at (14.8,-1.2) {$\mathscr L(W)\ne\varnothing$};
\node[pfgoal] (bot) at (11.4,-3) {Contradiction \textup{(\ref{cert:balanced-rank-two-branch})}};
\draw[pfedge] (s) -- node[pflab,sloped,above] {$E\in\A_Z$} (c1);
\draw[pfedge] (s) -- node[pflab,sloped,below] {$E\notin\A_Z$} (c2);
\draw[pfedge] (c2) -- node[pflab,sloped,above] {$\operatorname{rank}E=1$} (c2a);
\draw[pfedge] (c2) -- node[pflab,sloped,below] {$\operatorname{rank}E=2$} (c2b);
\draw[pfedge] (c2a) -- (two);
\draw[pfedge] (c2b) -- node[pflab,sloped,above] {otherwise} (two);
\draw[pfedge] (c2b) -- node[pflab,sloped,below,pos=.4] {$D=\{s=1\}$} (goal);
\draw[pfedge] (c2b) -- node[pflab,sloped,below,pos=.4] {$D=\{z=1\}$} (bot);
\draw[pfedge] (two) -- (goal);
\draw[pfedge] (c1) -| (goal);
\end{tikzpicture}
\caption{Case analysis in the proof of Lemma~\ref{app:return-balanced}.}
\end{inlinefigure}

\begin{itemize}
\item \step{Step 1:}{$W=\Lambda(Q_1\otimes V)$, where $E$ is $0$ or a diagonal projector of rank one or two.} Let $R=I\otimes P_r$ and $h'_v=WP_vW^\dagger$, so $R=(h'_a\pm h'_b)/\sqrt2$. Since $\omega(a,b)=1$, $h'_a$ and $h'_b$ anticommute, and
\[
J=\mp ih'_ah'_b=\pm h'_{a+b}\in\mathscr R_4,\qquad
h'_aR=(I+iJ)/\sqrt2.
\]
The gate $(I+iJ)/\sqrt2=h'_aR$ has finite level, because conjugating a Pauli by it gives $\pm h'_aP{h'_a}^\dagger$, and $J$ is a finite-level Hermitian involution. Lemma~\ref{lem:quarter-turn}(2) therefore shows that $\operatorname{supp}_PJ$ is isotropic. Since $R$ commutes with $Z_1=WZ_1W^\dagger$, the operator $(P_a\pm P_b)/\sqrt2=W^\dagger RW$ commutes with $Z_1$, so $a$ and $b$ commute with $z_1$. Since $J$ anticommutes with $h'_a$ and $h'_b$, it anticommutes with $R$, so every Pauli in its support anticommutes with $R$. Therefore, there exists a Clifford gate sending $R$ to $X_1$ and $J$ to $Z_1\otimes(-1)^f$, where $f:\F_2^3\to\F_2$ and $(-1)^f$ is the diagonal gate on qubits $2,3,4$: send one support label $u_0$ of $J$ to $z_1$, and a Lagrangian of $\Span\{r,u_0\}^\perp$ that contains the differences of support labels to the diagonal labels of qubits $2,3,4$. Choose another Clifford sending $P_a$ to $X_1$ and $P_{a+b}$ to $Z_1$. After these changes, and possibly an input $X_1$,
\[
WX_1'W^\dagger=X_1,\qquad WZ_1W^\dagger=Z_1\otimes(-1)^f,
\qquad X_1'=(X_1\pm Y_1)/\sqrt2.
\]
The original pair $(z_1,z_1)$ of $W$ becomes a Pauli pair whose input and output commute with $X_1$ and $Z_1$.

Let $Q_1$ be the power of $T_1$ with $Q_1^\dagger X_1Q_1=X_1'$, let $E=\Pi_{f=1}$ be the diagonal projector onto $\{f=1\}$, write $\Pi_\pm^{X_j}=(I\pm X_j)/2$, and let
\[
\Lambda=I-2\Pi_-^{X_1}\otimes E
=\Pi_+^{X_1}\otimes I+\Pi_-^{X_1}\otimes(-1)^f.
\]
Then $\Lambda X_1\Lambda^\dagger=X_1$ and $\Lambda Z_1\Lambda^\dagger =Z_1\otimes(-1)^f$. So $\Lambda^\dagger W(Q_1^\dagger\otimes I)$ commutes with $X_1$ and $Z_1$, and therefore equals $I\otimes V$:
\[
W=\Lambda(Q_1\otimes V).
\]
Output Clifford gates that fix $X_1$ act on $f$: a translation $X^c$ on qubits $2,3,4$ replaces $f$ by $f(x+c)$, a linear permutation of qubits $2,3,4$ by CNOT gates replaces it by $f\circ A$, and the controlled Pauli $\Pi_+^{X_1}\otimes I+\Pi_-^{X_1}\otimes(-1)^\ell$ with $\ell$ affine adds $\ell$. On three bits these reduce $f$ to $0$, $xy$ or $xyz$, so $E=0$ or $E$ is a diagonal projector of rank one or two.

\item \step{Step 2:}{$V$ and $(I-2E)V$ share a GSC pair, so we may take $V$ and $I-2E$ monomial.} Let $U_0=\Lambda(I\otimes V)=W(Q_1^\dagger\otimes I)$. If $P_v$ commutes with $Z_1$, then $Q_1$ commutes with $P_v$ and $U_0P_vU_0^\dagger=h'_v$, and $U_0X_1U_0^\dagger=X_1$. If $P_v$ anticommutes with $Z_1$, write $P_v\doteq X_1P_w$ with $P_w$ commuting with $Z_1$. Then $U_0P_vU_0^\dagger\doteq X_1h'_w$. This operator is a traceless Hermitian involution, being conjugate to a nonidentity Hermitian Pauli, and it is Gaussian because $h'_w$ is. It has finite level, since it conjugates a Pauli $P$ to $\pm X_1(h'_wP{h'_w}^\dagger)X_1$. Finally, $\Gamma(X_1h'_w)=X_1\Gamma(h'_w)X_1$, and $\Gamma(h'_w)\le\langle\Pau(4),h'_w\rangle$ is a finite $2$-group by Proposition~\ref{app:lem-R_n}. Therefore $X_1h'_w$ is generalized semi-Clifford by Theorem~\ref{thm:two-group-chamber}. Hence, by Proposition~\ref{app:lem-R_n}, $U_0P_vU_0^\dagger\in\mathscr R_4$ for every $v\ne0$. The pair $(x_1,x_1)$ and the original pair are commuting, independent Pauli pairs of $U_0$, so Lemma~\ref{app:return-two} gives a GSC pair for $U_0$. Since $U_0=\Pi_+^{X_1}\otimes V+\Pi_-^{X_1}\otimes(I-2E)V$, the gate $H_1U_0H_1=\operatorname{diag}(V,(I-2E)V)$ commutes with $Z_1$, and Lemmas~\ref{lem:tool-transport}(1) and~\ref{lem:tool-blocks}(1) give $(L,M)\in\mathscr L(V)\cap\mathscr L((I-2E)V)$.

Choose $C_1,C_2\in\C_2(3)$ with $C_2\A_LC_2^\dagger=\A_Z$ and $C_1\A_ZC_1^\dagger=\A_M$, and replace $V$ by $C_1^\dagger VC_2^\dagger$ and $E$ by $C_1^\dagger EC_1$. This replaces $W$ by a Clifford multiple of itself. Now, $V$ and $(I-2E)V$ both normalize $\A_Z$. Hence $V=\rho\Delta$ with $\rho$ a permutation and $\Delta$ diagonal, and $I-2E$ is monomial. The projector $E$ is now a stabilizer projector, not necessarily diagonal.

\item \step{Case 1:}{$E$ is diagonal.} Then $\Lambda$ commutes with every $I\otimes Z^c$ and $\Lambda Z_1\Lambda^\dagger=Z_1\otimes(I-2E)$ is diagonal. Since $Q_1$ commutes with $Z_1$ and $V$ normalizes $\A_Z$, $W=\Lambda(Q_1\otimes V)$ normalizes the diagonal MASA of all four qubits. This is a GSC pair for $W$.

\item \step{Case 2:}{$E$ is not diagonal.} We first determine $E$. Let $\mathcal T\le\F_2^3$ be the space of $X$-parts of the stabilizer group of $E$. The entry $E_{x,y}$ vanishes unless $y-x\in\mathcal T$, and each nonzero column of $E$ has exactly $2^{\dim\mathcal T}$ nonzero entries, of modulus $2^{-\dim\mathcal T}$, including its diagonal entry. Monomiality of $I-2E$ therefore forces $\dim\mathcal T=1$, say $\mathcal T=\Span\{u\}$. After a diagonal Clifford,
\[
E=\Pi_D\cdot\tfrac12(I-X^u),
\]
where $\Pi_D$ is the diagonal projector onto an affine set $D$ with $D+u=D$. On $\Span_{\CC}\{\ket x:x\in D\}$ the gate $I-2E$ acts as $X^u$, and it is the identity elsewhere. Thus $I-2E$ is the permutation gate of $\tau(x)=x+u\,[x\in D]$, and $\operatorname{rank}E=|D|/2$.

Next, $\rho$ and $\tau\rho$ are affine or affinely equivalent to $(x,y,z)\mapsto(x,y,z+xy)$. Indeed, since every $U_0P_vU_0^\dagger$ lies in $\mathscr R_4$, $U_0$ has finite level. Conjugating by $H_1$ makes $U_0$ block diagonal with blocks $V$ and $(I-2E)V$, so both blocks have finite level by Lemma~\ref{lem:tool-blocks}(2). This is unchanged by the Clifford changes above. Apply Proposition~\ref{lem:tool-monomial}(1) to the monomial gates $V=\rho\Delta$ and $(I-2E)V=\tau\rho\Delta$ with the GSC pair $(L_Z,L_Z)$: the Cliffords $C_1,C_2$ there normalize $\A_Z$, so their permutation parts are affine.

Finally, two independent $c\in\F_2^3$ with $c\cdot u=0$ and $c\cdot\rho$ affine give $\mathscr L(W)\ne\varnothing$. Let $c\cdot u=0$. Then $Z^c$ commutes with $I-2E$. If moreover $c\cdot\rho$ is affine, then $V^\dagger Z^cV=(-1)^{c\cdot\rho(x)}$ is a diagonal Pauli, and since $Q_1$ acts on the first qubit, $W$ sends it on qubits $2,3,4$ to $I\otimes Z^c$ in both eigenspaces of $X_1$. Two independent such $c$ therefore give two commuting, independent Pauli pairs of $W$, and Lemma~\ref{app:return-two} gives $\mathscr L(W)\ne\varnothing$.

\item \step{Case 2a:}{$\operatorname{rank}E=1$.} Since $|D|=2$, $\tau$ is a transposition. Affine permutations of $\F_2^3$ are even, and $(x,y,z)\mapsto(x,y,z+xy)$ is a transposition, so exactly one of $\rho$ and $\tau\rho$ is affine. For $c$ with $c\cdot u=0$ we have $c\cdot\tau\rho=c\cdot\rho$, so $c\cdot\rho$ is affine. The two independent choices of $c$ give $\mathscr L(W)\ne\varnothing$ as in Case~2.

\item \step{Case 2b:}{$\operatorname{rank}E=2$.} Name the coordinates of qubits $2,3,4$ as $(s,y,z)$. Here $\rho$ and $\tau\rho$ are affine or affinely equivalent to $(x,y,z)\mapsto(x,y,z+xy)$, $u\ne0$, and $D$ is an affine plane with $D+u=D$, since $\operatorname{rank}E=2$. These data range over finite sets. The replacements $(\rho,u,D)\mapsto(\alpha\rho\beta,\bar\alpha u,\alpha(D))$, with $\alpha,\beta$ among the $1{,}344$ affine permutations, multiply $W$ by monomial Clifford gates and preserve the number of independent $c$ as in Case~2, so the triples with fewer than two such $c$ are found by a finite orbit computation, which gives the following.

\begin{finitefact}{cert:normal-form-of-exceptional-balanced-configurations}{Normal form of the triples $(\rho,u,D)$}{Lemma~\ref{app:return-balanced}}
Let $\tau_0(x,y,z)=(x,y,z+xy)$, and let $\rho$ be a permutation of $\F_2^3$ that is affine or affinely equivalent to $\tau_0$. Let $u\ne0$ and let $D$ be an affine plane with $D+u=D$ such that $\tau\rho$, where $\tau(x)=x+u\,[x\in D]$, is again affine or affinely equivalent to $\tau_0$. If fewer than two independent $c$ with $c\cdot u=0$ make $c\cdot\rho$ affine, then after replacing $(\rho,u,D)$ by $(\alpha\rho\beta,\bar\alpha u,\alpha(D))$ for affine $\alpha,\beta$, with $\bar\alpha$ the linear part of $\alpha$, in coordinates $(s,y,z)$ we have $\rho(s,y,z)=(s,y,z+sy)$, $u=e_y$, and $D=\{s=1\}$ or $D=\{z=1\}$.
\tcblower
\textbf{Code.} \certlink{normal-form-of-exceptional-balanced-configurations}
\end{finitefact}
By \textup{(\ref{cert:normal-form-of-exceptional-balanced-configurations})}, there are two independent $c$ as in Case~2, which gives $\mathscr L(W)\ne\varnothing$, unless after such a replacement $\rho(s,y,z)=(s,y,z+sy)$, $u=e_y$, and $D=\{\ell=1\}$ with $\ell\in\{s,z\}$. We split into further cases by $\ell$.
\begin{itemize}
\item \step{Case $\ell=z$:}{this case does not occur.} Here $\rho$, $u$ and $D$ are explicit, and by \textup{(\ref{cert:balanced-rank-two-branch})}, for the correction $h=WX_sW^\dagger$ the gate $Z_1h\in\langle\Pau(4),h\rangle$ has order divisible by three, although this group is a finite $2$-group by Proposition~\ref{app:lem-R_n}, a contradiction.

\item \step{Case $\ell=s$:}{$D=\{s=1\}$.} Let $C=(H_1H_y\CZ_{1y}H_1H_y)\CNOT_{y\to z}$. On the eigenspace $s=0$, $E$ vanishes, so $\Lambda=I$, and $\rho$ is the identity. On $s=1$, $\Lambda=I-2\Pi_-^{X_1}\otimes\Pi_-^{X_y}=H_1H_y\CZ_{1y}H_1H_y$ and $\rho=\CNOT_{y\to z}$. Hence $W=\operatorname{diag}(G_0,G_1)$ with respect to $Z_s$, with blocks on the qubits $1,y,z$,
\[
G_0=Q_1\Delta_0,\qquad G_1=CQ_1\Delta_1,
\]
where $\Delta_0,\Delta_1$ are diagonal gates on $y,z$.

We now determine $\Delta_0$ and $\Delta_1$. For $t\in\F_2^2$, the Pauli $X_sX_1X^t$ anticommutes with $Z_s$, and the off-diagonal block of its conjugate is $G_1X_1X^tG_0^\dagger$. Since $Q_1X_1Q_1^\dagger\doteq S_1X_1$ and $\Delta_1X^t\Delta_0^\dagger=\operatorname{diag}(\theta_t)X^t$ with $\theta_t(x)=\Delta_1(x)/\Delta_0(x+t)$, this block equals $C(S_1\otimes\operatorname{diag}\theta_t)$ up to a Pauli factor and a scalar. By Proposition~\ref{app:lem-R_n} and Lemma~\ref{lem:tool-blocks}(3), with $Z_s$ in place of $Z_1$, it generates a finite $2$-group with the Paulis, so it normalizes a Pauli MASA, and \textup{(\ref{cert:balanced-rank-two-branch})} then makes every $\operatorname{diag}\theta_t$ a diagonal Clifford. In particular, every ratio $\theta_t/\theta_0=\Delta_0(x)/\Delta_0(x+t)$ is a diagonal Clifford. Modulo diagonal Clifford gates common to both blocks, and the input gates $\CZ_{sy}$, $\CZ_{sz}$, this gives
\[
\Delta_0=T_y^aT_z^b\mathrm{CS}_{yz}^c,\qquad
\Delta_1=\lambda S_y^dS_z^e\CZ_{yz}^f\Delta_0,
\quad a,b,c,d,e,f\in\F_2.
\]
For the Paulis $X_sX_1$, $X_sX_1X_y$, $X_sX_1X_z$, the off-diagonal blocks are Clifford and, as above, generate finite $2$-groups with the Paulis, so their symplectic actions are unipotent by Lemma~\ref{lem:tool-two-group}(2). Up to a scalar, $G_0$ and $G_1$ are explicit gates determined by $(a,\dots,f)\in\F_2^6$, so the three off-diagonal blocks are explicit, and a symplectic map $\gamma$ of $V_3$ is unipotent exactly when $(\gamma+I)^6=0$. The two earlier uses of \textup{(\ref{cert:balanced-rank-two-branch})} are finite in the same way: for $\ell=z$ all data are explicit, and for $\ell=s$ the values of $\theta_t$ are ratios of entries of $\Delta$, which lie in $\mu_8$ as in Proposition~\ref{lem:tool-monomial}(2), since $W$ is defined over $\Q(\zeta_8)$ up to scalar (Lemma~\ref{app:return-block}(3)). These checks give the following.

\begin{finitefact}{cert:balanced-rank-two-branch}{The case $\operatorname{rank}E=2$}{Lemma~\ref{app:return-balanced}}
In Case~2b of the proof of Lemma~\ref{app:return-balanced}: for $\ell=z$ and the correction $h=WX_sW^\dagger$, the monomial gate $Z_1h$, in the basis $\ket\pm\otimes\ket x$, has a permutation of cycle type $2^26^2$, so this case does not occur. For $\ell=s$, the Pauli-MASA condition on the off-diagonal block makes every $\operatorname{diag}\theta_t$ Clifford, and among the $64$ choices of $a,b,c,d,e,f\in\F_2$ the three unipotence conditions hold exactly when $a=b=c=e=0$ and $d=f$.
\tcblower
\textbf{By hand} ($\ell=z$). In the eigenbasis of $X_1$, $h=\operatorname{diag}(M,\tau M\tau)$ with $M$ monomial of permutation $\sigma(s,y,z)=(s+1,y,y+z)$ and $\tau(s,y,z)=(s,y+z,z)$. So $(Z_1h)^2$ has permutation $\tau\sigma\tau\sigma(s,y,z)=(s,y+z,y)$, of cycle type $1^23^2$.

\textbf{Code.} \certlink{balanced-rank-two-branch}
\end{finitefact}
By \textup{(\ref{cert:balanced-rank-two-branch})}, the three conditions are, respectively,
\[
(d+f,e)=(0,0),\quad
(d+f+a+c,e+c)=(0,0),\quad
(d+f,e+b)=(0,0).
\]
Hence $a=b=c=e=0$ and $d=f$. For $d=f=0$ both blocks normalize the Pauli MASA generated by $Z_1,X_y,X_z$, and for $d=f=1$ the one generated by $Z_1,X_y,Y_z$. Adjoining $Z_s$ gives a GSC pair for $W$. \qedhere
\end{itemize}
\end{itemize}
\end{proof}

\begin{lemma}\label{app:return-multiple}
Let $W=\operatorname{diag}(A,B)$ be as in Lemma~\ref{app:return-block}, with $K_v=BP_vA^\dagger$, and suppose that $|X(K_v)|\ge2$ for every $v\in V_3$. Then at least one of the following holds.
\renewcommand{\labelenumi}{(\theenumi)}
\begin{enumerate}
\item $\mathscr L(W)\ne\varnothing$.
\item $W$ has a Pauli pair $(s,t)$ with $P_t\ne Z_1$ and $P_tZ_1=Z_1P_t$.
\end{enumerate}
\end{lemma}
\begin{proof}
The mixed block $K_v$ is Gaussian, since $\sigma(K_v)\doteq BZ^pP_vZ^pA^\dagger\doteq K_v$ by Lemma~\ref{app:return-block}(3). By Lemma~\ref{app:return-block}(2), $\langle\Pau(3),K_v\rangle$ is a finite $2$-group and $K_v$ normalizes a Pauli MASA, so $K_v$ has a GSC pair, and it has finite level by the second paragraph of the proof of Proposition~\ref{app:lem-R_n}. Proposition~\ref{lem:tool-monomial} gives $K_v\doteq C_1\Pi_\pi i^fC_2$ with $C_1,C_2\in\C_2(3)$, $\pi=1$ or $\pi(x,y,z)=(x,y,z+xy)$, and $f:\F_2^3\to\Z/4$. Write $f=\ell+2q+axy+bxz+cyz+dxyz$ with $\deg\ell\le1$, $q$ quadratic, $a,b,c\in\{0,1\}$ and $d\in\Z/4$. Then $i^{\ell+2q}$ is a diagonal Clifford gate, which we absorb into $C_2$. So $K_v\doteq C_1DC_2$ with $D$ one of the $64$ forms of \textup{(\ref{cert:mixed-blocks-with-two-input-lagrangians})}, and Clifford factors change neither membership in $\C_3(3)$ nor $|X(\cdot)|$ (Lemma~\ref{lem:tool-transport}(1)), so it suffices to examine these $64$ forms.

\begin{finitefact}{cert:mixed-blocks-with-two-input-lagrangians}{Mixed blocks with two input Lagrangians}{Lemma~\ref{app:return-multiple}}
Let $D=\pi_0\,i^{axy+bxz+cyz+dxyz}$ with $\pi_0\in\{I,\pi\}$, $a,b,c\in\F_2$ and $d\in\Z/4$, a block form of Lemma~\ref{app:return-block}\textup{(4)} with $p'=0$. If $|X(D)|\ge2$, then $D\in\C_3(3)$.
\tcblower
\textbf{Code.} \certlink{mixed-blocks-with-two-input-lagrangians}
\end{finitefact}
By \textup{(\ref{cert:mixed-blocks-with-two-input-lagrangians})}, $K_v\in\C_3(3)$ for every $v\in V_3$. Conjugating $W$ by $X_1$ swaps $A$ and $B$ and replaces each $K_v$ by $K_v^\dagger$, which preserves the hypothesis. We may therefore assume that $A\notin\C_3(3)$ or that both blocks lie in $\C_3(3)$. Let $a_v=AP_vA^\dagger$ and $b_v=BP_vB^\dagger$, so that $K_v=K_0a_v$.

Write $\mathcal B(A,B)$ for the alternative
\[
\mathscr L(A)\cap\mathscr L(B)\ne\varnothing
\quad\text{or}\quad
\exists\,r,q\ne0:\ AP_rA^\dagger\doteq BP_rB^\dagger\doteq P_q.
\]
It suffices to prove $\mathcal B(A,B)$. Indeed, if $(L,M)\in\mathscr L(A)\cap\mathscr L(B)$, then $W\A_{\Span\{z_1\}\oplus L}W^\dagger=\A_{\Span\{z_1\}\oplus M}$ by Lemma~\ref{lem:tool-blocks}(1), which is (1). Otherwise, $AP_rA^\dagger=\eta_0P_q$ and $BP_rB^\dagger=\eta_1P_q$ for some $r,q\ne0$ and signs $\eta_0,\eta_1$, so
\[
W(I\otimes P_r)W^\dagger=\eta_0Z_1^\epsilon\otimes P_q,\qquad(-1)^\epsilon=\eta_0\eta_1,
\]
and $P_t=I\otimes P_r$ commutes with $Z_1$ and differs from it, which is (2). In both cases, Lemma~\ref{lem:tool-bridge} turns the finite $2$-groups $\langle\Pau(3),K_v\rangle$ of Lemma~\ref{app:return-block}(2) into unipotence conditions on one Clifford gate $F$, and \textup{(\ref{cert:block-pair-domains})} checks that these conditions give $\mathcal B(A,B)$.

\begin{inlinefigure}
\begin{tikzpicture}
\node[pfnode] (s) at (0,0) {$A,B$};
\node[pfnode] (c1) at (3.6,0.9) {Case 1};
\node[pfnode] (c2) at (3.6,-0.9) {Case 2};
\node[pfnode] (b) at (7.4,0) {$\mathcal B(A,B)$};
\node[pfgoal] (o1) at (11.2,0.9) {(1)};
\node[pfgoal] (o2) at (11.2,-0.9) {(2)};
\draw[pfedge] (s) -- node[pflab,sloped,above] {$A\notin\C_3(3)$} (c1);
\draw[pfedge] (s) -- node[pflab,sloped,below] {$A,B\in\C_3(3)$} (c2);
\draw[pfedge] (c1) -- node[pflab,sloped,above] {\textup{(\ref{cert:block-pair-domains})}(a)} (b);
\draw[pfedge] (c2) -- node[pflab,sloped,below] {\textup{(\ref{cert:block-pair-domains})}(b)} (b);
\draw[pfedge] (b) -- node[pflab,sloped,above] {common GSC pair} (o1);
\draw[pfedge] (b) -- node[pflab,sloped,below] {common Pauli pair} (o2);
\end{tikzpicture}
\caption{Case analysis in the proof of Lemma~\ref{app:return-multiple}.}
\end{inlinefigure}

\begin{itemize}
\item \step{Case 1:}{$A\notin\C_3(3)$.} We first reduce $A$ to a normal form and show that $X(K_0)$ contains its input Lagrangian $L_Z$. The $512$ block forms $D$ of Lemma~\ref{app:return-block}(4) are given by $(\pi_0,a,b,c,d,p')\in\{I,\pi\}\times\F_2^3\times\Z/4\times\F_2^3$. Their classes under $D\mapsto C_1DC_2$, $C_1,C_2\in\C_2(3)$, preserve the level, $\C_2(3)$ is finite modulo scalars, and membership in $\C_3(3)$ is decided on a basis of $V_3$. Sorting the forms into classes gives the following.

\begin{finitefact}{cert:block-form-clifford-classes}{Clifford classes of the block forms}{Lemma~\ref{app:return-multiple}}
Under $D\mapsto C_1DC_2$ with $C_1,C_2\in\C_2(3)$, the $512$ block forms $D$ of Lemma~\ref{app:return-block}\textup{(4)} fall into $30$ classes. Ten classes are third-level, and every correction $DP_vD^\dagger$ of the other twenty is third-level. We call the listed representative of a class its normal form. The third-level classes containing a form with $p'=0$ are those of the four gates $\mathrm{CS}_{yz}^a\CCZ^b$.
\tcblower
\textbf{Code.} \certlink{block-form-clifford-classes}
\end{finitefact}
Replace $W$ by $(I\otimes C_1^\dagger)W(I\otimes C_2^\dagger)$ with $C_1,C_2\in\C_2(3)$ so that $A$ is one of the twenty non-third-level normal forms. Then $L_A=L_Z$, $A$ is a block form $\pi_0\,i^{axy+bxz+cyz+dxyz}\,T_{p_A}$ of Lemma~\ref{app:return-block}(4), and Lemma~\ref{app:return-block}(3) gives $\sigma(A)\doteq AZ^{p_A}$ and $\sigma(B)\doteq BZ^{p_A}$. Let $J_A\le\Sp(V_3)$ be generated by the symplectic actions of the Clifford gates $a_vP_wa_v^\dagger$, where $v,w\in V_3$. The set $X(K_0)$ is invariant under $J_A$: as $a_v=K_0^\dagger K_v$, each Clifford $c=a_vP_wa_v^\dagger$ satisfies $K_0c=(K_vP_wK_v^\dagger)K_0$ with $K_vP_wK_v^\dagger$ Clifford, and Lemma~\ref{lem:tool-transport}(2) applies. The normal form $A$ is one of twenty, and $J_A$ depends only on $A$. By the argument before \textup{(\ref{cert:mixed-blocks-with-two-input-lagrangians})}, $K_0\doteq C_1DC_2$ with $C_1,C_2\in\C_2(3)$ and $D$ one of the $64$ forms there, and $D$ is third-level since $K_0$ is. By \textup{(\ref{cert:block-form-clifford-classes})}, $D\doteq CYC'$ with $C,C'\in\C_2(3)$ and $Y$ one of the four gates $\mathrm{CS}_{yz}^a\CCZ^b$. So $K_0\doteq F\,(C''^\dagger YC'')$ with $F=C_1CC'C_2$ and $C''=C'C_2$ in $\C_2(3)$, and since a left Clifford factor does not change input Lagrangians (Lemma~\ref{lem:tool-transport}(1)), $X(K_0)=X(C''^\dagger YC'')$, and it remains to check the following for the twenty normal forms.

\begin{finitefact}{cert:invariant-sets-of-input-lagrangians}{Invariant sets of input Lagrangians}{Lemma~\ref{app:return-multiple}}
Let $A$ be a non-third-level normal form of \textup{(\ref{cert:block-form-clifford-classes})}, and let $D$ be a Clifford conjugate of one of the four gates $\mathrm{CS}_{yz}^a\CCZ^b$. If $X(D)$ is invariant under $J_A$, then $L_Z\in X(D)$.
\tcblower
\textbf{Code.} \certlink{invariant-sets-of-input-lagrangians}
\end{finitefact}
By \textup{(\ref{cert:invariant-sets-of-input-lagrangians})}, $L_Z\in X(K_0)$.

Next we reduce $B$ to a normal form. Since $A$ normalizes $\A_Z$, $K_0\A_ZK_0^\dagger=B\A_ZB^\dagger$ is a Pauli MASA, so $B=F_0B_0$ with $F_0\in\C_2(3)$ and $B_0$ monomial. By Proposition~\ref{lem:tool-monomial}, after a scalar $B_0$ has eighth-root entries and its permutation is affine or affinely equivalent to $\pi$ of Lemma~\ref{app:return-block}(4), and Proposition~\ref{lem:tool-galois}(1) gives $B_0^{-1}\sigma(B_0)\doteq B^{-1}\sigma(B)\doteq Z^{p_A}$. As in the proof of Lemma~\ref{app:return-block}(4), this gives $B_0\doteq\Pi\,i^gT_{p_A}$ with $\Pi$ a permutation gate and $g:\F_2^3\to\Z/4$. The data are $A$ (twenty normal forms), the permutation of $\Pi$ (among the $10{,}752$ permutations that are affine or affinely equivalent to $\pi$) and $g\in(\Z/4)^{\F_2^3}$. Left multiplication of $B_0$ by monomial Clifford gates, absorbed into $F_0$, acts on these finitely many data, and an orbit computation gives the following.

\begin{finitefact}{cert:reduction-to-common-input-pair-domain}{Reduction of $B_0$ to $256$ gates}{Lemma~\ref{app:return-multiple}}
Let $A$ be a non-third-level normal form with $T$-factor $T_{p_A}$, and $B_0=\Pi_\rho\,i^{g}T_{p_A}$ with $g:\F_2^3\to\Z/4$ and $\rho$ affine or affinely equivalent to $\pi$ of Lemma~\ref{app:return-block}\textup{(4)}. Let $\rho_0$ be the identity, and for $0\ne k\in\F_2^3$ let $\rho_k=\lambda_k^{-1}\pi\lambda_k$ with $\lambda_k$ linear and $\lambda_k(k)=(0,0,1)$. Modulo left multiplication by monomial Clifford gates, $B_0$ is one of the $256$ gates $\Pi_{\rho_k}\,i^{axy+bxz+cyz+dxyz}\,T_{p_A}$ with $k\in\F_2^3$, $a,b,c\in\F_2$ and $d\in\Z/4$, or $\rho(x+k)+\rho(x)$ does not depend on $x$ exactly for $k\in\{0,(1,1,1)\}$, and then some $B_0P_vA^\dagger\notin\C_3(3)$.
\tcblower
\textbf{Code.} \certlink{reduction-to-common-input-pair-domain}
\end{finitefact}
By \textup{(\ref{cert:reduction-to-common-input-pair-domain})}, after absorbing monomial Cliffords into $F_0$, $B_0$ is one of the $256$ gates listed there, or the permutation $\rho$ of $B_0$ is not affine and $\rho(x+k)+\rho(x)$ does not depend on $x$ exactly for $k\in\{0,(1,1,1)\}$. In the latter case some $B_0P_vA^\dagger$ is not third-level, contrary to \textup{(\ref{cert:mixed-blocks-with-two-input-lagrangians})}.

Now every $D_v=B_0P_vA^\dagger$ is third-level, since $K_v=F_0D_v$ is. The pair $(A,B_0)$ is one of the finitely many pairs left by \textup{(\ref{cert:reduction-to-common-input-pair-domain})}, so the $D_v$ are explicit, and their decompositions are found by a search over $\C_2(3)$, which is finite modulo scalars. The unipotence conditions below and $\mathcal B(A,B)$ depend on $F$ only through $\gamma_F\in\Sp(6,2)$, a group of order $1{,}451{,}520$, since $\mathscr L(FB')=\{(L,\gamma_FM):(L,M)\in\mathscr L(B')\}$ by Lemma~\ref{lem:tool-transport}(1) and $FP_qF^\dagger\doteq P_{\gamma_Fq}$. A search over $\Sp(6,2)$ for each pair gives the following.

\begin{finitefact}{cert:block-pair-domains}{Pair domains for blocks}{Lemma~\ref{app:return-multiple}}
(a) Let $A$ be a non-third-level normal form of \textup{(\ref{cert:block-form-clifford-classes})} and $B_0$ one of the $256$ gates of \textup{(\ref{cert:reduction-to-common-input-pair-domain})} for this $A$ ($5{,}120$ pairs). For the $360$ pairs with every $D_v=B_0P_vA^\dagger\in\C_3(3)$, each $D_v=G_vY_v$ with $G_v\in\C_2(3)$ and $Y_v$ a Clifford conjugate of some $\mathrm{CS}_{yz}^a\CCZ^b$, and $\mathcal B(A,FG_0^{-1}B_0)$ holds for every $F\in\C_2(3)$ with $\gamma_F\gamma_{G_0^{-1}G_v}e$ unipotent for all $v\in V_3$ and $e\in E_{Y_v}$.

(b) Let $A$ be a third-level normal form and $D$ a Clifford conjugate of some $\mathrm{CS}_{yz}^a\CCZ^b$ with $X(D)$ invariant under $E_A$ ($13{,}960$ pairs, reduced to $2{,}076$). Then $\mathcal B(A,FDA)$ holds for every $F\in\C_2(3)$ with $\gamma_Fef$ unipotent for all $e\in E_D$ and $f\in E_A$.
\tcblower
\textbf{Code.} \certlink{block-pair-domains}
\end{finitefact}
By \textup{(\ref{cert:block-pair-domains})}(a), $D_v=G_vY_v$ with $G_v\in\C_2(3)$ and $Y_v$ a Clifford conjugate of some $\mathrm{CS}_{yz}^a\CCZ^b$. Let $F=F_0G_0$, so that $B=F_0B_0=FG_0^{-1}B_0$. Then $K_v=F_0D_v=(FG_0^{-1}G_v)Y_v$, and $\langle\Pau(3),K_v\rangle$ is a finite $2$-group, so Lemma~\ref{lem:tool-bridge} makes $\gamma_F\gamma_{G_0^{-1}G_v}e$ unipotent for every $v\in V_3$ and $e\in E_{Y_v}$. By \textup{(\ref{cert:block-pair-domains})}(a), these conditions give $\mathcal B(A,B)$.

\item \step{Case 2:}{$A,B\in\C_3(3)$.} We first reduce $K_0$ to a normal form. As in Case~1, take $A$ to be one of the ten third-level normal forms of \textup{(\ref{cert:block-form-clifford-classes})}, which are diagonal. The mixed block $K_0=BA^\dagger$ is Gaussian, since $\sigma(K_0)\doteq BZ^pZ^pA^\dagger=K_0$ by Lemma~\ref{app:return-block}(3). It lies in $\C_3(3)$ by \textup{(\ref{cert:mixed-blocks-with-two-input-lagrangians})}, and has a GSC pair. By Proposition~\ref{lem:tool-monomial} it is $C_1D_0C_2$ up to scalar, with $C_1,C_2$ Clifford and $D_0=\pi_0i^{axy+bxz+cyz+dxyz}$, a form $D$ of Lemma~\ref{app:return-block}(4) with $p'=0$. It is third-level since $K_0$ is, so by \textup{(\ref{cert:block-form-clifford-classes})} it lies in the Clifford class of one of the four gates $\mathrm{CS}_{yz}^a\CCZ^b$. So $K_0\doteq CYC'$ with $C,C'$ Clifford and $Y$ one of these gates, and $K_0\doteq(CC')(C'^\dagger YC')=FD$ with $F\in\C_2(3)$ and $D$ a Clifford conjugate of $Y$. The set $X(D)=X(K_0)$ is invariant under $E_A$: both $a_v$ and $b_v$ are Clifford, and $K_0a_v=K_v\doteq b_vK_0$, so Lemma~\ref{lem:tool-transport}(2) applies.

For every $v$, $K_v=(Fa_v)(a_v^\dagger Da_v)$ with $a_v^\dagger Da_v$ again a Clifford conjugate of some $\mathrm{CS}_{yz}^a\CCZ^b$, and $\langle\Pau(3),K_v\rangle$ is a finite $2$-group. Lemma~\ref{lem:tool-bridge} makes $\gamma_F\gamma_{a_v}e'$ unipotent for every $e'\in E_{a_v^\dagger Da_v}=\gamma_{a_v}^{-1}E_D\gamma_{a_v}$, that is, $\gamma_Fef$ is unipotent for all $e\in E_D$ and $f\in E_A$. As in Case~1, this is a finite check: $A$ is one of ten normal forms, $D$ is a Clifford conjugate of one of four gates, so it takes finitely many values modulo scalars, and only $\gamma_F\in\Sp(6,2)$ matters. By \textup{(\ref{cert:block-pair-domains})}(b), $\mathcal B(A,B)$ holds. \qedhere
\end{itemize}
\end{proof}

\begin{lemma}\label{app:return-two}
Let $W\in\mathrm U(16)$ satisfy $WP_vW^\dagger\in\mathscr R_4$ for every $0\ne v\in V_4$. Suppose $W$ has Pauli pairs $(s_1,t_1)$ and $(s_2,t_2)$ with $t_1,t_2$ linearly independent and $\omega(t_1,t_2)=0$. Then $\mathscr L(W)\ne\varnothing$.
\end{lemma}
\begin{proof}
Since $\omega(t_1,t_2)=0$ and $\omega(s_1,s_2)=0$, there are $C_1,C_2\in\C_2(4)$ with $C_2P_{t_j}C_2^\dagger=Z_j$ and $C_1(\pm P_{s_j})C_1^\dagger=Z_j$ for $j=1,2$. Replacing $W$ by $C_1WC_2^\dagger$, which keeps the hypothesis, we get $WZ_jW^\dagger=Z_j$, so
\[
W=\bigoplus_{x\in\F_2^2}B_x,
\]
where $B_x$ acts on qubits $3,4$ in the joint eigenspace $Z_j=(-1)^{x_j}$. By Lemma~\ref{lem:tool-blocks}(1), it suffices to find a common GSC pair of the four blocks.

We first show that the blocks are pairwise compatible. For $x\ne y$, after conjugating $W$ by an affine permutation of the two control qubits, which keeps the hypothesis, $x$ and $y$ differ only in the second control bit. Then $\operatorname{diag}(B_x,B_y)$ is a block of $W$ with respect to $Z_1$, and it satisfies the hypotheses of Proposition~\ref{gauss-prop:n3-gsc}: its corrections are blocks of corrections $WP_vW^\dagger\in\mathscr R_4$ with $P_v$ commuting with $Z_1$, so they are Gaussian, generalized semi-Clifford by Lemma~\ref{lem:tool-blocks}(1), and of finite level by Lemma~\ref{lem:tool-blocks}(2). Proposition~\ref{gauss-prop:n3-gsc} and Lemma~\ref{lem:tool-blocks}(1) give
\[
\mathscr L(B_x)\cap\mathscr L(B_y)\ne\varnothing\qquad\text{for all }x,y\in\F_2^2.
\]
Each $B_x$ has finite level by Lemma~\ref{lem:tool-blocks}(2), so by Lemma~\ref{lem:tool-unique}(2) the input Lagrangians of its GSC pairs are exactly the Lagrangians inside $\mathsf P(B_x)$. Hence $\dim\mathsf P(B_x)\in\{2,3,4\}$, with $4$ exactly when $B_x$ is Clifford. We split on these dimensions. Case~1a uses \tool{U}, Cases~1b and~1c combine the Galois automorphism with \tool{U} on the three Lagrangians containing a fixed label, and Case~2 is a special case of Theorem~\ref{lem:fib-sharp-small}.

\begin{inlinefigure}
\begin{tikzpicture}
\node[pfnode] (s) at (0,0) {$B_x$, $x\in\F_2^2$};
\node[pfnode] (c1) at (4,1) {Case 1};
\node[pfnode] (c2) at (4,-1.6) {Case 2};
\node[pfnode] (c1a) at (8.4,2.2) {Case 1a};
\node[pfnode] (c1b) at (8.4,1) {Case 1b};
\node[pfnode] (c1c) at (8.4,-0.2) {Case 1c};
\node[pfgoal] (goal) at (13,0.4) {$\bigcap_x\mathscr L(B_x)\ne\varnothing$};
\draw[pfedge] (s) -- node[pflab,sloped,above] {$\exists x:\ B_x\notin\C_2(2)$} (c1);
\draw[pfedge] (s) -- node[pflab,sloped,below] {$\forall x:\ B_x\in\C_2(2)$} (c2);
\draw[pfedge] (c1) -- node[pflab,sloped,above] {$\exists x:\ \dim\mathsf P(B_x)=2$} (c1a);
\draw[pfedge] (c1) -- node[pflab,above,pos=.62] {$q_x$ not all equal} (c1b);
\draw[pfedge] (c1) -- node[pflab,sloped,below] {all $q_x$ equal} (c1c);
\draw[pfedge] (c1a) -- (goal);
\draw[pfedge] (c1b) -- (goal);
\draw[pfedge] (c1c) -- (goal);
\draw[pfedge] (c2) -| (goal);
\end{tikzpicture}
\caption{Case analysis in the proof of Lemma~\ref{app:return-two}.}
\end{inlinefigure}

\begin{itemize}
\item \step{Case 1:}{some $B_x$ is not Clifford.} If $\dim\mathsf P(B_x)=2$ for some $x$, we are in Case~1a discussed below. Otherwise, we first find a label $p$ common to all blocks. As in the proof of Lemma~\ref{app:return-block}(3), Proposition~\ref{lem:tool-galois}(3) gives $\sigma(W)\doteq WP_a$ with $P_a$ commuting with $Z_1,Z_2$. Hence there is one label $p\in V_2$ with $\sigma(B_x)\doteq B_xP_p$ for all $x$, and by Proposition~\ref{lem:tool-galois}(2), $p$ lies in the input Lagrangian of every GSC pair of every $B_x$.

If $p=0$, each $B_x$ is defined over $\Q(i)$ up to scalar, so by Proposition~\ref{lem:tool-monomial}(2) it is Clifford equivalent to a diagonal gate with entries in $\mu_4$, and after diagonal Cliffords to $I$ or $\mathrm{CS}$. A non-Clifford block would then have $\dim\mathsf P=2$, as $\mathrm{CS}$ does, which we have excluded. Hence $p\ne0$. Then no $B_x$ is Clifford, because $\sigma(C)\doteq C$ for every Clifford $C$ by Proposition~\ref{lem:tool-galois}(1). So $\dim\mathsf P(B_x)=3$, that is, $\mathsf P(B_x)=p_x^\perp$ for a nonzero $p_x$, and the input Lagrangians of $B_x$ are the three Lagrangians $\mathcal L_{p_x}=\{L=L^\perp:p_x\in L\}$. Since $p$ lies in all three, whose intersection is $\Span\{p_x\}$, $p_x=p$ for every $x$.

Let $B_xP_pB_x^\dagger=\pm P_{q_x}$, and let
\[
\phi_x:\mathcal L_p\to\mathcal L_{q_x},\qquad
B_x\A_LB_x^\dagger=\A_{\phi_x(L)},
\]
be the resulting bijection. Pairwise compatibility says that for all $x,y$ there is $L\in\mathcal L_p$ with $\phi_x(L)=\phi_y(L)$. Cases~1b and~1c split on whether the $q_x$ are all equal.

\item \step{Case 1a:}{$\dim\mathsf P(B_x)=2$ for some $x$.} By Lemma~\ref{lem:tool-unique}(2), $\mathscr L(B_x)$ consists of a single pair $(L,M)$. By pairwise compatibility, $(L,M)\in\mathscr L(B_y)$ for every $y$.

\item \step{Case 1b:}{the $q_x$ are not all equal.} Pairwise compatibility gives, for every pair with $q_x\ne q_y$, a Lagrangian containing both $q_x$ and $q_y$, so all $q_x$ are pairwise orthogonal and lie in one Lagrangian $M$. For $q_x\ne q_y$ the common output is $M$, so $\phi_x^{-1}(M)=\phi_y^{-1}(M)$. Comparing two blocks with the same $q$ through a block with a different one shows that all $\phi_x^{-1}(M)$ agree. This is a common GSC pair.

\item \step{Case 1c:}{all $q_x$ equal $q$.} Then $\pi_x=\phi_0^{-1}\phi_x$ is a permutation of the three-element set $\mathcal L_p$, and pairwise compatibility says that any two of $\pi_0=1,\pi_1,\pi_2,\pi_3$ agree at some point. By Lemma~\ref{lem:tool-one-qubit}(1), all $\pi_x$ fix a common $L$, and $(L,\phi_0(L))$ is a common GSC pair.

\item \step{Case 2:}{every $B_x$ is Clifford.} Then $W=\bigoplus_xB_x$ lies in $\mathcal F_2(2,2)$, of finite level since every correction lies in $\mathscr R_4$. Thus $W\in\GSC(4)$ by Theorem~\ref{lem:fib-sharp-small}. \qedhere
\end{itemize}
\end{proof}

\medskip

\begin{proof}[Proof of Theorem~\ref{thm:n4-gaussian} when $n=4$]
    The lemmas above combine as in Fig.~\ref{fig:F}. Let every $h_v$ be Gaussian and generalized semi-Clifford. By Lemma~\ref{app:finite-endpoint}, $\mathscr L(U)\ne\varnothing$ or $U$ has a Pauli pair, and in the latter case Lemma~\ref{app:return-block} gives $W=\operatorname{diag}(A,B)$. If $|X(K_q)|=1$ for some $q$, Lemma~\ref{app:return-single} applies, and otherwise Lemma~\ref{app:return-multiple} applies. Their alternative~(2) gives a Pauli pair commuting with and independent of $(z_1,z_1)$, so Lemma~\ref{app:return-two} applies, and alternative~(3) of Lemma~\ref{app:return-single} is settled by Lemma~\ref{app:return-balanced}. So $U\in\GSC(4)$ whenever every $h_v$ is Gaussian and generalized semi-Clifford. This is the hypothesis of Lemma~\ref{app:lem-ge-equiv} for $n=4$, so the first assertion of Theorem~\ref{gauss-thm:rigidity} holds for $n=4$. Now let every $h_v$ be Gaussian. By Lemma~\ref{app:lem-ge-equiv}, every $h_v$ is generalized semi-Clifford and $U\in\GSC(4)$. Then Proposition~\ref{lem:level-bound} with $m=2$ gives $U\in\C_6(4)$, and therefore $h_v\in\C_5(4)$ for every $v$.
\end{proof}

\subsection{A finite level reduction}\label{subsec:n4-finite-cutoffs}
We show that, to prove $\C_k(4)\subseteq\GSC(4)$ for every $k$, it suffices to check the single level $k=4290$.
\begin{proposition}\label{prop:n4-finite-cutoffs}
Suppose that $\C_k(4)\not\subseteq\GSC(4)$ for some $k$. Then there is a gate $U'\in\C_{4290}(4)\setminus\GSC(4)$ whose corrections are generalized semi-Clifford and with $\operatorname{cond}(U')\le285$. In particular, $\C_k(4)\subseteq\GSC(4)$ for every $k$ if and only if $\C_{4290}(4)\subseteq\GSC(4)$.
\end{proposition}

\begin{proof}
Let $e_1,\dots,e_8$ be the standard basis of $V_4$ and $D=\Z[1/2]/\Z$. A coordinate $\theta\in D$ stands for the root of unity $e^{2\pi i\theta}$, and $\zeta^{a\cdot\theta}=e^{2\pi i\,a\cdot\theta}$ for $a\in\Z^{128}$ and $\theta\in D^{128}$. An exact Clifford gate is a product of $H$, $S$ and $\CNOT$. Its entries lie in $\Z[\zeta_8,1/2]$, its determinant is a root of unity of $2$-power order, and every Clifford gate is a scalar multiple of one.
\begin{itemize}
\item \step{Step 1:}{a counterexample of least level has a finite description with $128$ phases.} Let $U\in\C(4)\setminus\GSC(4)$ have least level. Its corrections have lower level, so they are generalized semi-Clifford. By Proposition~\ref{lem:tool-monomial}, $h_{e_i}=\lambda_iC_{1i}\Pi_{\pi_i}D_iC_{2i}$ with exact Clifford gates $C_{1i},C_{2i}$, a permutation $\pi_i$, a diagonal $D_i$ with entries in $\mu_{2^s}$ for some $s$, and $\lambda_i\in\CC^\times$. As $h_{e_i}$ is a traceless Hermitian involution, $\det h_{e_i}=1$, so $\lambda_i^{16}$ is a root of unity of $2$-power order, and we absorb $\lambda_i$ into $D_i$. Thus $h_{e_i}=g_i(\theta_U)$ for some $\theta_U\in D^{128}$, where
\[
g_i(\theta)=C_{1i}\Pi_{\pi_i}\diag_{x\in\F_2^4}\bigl(e^{2\pi i\theta_{i,x}}\bigr)C_{2i}\qquad(1\le i\le8).
\]
For every $v\in V_4$, Proposition~\ref{lem:tool-monomial} also gives exact Clifford gates $C_v,C'_v$ and a permutation $\rho_v$ preserving a complete affine flag $F_v$ such that $C_v^{-1}h_vC_v'^{-1}$ is monomial with permutation $\rho_v$. Finally, $U\notin\GSC(4)$, and $h_{u+v}\doteq h_uh_v$, so by Lemma~\ref{lem:main-local-global}, for each of the $2{,}295$ Lagrangians $M$ (Lemma~\ref{lem:lagrangian-count}) there are $i$, a basis label $b$ of $M$ and $w\notin M$ such that the coefficient $c_M$ of $P_w$ in $h_{e_i}P_bh_{e_i}^\dagger$ is nonzero. Fix these discrete data, and for $\theta\in D^{128}$ consider the conditions:
\begin{itemize}
\item[(E1)] $g_i=g_i^\dagger$, and $g_ig_j=\pm g_jg_i$ with the sign of $P_{e_i}P_{e_j}=\pm P_{e_j}P_{e_i}$;
\item[(E2)] for every $v=\sum_iv_ie_i$, the entries of $C_v^{-1}w_vC_v'^{-1}$ off the permutation $\rho_v$ vanish, where $w_v=g_1^{v_1}\cdots g_8^{v_8}$;
\item[(N)] for every Lagrangian $M$, the coefficient $c_M$, computed from the $g_i$, is nonzero.
\end{itemize}
All of them hold at $\theta_U$, since $w_v(\theta_U)\doteq h_v$.

\item \step{Step 2:}{if $\theta\in D^{128}$ satisfies \textup{(E1)}, \textup{(E2)} and \textup{(N)}, and its coordinates have order dividing $2^N$ with $N\ge3$, then there is $U'\in\C_{15N+15}(4)\setminus\GSC(4)$ with generalized semi-Clifford corrections and $\operatorname{cond}(U')\le N$.} The $g_i$'s are unitary, and by (E1) they are Hermitian and satisfy the commutation relations of the $P_{e_i}$. A $16$-dimensional representation of these relations is unique up to unitary equivalence, so $g_i=U'P_{e_i}U'^\dagger$ for a unitary $U'$, and $h'_v=U'P_vU'^\dagger\doteq w_v$. By (E2), $C_v^{-1}h'_vC_v'^{-1}$ is monomial with permutation $\rho_v$. Its nonzero entries $\alpha$ satisfy $\alpha\bar\alpha=1$ and lie in $\mathcal O_{K_N}[1/2]$. The prime $2$ is totally ramified in $K_N$, and complex conjugation fixes the prime above it, so $\alpha$ is an algebraic integer whose conjugates all have absolute value one. By Kronecker's theorem $\alpha$ is a root of unity in $K_N$, so $\alpha\in\mu_{2^N}$. Hence $[C_v^{-1}h'_vC_v'^{-1}]$ lies in the group $\mathcal M_{F_v}(N)$ of the proof of Proposition~\ref{app:lem-R_n}, which has order $2^{15}\cdot2^{15N}$ by Lemma~\ref{lem:tool-flags}(1). As shown there, a finite $2$-group containing $\overline{\Pau}(4)$ lies in $[\C_c(4)]$, where $c$ is its nilpotency class, here at most $15N+14$. Therefore every $h'_v$ lies in $\C_{15N+14}(4)$ and is generalized semi-Clifford, and $U'\in\C_{15N+15}(4)$. By (N), no $\A_M$ is normalized by all $g_i$, so $U'\notin\GSC(4)$ by Lemma~\ref{lem:main-local-global}. Finally, the solutions of the linear system $XP_{e_i}=g_iX$, $1\le i\le8$, which has coefficients in $K_N$, are the multiples of $U'$, so $\operatorname{cond}(U')\le N$.

\item \step{Step 3:}{on a fixed cancellation pattern the conditions are linear in $\theta$.} Each entry or coefficient in (E1), (E2), and (N) is a sum $\sum_tc_t\zeta^{a_t\cdot\theta}$ with distinct $a_t\in\Z^{128}$ and $c_t\in K_3$. Roots of unity of $2$-power order in distinct cosets of $\mu_8$ are linearly independent over $K_3$, since $\zeta_{2^m}^j$, $0\le j<2^{m-3}$, form a basis of $K_m$ over $K_3$. Group the terms of each sum into blocks by the coset of $\zeta^{a_t\cdot\theta_U}$, fix a term $t_0$ in each block, and let $\zeta^{(a_t-a_{t_0})\cdot\theta_U}=\zeta_8^{\ell_t}$. Let $S\subseteq D^{128}$ be the set of $\theta$ with
\[
(a_t-a_{t_0})\cdot\theta=\ell_t/8\ \text{ for all blocks},\qquad
8(a_t-a_{t'})\cdot\theta\ne0\ \text{ for $t,t'$ in distinct blocks of a sum in (N)}.
\]
For $\theta\in S$, each block sum equals the block sum at $\theta_U$ times a root of unity. So the sums in (E1) and (E2) still vanish, and each sum in (N) still has a nonzero block in a coset of its own, so it is nonzero. Thus every point of $S$ satisfies (E1), (E2) and (N), and $\theta_U\in S$.

\item \step{Step 4:}{the equations have integer rows of norm at most $4$, and there are fewer than $2^{26}$ disequalities, with rows of norm at most $16\sqrt2$.} The entries of $g_i$ are linear in the $16$ variables $e^{2\pi i\theta_{i,x}}$, and those of $g_i^\dagger$ in their inverses. In (E1), differences of exponents have norm at most $2$. In (E2), $w_v$ contains each $g_i$ at most once, so a difference has norm at most $\sqrt2$ in each of the eight groups of variables, and at most $4$ in total. In (N), $g_iP_bg_i^\dagger$ has the $241$ exponents $0$ and $e_{i,x}-e_{i,y}$, $x\ne y$, so differences have norm at most $2\sqrt2$, and the $2{,}295$ sums give at most $2{,}295\binom{241}{2}<2^{26}$ disequalities.

\item \step{Step 5:}{$S$ contains a point whose coordinates have order dividing $2^{285}$.} Write the equalities as $A\theta=b$ with $A$ an integer matrix of rank $r\le128$ and $b$ with entries in $\frac18\Z/\Z$. Let $A=P\operatorname{diag}(d_1,\dots,d_r,0,\dots,0)Q$ be its Smith normal form, with $P,Q$ unimodular, and let $y=Q\theta$. The equalities prescribe $d_jy_j$ for $j\le r$ and leave $y_{r+1},\dots,y_{128}$ free. At $\theta_U$, each $y_j$ with $j\le r$ has order dividing $2^{3+\nu_2(d_j)}$, and $\sum_j\nu_2(d_j)\le\log_2\delta_r(A)\le r\log_24\le256$ by Hadamard's inequality, where $\delta_j$ is the gcd of the nonzero $j\times j$ minors. Keep these coordinates of $\theta_U$. Every disequality that is constant in the free coordinates then still holds. Another disequality $\lambda\cdot\theta\ne0$ becomes $q\cdot(y_{r+1},\dots,y_{128})\ne c$ for an integer vector $q\ne0$, so $r\le127$, and
\[
\nu_2(\gcd q)=\nu_2\,\delta_{r+1}\!\begin{pmatrix}A\\\lambda\end{pmatrix}-\nu_2\,\delta_r(A)\le\log_2\bigl(16\sqrt2\cdot4^{r}\bigr)<259,
\]
so $\nu_2(\gcd q)\le258$. On the grid of free coordinates of order dividing $2^{285}$, such a disequality fails on a proportion at most $2^{258-285}=2^{-27}$ of the points. There are fewer than $2^{26}$ of them, so some grid point satisfies all, and the resulting $\theta=Q^{-1}y$ has coordinates of order dividing $2^{\max\{259,285\}}=2^{285}$.
\end{itemize}
Steps~2 and~5 with $N=285$ give $U'\in\C_{4290}(4)\setminus\GSC(4)$ with generalized semi-Clifford corrections and $\operatorname{cond}(U')\le285$. In the last assertion the forward direction is trivial, and the converse follows from the first assertion.
\end{proof}

\section{Beyond generalized semi-Clifford}\label{sec:beyond-gsc}

In this section, we prove Lemma~\ref{lem:control-form} and Proposition~\ref{clifford-fiber} and construct gates outside $\GSC$ of every exact level $k\ge5$ on every $n\ge5$ qubits. We use the sets $\mathcal F(n_c,n_t)$ and $\mathcal F_j(n_c,n_t)$ defined in Section~\ref{sec:conjectures} with $n=n_c+n_t$. For $U\in\mathcal F_2(n_c,n_t)$, Lemma~\ref{lem:basic-operations}(1) allows us to drop the Clifford factors $C_1,C_2$ and write
\begin{equation}\label{eq:fib-form}
 U=\sum_{x\in\F_2^{n_c}}\ket{\pi(x)}\bra x\otimes V_x
\end{equation}
with a permutation $\pi$ of $\F_2^{n_c}$ and Clifford gates $V_x$. Let $S_x=\gamma_{V_x}\in\Sp(V_{n_t})$ be the symplectic action of $V_x$, so that $V_xP_vV_x^\dagger\doteq P_{S_xv}$.

\begin{proof}[Proof of Lemma~\ref{lem:control-form}]
Let $R\le V_n$ be the span of the labels of $Z_1,\dots,Z_{n_c}$, and for $x\in\F_2^{n_c}$ let $Q_x=\ket x\bra x\otimes I$. The $Q_x$ are pairwise orthogonal projections with sum $I$, and they span $\A_R$. Hence every projection in $\A_R$ is $\sum_{y\in Y}Q_y$ for some $Y\subseteq\F_2^{n_c}$, and it has rank $2^{n_t}|Y|$. Recall that $\gamma_C$ is symplectic and $C\A_SC^\dagger=\A_{\gamma_CS}$ for $C\in\C_2(n)$ and $S\le V_n$.

We first show that every $U\in\mathcal F(n_c,n_t)$ satisfies the condition. Write $U=C_1MC_2$ with $C_1,C_2\in\C_2(n)$ and $M=\sum_x\ket{\pi(x)}\bra x\otimes V_x$. Then $MQ_xM^\dagger=Q_{\pi(x)}$ for every $x$, so $M\A_RM^\dagger=\A_R$. The subspaces $S=\gamma_{C_2}^{-1}R$ and $T=\gamma_{C_1}R$ are isotropic of dimension $n_c$, and
\[
U\A_SU^\dagger=C_1M\A_RM^\dagger C_1^\dagger=C_1\A_RC_1^\dagger=\A_T.
\]

Conversely, let $U\in\C(n)$ satisfy $U\A_SU^\dagger=\A_T$ for isotropic subspaces $S,T\le V_n$ of dimension $n_c$. A basis of an isotropic subspace extends to a symplectic basis of $V_n$, so $\Sp(V_n)$ acts transitively on the isotropic subspaces of dimension $n_c$. Since every binary symplectic map is induced by a Clifford gate, there are $C,C'\in\C_2(n)$ with $\gamma_CS=R$ and $\gamma_{C'}T=R$. Then $M=C'UC^\dagger$ satisfies $M\A_RM^\dagger=\A_R$. For each $x$, $MQ_xM^\dagger$ is a projection in $\A_R$ of rank $2^{n_t}$, so $MQ_xM^\dagger=Q_{\pi(x)}$ for some $\pi(x)\in\F_2^{n_c}$. For $x\ne x'$, the projections $MQ_xM^\dagger$ and $MQ_{x'}M^\dagger$ are orthogonal, so $\pi$ is a permutation of $\F_2^{n_c}$. Hence $MQ_x=Q_{\pi(x)}M$, and the unitary $M$ maps the range $\ket x\otimes\CC^{2^{n_t}}$ of $Q_x$ onto the range $\ket{\pi(x)}\otimes\CC^{2^{n_t}}$ of $Q_{\pi(x)}$. Therefore
\[
M=\sum_xQ_{\pi(x)}MQ_x=\sum_x\ket{\pi(x)}\bra x\otimes V_x,\qquad V_x=(\bra{\pi(x)}\otimes I)\,M\,(\ket x\otimes I)\in\mathrm U(2^{n_t}),
\]
and $U=C'^\dagger MC$ lies in $\mathcal F(n_c,n_t)$.
\end{proof}

\subsection{Proof of Proposition~\ref{clifford-fiber}}\label{subsec:clifford-fibers}

\begin{theorem}\label{lem:fib-sharp-small}
Let $U=C_1\bigl(\sum_x\ket{\pi(x)}\bra x\otimes V_x\bigr)C_2\in\mathcal F_2(n_c,n_t)$, and let $S_x$ be the symplectic action of $V_x$.
\begin{enumerate}
\item $U\in\GSC(n)$ if and only if there are Lagrangians $L,L'$ of $V_{n_t}$ with $S_xL=L'$ for every $x\in\F_2^{n_c}$.
\item If $n_c\le1$ or $n_t\le2$, then $U\in\GSC(n)$.
\item If $n_c\ge2$ and $n_t\ge3$, then $\mathcal F_2(n_c,n_t)\not\subseteq\mathsf{GSC}(n)$.
\end{enumerate}
In particular, (2) and (3) imply Proposition~\ref{clifford-fiber}.
\end{theorem}

\begin{proof}
By Lemma~\ref{lem:basic-operations}(1), we may assume $C_1=C_2=I$, which changes neither the blocks nor the level.

\textit{Part (1).} Let $R\le V_n$ be spanned by the $Z$-labels of the $n_c$ control qubits. For a Lagrangian $L$ of $V_{n_t}$, the space $R\oplus L$ is a Lagrangian of $V_n$, and $\A_{R\oplus L}$ is spanned by the operators $\ket x\bra x\otimes P_v$ with $x\in\F_2^{n_c}$ and $v\in L$. Since
\[
 U\bigl(\ket x\bra x\otimes P_v\bigr)U^\dagger=\ket{\pi(x)}\bra{\pi(x)}\otimes V_xP_vV_x^\dagger\doteq\ket{\pi(x)}\bra{\pi(x)}\otimes P_{S_xv},
\]
the condition $S_xL=L'$ for every $x$ gives $U\A_{R\oplus L}U^\dagger\subseteq\A_{R\oplus L'}$, with equality because both algebras have dimension $2^n$. Hence $U\in\GSC(n)$ by Lemma~\ref{lem:lagrangian}.

Conversely, let $(S,T)\in\mathscr L(U)$. Since $U(\ket x\bra x\otimes I)U^\dagger=\ket{\pi(x)}\bra{\pi(x)}\otimes I$, we have $U\A_RU^\dagger=\A_R$, and Lemma~\ref{lem:root-witness}, applied to $(R,R)$ and $(S,T)$, gives a GSC pair of $U$ whose two Lagrangians contain $R$. A Lagrangian $N$ with $R\le N\le R^\perp=R\oplus V_{n_t}$ equals $R\oplus(N\cap V_{n_t})$, as in the proof of Lemma~\ref{lem:tool-blocks}, so this pair is $(R\oplus L,R\oplus L')$ for Lagrangians $L,L'$ of $V_{n_t}$. By the display above, $\ket{\pi(x)}\bra{\pi(x)}\otimes P_{S_xv}$ lies in $\A_{R\oplus L'}$ for every $x$ and $v\in L$, so $S_xL\subseteq L'$, with equality by dimension.

\textit{Parts (2) and (3).} A \emph{Sylow coset} is a right coset $Pg$ of a Sylow $2$-subgroup $P$ of $\Sp(4,2)\cong S_6$. We split into three cases according to $(n_c,n_t)$.

\begin{itemize}
\item \step{Case 1:}{$n_c\leq1$ or $n_t=0$.} If $n_t=0$, take $L=L'=0$, and if $n_c=0$, take any $L$ and $L'=S_0L$. If $n_c=1$, the gate $(X^{\pi(0)}\otimes I)\,U\,(I\otimes V_0^\dagger)=\operatorname{diag}(I,V_1V_0^\dagger)$ has finite level $k$. By Proposition~\ref{lem:controlled-clifford}, $(V_1V_0^\dagger)^{2^{k-2}}\in\mu_4\mathcal P(n_t)$, so $S_1S_0^{-1}$ has $2$-power order. The cyclic $2$-group it generates fixes a Lagrangian $L'$, since the number of Lagrangians is odd (Lemma~\ref{lem:lagrangian-count}), and $L=S_0^{-1}L'$ gives $S_0L=S_1L=L'$. Part~(1) applies.

\item \step{Case 2:}{$n_c\geq2$ and $n_t\in\{1,2\}$.} For $n_t=1$, tensoring with the identity on one additional target qubit gives $U\otimes I\in\mathcal F_2(n_c,2)$, which is generalized semi-Clifford exactly when $U$ is, by Lemma~\ref{lem:basic-operations}(2). We may therefore take $n_t=2$. We show that the $S_x$ lie in one Sylow coset $Pg$. Then $P$ fixes a Lagrangian $L'$, since the number of Lagrangians is odd, and $L=g^{-1}L'$ gives $S_xL=L'$ for every $x$, so Part~(1) applies.

The proof is by induction on the level for all gates of the form~\eqref{eq:fib-form} with the same $n_c$ and $n_t=2$ at once. At level one, the blocks are Paulis, so every $S_x=I$. In the induction step, we first show that each set $\mathcal T_a=\{S_{x+a}S_x^{-1}:x\in\F_2^{n_c}\}$ generates a $2$-group. The correction
\[
 U(X^a\otimes I)U^\dagger=\sum_x\ket{\pi(x+a)}\bra{\pi(x)}\otimes V_{x+a}V_x^\dagger
\]
is again of the form~\eqref{eq:fib-form}, with the permutation $x\mapsto\pi(\pi^{-1}(x)+a)$ and lower level, and the actions of its blocks form $\mathcal T_a$. By induction, $\mathcal T_a$ lies in one Sylow coset, so $H=\langle s\tau^{-1}:s\in\mathcal T_a\rangle$ is a $2$-group for fixed $\tau\in\mathcal T_a$. Since $\mathcal T_a$ is inverse-closed, $\tau^{-2}\in H$ and
\[
 \tau(s\tau^{-1})\tau^{-1}=(s^{-1}\tau^{-1})^{-1}\tau^{-2}\in H.
\]
Thus $\tau$ normalizes $H$ and $\langle\mathcal T_a\rangle=H\langle\tau\rangle$ is a $2$-group.

Next let $\{x,x+a,x+b,x+a+b\}$ be an affine plane, and let $\alpha=S_{x+a}S_x^{-1}$, $\beta=S_{x+b}S_x^{-1}$, $\gamma=S_{x+a+b}S_x^{-1}$. Then $\gamma\beta^{-1},\alpha\in\mathcal T_a$, $\gamma\alpha^{-1},\beta\in\mathcal T_b$ and $\beta\alpha^{-1},\gamma\in\mathcal T_{a+b}$, so $\langle\alpha,\gamma\beta^{-1}\rangle$, $\langle\beta,\gamma\alpha^{-1}\rangle$ and $\langle\gamma,\beta\alpha^{-1}\rangle$ lie in the $2$-groups $\langle\mathcal T_a\rangle$, $\langle\mathcal T_b\rangle$ and $\langle\mathcal T_{a+b}\rangle$. For every $n_c$, $\alpha,\beta,\gamma\in\Sp(4,2)\cong S_6$, a group of order $720$, in which a subgroup is a $2$-group exactly when it lies in one of the $45$ Sylow $2$-subgroups. Checking all triples gives the following.

\begin{finitefact}{cert:s6-two-group-criterion}{A $2$-group criterion in $S_6$}{Theorem~\ref{lem:fib-sharp-small}}
Let $a,b,c\in\Sp(4,2)\cong S_6$. If $\langle a,cb^{-1}\rangle$, $\langle b,ca^{-1}\rangle$ and $\langle c,ba^{-1}\rangle$ are $2$-groups, then so is $\langle a,b,c\rangle$.
\tcblower
\textbf{Code.} \certlink{s6-two-group-criterion}
\end{finitefact}
By \textup{(\ref{cert:s6-two-group-criterion})}, $\langle\alpha,\beta,\gamma\rangle$ is a $2$-group, so the four actions $S_x,\alpha S_x,\beta S_x,\gamma S_x$ on the plane lie in one Sylow coset.

Finally, every three points of $\F_2^{n_c}$ lie in an affine plane, so every subset of $F=\{S_x:x\in\F_2^{n_c}\}$ with at most three elements lies in a Sylow coset. For every $n_c$, $F\subseteq\Sp(4,2)\cong S_6$. Right multiplication preserves Sylow cosets, so we may assume $1\in F$. Then $F\setminus\{1\}$ is a clique of the graph on the $255$ nonidentity elements of $2$-power order, joined when they lie in a common Sylow $2$-subgroup, so it lies in one of the finitely many maximal cliques, and listing them gives the following.

\begin{finitefact}{cert:s6-sylow-cosets-and-three-element-subsets}{Sylow cosets in $S_6$ and three-element subsets}{Theorem~\ref{lem:fib-sharp-small}}
Let $F\subseteq\Sp(4,2)\cong S_6$. If every subset of $F$ with at most three elements lies in a right coset $Pg$ of a Sylow $2$-subgroup $P$, then so does $F$.
\tcblower
\textbf{Code.} \certlink{s6-sylow-cosets-and-three-element-subsets}
\end{finitefact}
By \textup{(\ref{cert:s6-sylow-cosets-and-three-element-subsets})}, $F$ lies in one Sylow coset.

\item \step{Case 3:}{$n_c\geq2$ and $n_t\geq3$.} By Theorem~\ref{thm:exact-cells} with $k=n=5$, the gate $U_{E_6}\in\mathcal F_2(2,3)$ is not generalized semi-Clifford. Letting it act as the identity on $n_c-2$ further control qubits and $n_t-3$ further target qubits gives a gate of $\mathcal F_2(n_c,n_t)$ for every $n_c\ge2$ and $n_t\ge3$, which is not generalized semi-Clifford by Lemma~\ref{lem:basic-operations}(2). \qedhere
\end{itemize}
\end{proof}

\subsection{Constructing non-generalized semi-Clifford gates}\label{subsec:orthogonal-fibers}

For a Hermitian Pauli $P_p$ let
\begin{equation}\label{eq:fib-root-lift}
 R_p=\frac{I+P_p}{2}+i\,\frac{I-P_p}{2},
\end{equation}
a Clifford gate with entries in $\mathbb Z[i,1/2]$ whose symplectic action is the transvection $t_p(v)=v+\omega(v,p)p$. On three qubits take
\[
 (P_{p_0},\dots,P_{p_5})=(Z_0,\ X_0,\ Z_0Z_1,\ X_1,\ Z_1Z_2,\ X_1X_2),
\]
in which $P_{p_i}$ and $P_{p_j}$ anticommute exactly for $ij\in\{01,12,23,34,25\}$, the edges of the Dynkin diagram of $E_6$. Let $\widehat A=R_{p_1}R_{p_2}R_{p_3}R_{p_4}R_{p_5}$ and $\widehat B=R_{p_0}R_{p_2}$, with symplectic actions $A=t_{p_1}\cdots t_{p_5}$ and $B=t_{p_0}t_{p_2}$. For a root of unity $\xi$, define the five-qubit gate with control qubits $c,t$
\begin{equation}\label{eq:fib-e6-family}
 U_\xi=\sum_{c,t\in\F_2}|ct\rangle\langle ct|\otimes(\xi\widehat A)^{c}\,\widehat B^{\,t},
 \qquad U_{E_6}=U_1.
\end{equation}

\begin{lemma}\label{lem:fib-phase}
Let $f\in\C_a(m)$ with $a\ge1$ satisfy $fZ_cf^\dagger=\pm Z_c$ for a qubit $c$, $\lambda$ a root of unity of order $2^\ell$, and $\Delta_\lambda= \operatorname{diag}(\lambda,1)$ act on the qubit $c$. Then $f\Delta_\lambda\in\C_{\max(a,\ell)}(m)$.
\end{lemma}
\begin{proof}
We use induction on $N=\max(a,\ell)$. If $N\le2$, then $f$ and $\Delta_\lambda$ lie in the group $\C_N(m)$, and so does $f\Delta_\lambda$. Let $N\ge3$ and let $Q$ be a Hermitian Pauli. If $Q$ commutes with $Z_c$, then so does $\Delta_\lambda$, and the correction of $f\Delta_\lambda$ at $Q$ is $fQf^\dagger\in\C_{N-1}(m)$. Otherwise, $Q$ swaps the two eigenspaces of $Z_c$, so $Q^\dagger\Delta_\lambda Q=\operatorname{diag}(1,\lambda)$ and
\[
 \Delta_\lambda Q\Delta_\lambda^\dagger=Q\operatorname{diag}(\bar\lambda,\lambda)=\lambda\,Q\Delta_{\lambda^{-2}},\qquad
 (f\Delta_\lambda)Q(f\Delta_\lambda)^\dagger=\lambda\,(fQf^\dagger)(f\Delta_{\lambda^{-2}}f^\dagger).
\]
Since $f$ fixes or swaps the eigenspaces of $Z_c$, $f\Delta_{\lambda^{-2}}f^\dagger$ is $\Delta_{\lambda^{-2}}$ or $\lambda^{-2}\Delta_{\lambda^2}$. So the correction is a scalar multiple of $g\Delta_\mu$, where $g=fQf^\dagger\in\C_{\max(a-1,1)}(m)$ satisfies $gZ_cg^\dagger=\pm Z_c$ and $\mu$ has order $2^{\max(\ell-1,0)}$. By induction, $g\Delta_\mu\in\C_{N-1}(m)$, and as $N-1\ge2$ the scalar does not matter. Hence $f\Delta_\lambda\in\C_N(m)$.
\end{proof}

\begin{theorem}\label{thm:exact-cells}
For every $k\ge5$ and $n\ge5$, there is a gate in $\mathcal F_2(2,n-2)\cap\C_k(n)$ that is neither in $\C_{k-1}(n)$ nor generalized semi-Clifford.
\end{theorem}
\begin{proof}
Let $\xi$ be chosen so that $\xi=1$ if $k=5$, and $\xi$ has order $2^k$ if $k\ge6$. We show that $U_\xi\in\C_k(5)\setminus\C_{k-1}(5)$ and $U_\xi\notin\GSC(5)$. Then $U_\xi\otimes I$, with $n-5$ further target qubits, proves the theorem: it lies in $\mathcal F_2(2,n-2)\cap\C_k(n)$ and is not generalized semi-Clifford by Lemma~\ref{lem:basic-operations}(2), and it is not in $\C_{k-1}(n)$ by the following fact, which we use twice. If $V\otimes I\in\C_j$, then $V\in\C_j$, by induction on $j$ from $(V\otimes I)(P\otimes I)(V\otimes I)^\dagger=VPV^\dagger\otimes I$; here the identity may act on any of the qubits.
\begin{itemize}
\item \step{Step 1:}{$U_{E_6}\in\C_5(5)$.} $U_{E_6}$ is an explicit $32\times32$ matrix with entries in $\mathbb Z[i,1/2]$, and membership in $\C_5(5)$ is decided by finitely many conjugations by Paulis and Clifford tests, recorded in the following fact.

\begin{finitefact}{cert:e6-gate-level-five}{The gate $U_{E_6}$ lies in $\C_5(5)$}{Theorem~\ref{thm:exact-cells}}
The gate $U_{E_6}$ of~\eqref{eq:fib-e6-family} lies in $\C_5(5)$.
\tcblower
\textbf{Code.} \certlink{e6-gate-level-five}
\end{finitefact}
This is \textup{(\ref{cert:e6-gate-level-five})}.

\item \step{Step 2:}{$U_\xi\in\C_k(5)$.} Since $U_\xi=U_{E_6}\operatorname{diag}(1,\xi)_c=\xi\,U_{E_6}\Delta_{\bar\xi}$ and $U_{E_6}$ commutes with $Z_c$, Lemma~\ref{lem:fib-phase} with $a=5$ gives $U_\xi\in\C_{\max(5,k)}(5)=\C_k(5)$.

\item \step{Step 3:}{$U_\xi\notin\C_{k-1}(5)$.} The correction at $X_c$ is
\[
 U_\xi X_cU_\xi^\dagger=\sum_{t\in\F_2}\Bigl(|1t\rangle\langle0t|\otimes\xi\widehat A+|0t\rangle\langle1t|\otimes(\xi\widehat A)^{-1}\Bigr).
\]
It acts as the identity on $t$, and on $c$ and the targets as $(X\otimes\xi\widehat A)\operatorname{diag}\bigl(I,(\xi\widehat A)^{-2}\bigr)$. If $U_\xi\in\C_{k-1}(5)$, then this gate lies in $\C_{k-2}(4)$, so $\operatorname{diag}(I,(\xi\widehat A)^{-2})\in\C_{k-2}(4)$ and $(\xi\widehat A)^{2^{k-3}}\in\mu_4\mathcal P(3)$ by Proposition~\ref{lem:controlled-clifford}. The matrices $A$ and $B$ are explicit $6\times6$ products of transvections over $\F_2$, and computing with them gives the following.

\begin{finitefact}{cert:e6-root-words}{Root words for $E_6$}{Theorem~\ref{thm:exact-cells}}
Let $A$ and $B$ be as in Appendix~\ref{subsec:orthogonal-fibers}. Then $A$ has order $8$ and $\chi_{A^2B}(z)=z^6+z^3+1$.
\tcblower
\textbf{Code.} \certlink{e6-root-words}
\end{finitefact}
By \textup{(\ref{cert:e6-root-words})}, $A$ has order $8$. For $k=5$, this contradicts $\widehat A^4\in\mu_4\mathcal P(3)$. For $k\ge6$, $\widehat A^8$ is a Pauli up to a unit-modulus element of $\mathbb Z[i,1/2]$, which lies in $\mu_4$. Hence $\widehat A^{2^{k-3}}\in\mu_4\mathcal P(3)$, and then $\xi^{2^{k-3}}\in\mu_4$, although it has order $8$.

\item \step{Step 4:}{$U_\xi\notin\GSC(5)$.} The blocks of $U_\xi$ have the symplectic actions $A^cB^t$. If $U_\xi$ were generalized semi-Clifford, Theorem~\ref{lem:fib-sharp-small}(1) would give Lagrangians $L,L'$ of $V_3$ with $A^cB^tL=L'$ for all $c,t$. Then $c=t=0$ gives $L'=L$, and $AL=BL=L$. But by \textup{(\ref{cert:e6-root-words})}, $A^2B$ has characteristic polynomial $z^6+z^3+1$, the ninth cyclotomic polynomial, which is irreducible over $\F_2$ because $2$ has order $6$ modulo $9$. The characteristic polynomial of $A^2B$ on a subspace invariant under $A$ and $B$ divides it, so such a subspace is $0$ or $V_3$.
\end{itemize}
\end{proof}

\end{document}